\documentclass[10pt,aps,pra,amsmath,amssymb,amsfonts,superscriptaddress,twocolumn]{revtex4-2}
\usepackage{graphicx} 
\usepackage{appendix}
\usepackage{amsfonts,amssymb}
\usepackage{physics}
\usepackage{tikz}
\usepackage{dsfont}
\usepackage{quantikz}
\usepackage{amsthm}

\newcommand{\sumint}{%
  \mathop{%
    \vcenter{%
      \offinterlineskip
      \halign{%
        \hfil##\hfil\cr
        $\sum$\cr
        \noalign{\kern-1.6ex}
        $\int$\cr
      }%
    }%
  }%
}

\newtheorem{theorem}{Theorem}[section]  
\newtheorem{lemma}[theorem]{Lemma}

\theoremstyle{definition}
\newtheorem{definition}[theorem]{Definition}

\begin{document}

\title{Resource-optimal quantum mode parameter estimation with multimode Gaussian states }

\author{Maximilian Reichert}
\email{maximilian.reichert@ehu.eus}
\affiliation{Department of Physical Chemistry, University of the Basque Country UPV/EHU, Apartado 644, 48080 Bilbao, Spain}
\affiliation{EHU Quantum Center, University of the Basque Country UPV/EHU, Apartado 644, 48080 Bilbao, Spain} 
\author{Mikel Sanz}
\affiliation{Department of Physical Chemistry, University of the Basque Country UPV/EHU, Apartado 644, 48080 Bilbao, Spain}
\affiliation{EHU Quantum Center, University of the Basque Country UPV/EHU, Apartado 644, 48080 Bilbao, Spain}
\affiliation{Basque Center for Applied Mathematics (BCAM), Alameda de Mazarredo 14, 48009 Bilbao, Spain}
\author{Nicolas Fabre}
\email{nicolas.fabre@telecom-paris.fr}
\affiliation{Telecom Paris, Institut Polytechnique de Paris, 19 Place Marguerite Perey, 91120 Palaiseau, France}

\date{\today}

\begin{abstract}
Quantum mode parameter estimation addresses the problem of determining parameters that govern the shape of electromagnetic modes occupied by a quantum state of radiation. Canonical examples — estimating time delays and frequency shifts — underpin technologies such as radar, lidar, and optical clocks. A comprehensive framework for this class of problems was recently established, revealing that a broad family of quantum state designs can attain the Heisenberg limit and thereby surpass the precision of any classical strategy. This naturally raises the fundamental question: among all quantum-enhanced strategies, which is truly optimal? Answering this question requires identifying the physically meaningful resources that govern a given mode parameter estimation task, so that quantum states can be compared on equal footing. We show that these resources are intimately connected to the eigenmode basis of the generator of the relevant mode transformation. For time-shift estimation, whose generator is diagonal in the frequency domain, the pertinent resources are the mean frequency and bandwidth of the quantum state; analogous quantities emerge for other transformations. Crucially, our framework unifies two perspectives that have historically been treated separately: the particle-number (photon-counting) aspect and the mode-structure perspective of quantum light, providing a single coherent picture of quantum-enhanced sensing with multimode radiation. Within this unified framework, we derive a tight upper bound on the quantum Fisher information for Gaussian states, expressed in terms of these natural resources, and analytically identify the optimal Gaussian states saturating it. Strikingly, these optimal states take a particularly transparent form when expressed in the eigenbasis of the generator, which is a structural simplicity that reflects the deep connection between the geometry of the mode transformation and the architecture of the optimal probe. We further demonstrate that multimode homodyne detection constitutes the optimal measurement, achieving the quantum Fisher information bound and thus, completing the end-to-end characterization of optimal quantum metrology strategies for mode parameter estimation.

\end{abstract}

\maketitle

\section{Introduction}
Quantum light carrying information of interest -- whether in coherent, squeezed, thermal, or Fock states -- is characterized not only by its particle-number statistics but also by the specific electromagnetic modes involved. When used as a probe for parameter estimation, both degrees of freedom jointly determine the achievable precision. Modes are orthonormal solutions of Maxwell's equations \cite{fabre2020modes}, and the parameters that characterize their shape, such as central frequency, bandwidth, temporal duration, or spatial profile, are called mode parameters. Canonical examples arise throughout science and technology: the time delay of a returned pulse in lidar or radar, the frequency shift induced by the Doppler effect, and the transverse displacement of a beam in quantum microscopy. The natural framework for comparing classical and quantum strategies in such tasks is quantum metrology \cite{1976297,PhysRevD.23.1693}, which establishes the ultimate precision limits of any parameter estimation protocol regardless of the measurement performed. These limits are set by the quantum Cramér–Rao bound (QCRB), a fundamental lower bound on the variance of any unbiased estimator expressed in terms of the quantum Fisher information (QFI) of the probe state, and asymptotically saturable in the limit of many identical repetitions. A central theoretical goal is then to identify the optimal quantum probe states, that is, the best combination of mode engineering and particle-number statistics  that minimizes the QCRB for a given estimation task, and to determine the measurement that saturates it.

A broad consensus in the literature is to quantify resources by the mean photon number $N_S$, comparing only protocols that share the same value of this resource, while other state parameters such as bandwidth are often treated as classical resources. Within this framework, classical strategies are limited by the standard quantum limit (SQL), with estimator variance scaling as $N_S^{-1}$, whereas quantum-enhanced protocols  attain the Heisenberg limit, with variance scaling as $N_S^{-2}$. This paradigm has been demonstrated across a wide range of mode-parameter estimation problems: squeezed states in gravitational-wave interferometry \cite{PhysRevLett.123.231107,PhysRevLett.123.231108}, NOON-state phase estimation \cite{Lee01112002,dowling_quantum_2008} — experimentally demanding in both photonic and atomic platforms —, nonlinear interferometers based on stimulated emission of squeezed light \cite{qin_unconditional_2023}, which offer greater robustness to losses, time-frequency cluster states for time estimation \cite{descamps_quantum_2023}, and phase measurements using vibrational modes of trapped ions \cite{doi:10.1126/science.1097576,roos_designer_2006} or spin-squeezed states \cite{PhysRevA.46.R6797,sinatra_spin-squeezed_2022}. In all these cases, however, decoherence and losses ultimately prevent reaching the full Heisenberg limit at large photon number \cite{dowling_quantum_2008,descamps_quantum_2023}.
Using quantum resources are not only about beating the SQL: quantum metrology framework allows to highlight other quantum advantages in the low photon-number regime such as same resolution over the measurement of a given parameter using an entangled photon pair and low-classical field with phase stabilization, a higher signal-to-noise ratio with quantum states, that has impact on scanning fragile samples such as biological ones \cite{taylor_quantum_2016, lyons_attosecond-resolution_2018, chen_hong-ou-mandel_2019, fabre_parameter_2021,ndagano_quantum_2022}. In \cite{fabre_parameter_2021}, it was shown that two-photon spectral engineering at constant level of bandwidth allows reaching higher resolution on time and frequency estimation.  \\

More recently, Ref.~\cite{gessner2023b,sorelli2024,boeschoten2025}  developed a unified framework for mode-parameter estimation, providing a general recipe for designing quantum states that surpass the SQL and showing that populating at least two modes is generally required. These works  led to the identification of a broad class of quantum states that beat the classical limit, including coherent states with added squeezing, multimode squeezed vacuum states, and Fock states.  Quantum advantage arises not only from particle-number statistics of the probe but also from the modal structure of the field \cite{descamps_quantum_2023}, as illustrated for instance by the interplay between time and energy \cite{PRXQuantum.6.020351}. Recent work has shown that adding a phase reference to multimode entangled probes enables one to unambiguously identify the origin of quantum-enhanced precision, whether it stems from modal structure or from particle-number statistics \cite{saharyan_resources_2025}. Another important class of mode parameter estimation problems arises in lidar and radar \cite{Woodward1964,helstrom,reichert_quantum-enhanced_2022}, whose main tasks are to estimate the distance and velocity of a remote target by transmitting electromagnetic radiation and measuring the reflected echo. The distance is determined by measuring the time delay between the emitted and received signals, while the velocity is inferred from the Doppler-induced frequency shift. Various quantum protocols have been proposed for these tasks \cite{giovannetti2001quantum,giovannetti2002positioning,shapiro2007quantum,maccone2023gaussian,huang2021quantum,zhuang2017entanglement}, some of which demonstrate that quantum enhancement beyond the SQL is achievable.

In this manuscript, we study arbitrary mode parameter estimation protocols for a single unknown parameter $\lambda$, as represented in Fig.~\ref{introducfigure}. We first consider estimation with a finite set of bosonic modes, and show that the ultimate accuracy limit depends on the signal photon number  $N_S$ and two \textit{modal resource parameters}, $\bar{g}$ and $\Delta g$, which quantify characteristic features of the intensity distribution associated with a mode basis naturally induced by the mode transformation, for pure Gaussian quantum states. We derive an upper bound on the QFI of multimode pure Gaussian states expressed as a simple function of these three resources, exhibiting Heisenberg-limited scaling, and identify an optimal squeezed-vacuum state — constructed as a product of two single-mode squeezers in the mode basis induced by the transformation — whose QFI saturates this bound. This yields a resource-normalized framework in which quantum states sharing the same values of $\bar{g}$, $\Delta g$, and $N_S$ can be compared directly, with their QFIs serving as a benchmark for how effectively each state exploits the available resources. In particular, we show that differently spectrally engineered quantum states exhibit varying degrees of resolution enhancement through increased prefactors, while achieving Heisenberg scaling. We further show that multimode homodyne detection on these two modes constitutes an optimal measurement for achieving this bound. In the case of parameter shifts transformations, the two modal resources play physically distinct roles : achieving the precision associated with $\bar{g}$ requires phase-sensitive interferometric measurements that rely on prior knowledge of the parameter, whereas the precision enabled by $\Delta g$ can be attained with phase-insensitive measurements requiring no such prior knowledge.

We then apply our general framework to several concrete estimation scenarios, including time delays and frequency shifts \cite{fabre_parameter_2021,descamps_quantum_2023} as well as beam displacement and beam tilt \cite{delaubert2006tem,delaubert2006quantum,tsang_quantum_2016}. This requires starting from a continuous, infinite set of modes \cite{blow_continuum_1990} — the setting that arises naturally in the estimation of time delays and frequency shifts — which is subsequently discretized and truncated to make contact with the finite-mode results established in the first part. In these cases, the modal resource parameters $\bar{g}$ and $\Delta g$ reduce to familiar physical quantities: for time-shift (resp. frequency-shift) estimation, they correspond respectively to the mean frequency (mean time) and the spectral  bandwidth (time duration) of the probe state. This distinction between phase-sensitive and phase-insensitive measurements has direct implications for the practical implementation of optimal protocols across the scenarios we consider. The framework applies broadly, encompassing superresolution imaging protocols as well as lidar and radar scenarios, and provides a unified basis for identifying and comparing optimal quantum strategies across these diverse applications. We also discuss realistic experimental implementations of the optimal protocols identified throughout this work.

The article is organized as follows. In Sec.~\ref{sec:Preliminaries}, we introduce the necessary background on mode transformations, Gaussian states, and quantum estimation theory.  Sec.~\ref{sec:ModeParamterEstimationGaussian} presents the general expression for the QFI of a Gaussian state undergoing a mode transformation. There, we also derive a tight upper bound, identify a state that saturates it, and show that homodyne detection constitutes an optimal measurement. In Sec.~\ref{sec:TimeFrequencyEstimation} and Sec.~\ref{sec:BeamDisplacementAndTiltEstimation}, we apply our framework to several representative mode-parameter estimation problems, time and frequency parameters as well as beam-displacement and beam-tilt estimation, illustrating the practical relevance and usefulness of our results. Finally, in Sec.~\ref{conclusion}, we will conclude and open new perspectives.

\begin{figure}[htbp]
\centering
 \includegraphics[width=0.48\textwidth]{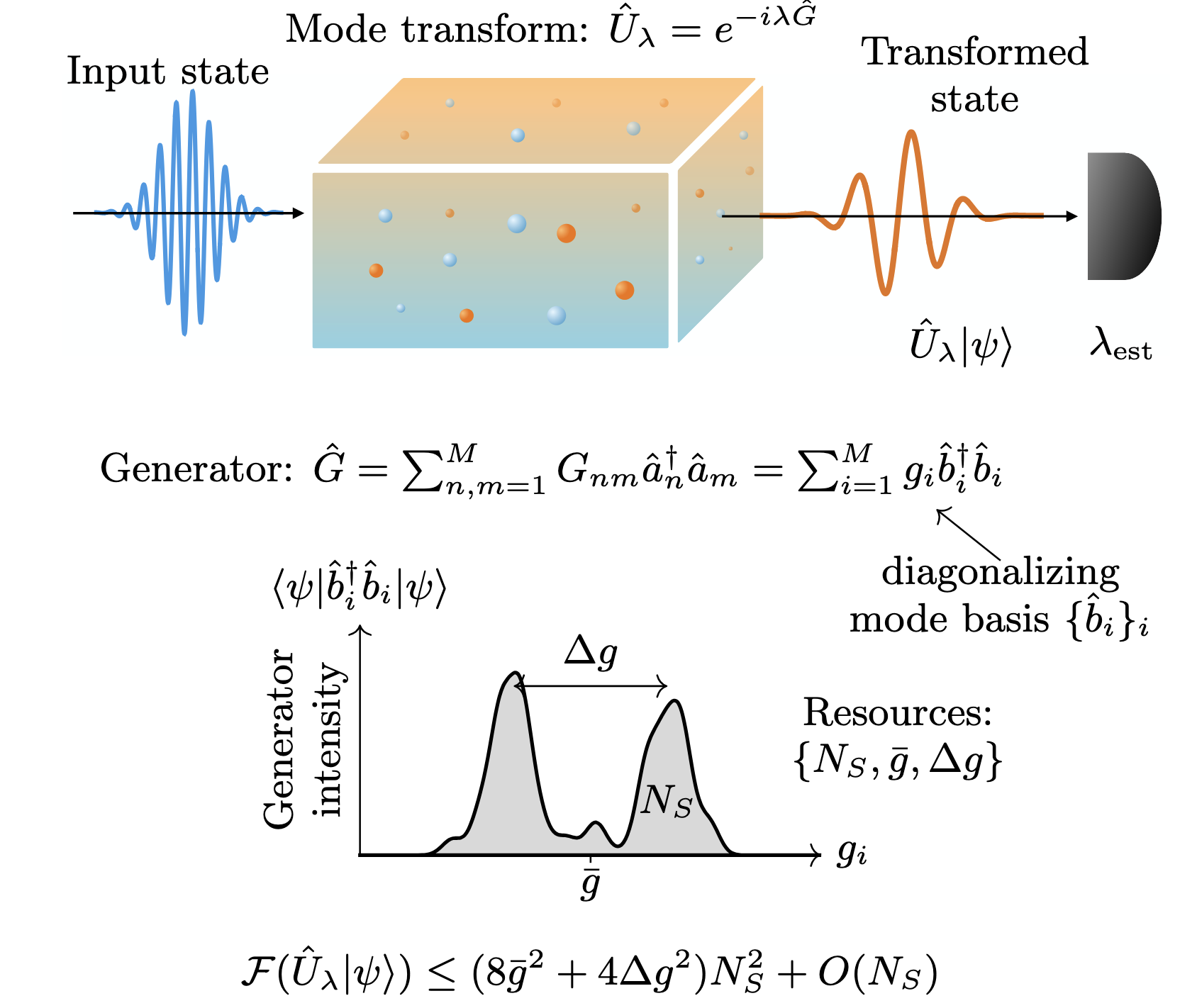}
\caption{ \label{introducfigure} 
(Top) General mode parameter estimation protocol. An input state $\ket{\psi}$ passes through an optical medium that induces a unitary mode transformation $\hat{U}_\lambda$, encoding the parameter of interest $\lambda$ into the state. $M$ being the number of modes. (Bottom) The relevant resources for the mode parameter estimation task are the signal photon number $N_S$, and the two modal resource parameters $\bar{g}$ and $\Delta g$, corresponding to the mean and variance of the normalized generator intensity distribution, as defined in the main text. }
\end{figure}

\section{Preliminaries}\label{sec:Preliminaries}

\subsection{Mode transforms}\label{sub:ModeTransform}
Let us introduce the modes and mode transformations that encode a parameter $\lambda$ into the field, where more details are given in Appendix~\ref{app:modetransformation}.
We introduce the set $\{\hat a_n\}_{n=1}^M$ of bosonic annihilation operators satisfying the commutation relations $[\hat a_n,\hat a_m^\dag]=\delta_{nm}$. For now, the number of modes $M$ is taken to be finite; in Sections~\ref{sec:TimeFrequencyEstimation} and \ref{sec:BeamDisplacementAndTiltEstimation} we show that our results extend to the case of infinitely many modes.
We note that our results apply to any physical system described by bosonic mode operators, but here we focus on quantum light.
The positive-frequency part of the electromagnetic field in a given quantization volume $\mathcal{V}$ can be decomposed as \cite{walschaers2021non,fabre2020modes}
\begin{align}
    \hat{\mathbf{E}}^{(+)}(\mathbf{r},t) = \sum_{n=1}^M \mathcal{E}_n\,\mathbf{v}_n(\mathbf{r},t)\, \hat a_n,
\end{align}
where $\mathcal{E}_n$ is the electric field per photon and $\{\mathbf{v}_n(\mathbf{r},t)\}_{n=1}^M$ is a set of $M$ mode functions that are solutions of Maxwell's equations, orthonormal under the scalar product $\int_{\mathcal{V}} d^3r\, \mathbf{v}_n^*(\mathbf{r},t)\, \mathbf{v}_m(\mathbf{r},t)/\mathcal{V} = \delta_{nm}$. The mode set can be chosen arbitrarily, though in many cases a preferred choice exists, adapted to the problem at hand. Here, $\hat a_n$ ($\hat a_n^\dag$) is the annihilation (creation) operator associated with the mode $\mathbf{v}_n(\mathbf{r},t)$. A mode transformation
\begin{align}
    \hat U_\lambda = \exp\!\left(-i\lambda \sum_{n,m=1}^M G_{nm}\, \hat a_n^\dag \hat a_m\right)\label{eq:UnitaryModeTransform}
\end{align}
is a unitary transformation generated by the Hermitian operator $\hat G=\sum_{n,m=1}^M G_{nm}\, \hat a_n^\dag \hat a_m$, where $G_{nm}$ are the elements of the Hermitian $M\times M$ generator matrix $G$. One could consider more general transformations with a parameter-dependent generator  $G_{nm}(\lambda)$. However, in local parameter estimation, this effectively reduces to our case given sufficiently accurate prior knowledge. For clarity, we therefore focus on this setting.  The action of $\hat U_\lambda$ on the mode operators is
\begin{align}
    \hat a_n(-\lambda) :=\hat U_\lambda^\dag\, \hat a_{n}\,\hat U_\lambda = \sum_{i=1}^{M} U_{ni}(\lambda)\,\hat a_i, \label{eq:HeisenbergTranformModeOperators}
\end{align}
where $U_{ni}(\lambda) = (e^{-i\lambda G})_{ni}$ are the elements of a unitary $M\times M$ matrix, and one verifies that $[\hat a_n(\lambda),\hat a_m^\dag(\lambda)]=\delta_{nm}$. This induces the following transformation of the electromagnetic field in the Heisenberg picture:
\begin{align}
     \hat{\mathbf{E}}^{(+)}(\mathbf{r},t;\lambda) = \sum_{n=1}^M \mathcal{E}_n\,\mathbf{v}_n(\mathbf{r},t)\,\hat a_n(-\lambda) = \sum_{i=1}^M \mathcal{E}_i\,\mathbf{v}_i(\mathbf{r},t;\lambda)\,\hat a_i,
\end{align}
where $\mathbf{v}_i(\mathbf{r},t;\lambda) = \sum_{n=1}^M U_{ni}(\lambda)\,\mathbf{v}_n(\mathbf{r},t)$ is a new set of parameter-dependent orthonormal modes.
The modes discussed above are cavity modes. In Sec.~\ref{sec:TimeFrequencyEstimation} and \ref{sec:BeamDisplacementAndTiltEstimation} we will also introduce spectro-temporal and spatial-momentum transverse modes.

\subsection{Gaussian states}
The quantum states of radiation we primarily consider in this paper are pure Gaussian states. In a given mode basis $\{\hat a_n\}_n$, any pure Gaussian state can be written as \cite{serafini2017}:
\begin{multline}\label{eq:GeneralGaussianArbitraryBasis}
    \vert\psi\rangle = \exp \left( \sum_{n=1}^M \beta_n \hat a_n^\dag - \mathrm{h.c.} \right) \\
     \times\exp \left( \frac{1}{2} \sum_{n,m=1}^M f_{nm}\,\hat a_n^\dag \hat a_m^\dag - \mathrm{h.c.} \right) \vert 0\rangle.
\end{multline}
 Here, $\vert 0\rangle$ is the vacuum state satisfying $\hat a_n\vert 0\rangle = 0$ for all $n$, $\beta_n\in\mathbb{C}$ is the displacement amplitude in mode $\hat a_n$, and $f_{nm}$ are the elements of the complex $M\times M$ squeezing matrix $f$ in the basis $\{\hat a_n\}_n$. Without loss of generality, $f$ can be taken to be symmetric, since its antisymmetric part cancels identically. The displacement vector $\boldsymbol{\beta}=(\beta_1,\ldots,\beta_M)^T$ and the squeezing matrix $f$ thus fully characterize the state. These quantities are straightforwardly related to the more commonly used mean vector and covariance matrix, as shown in Appendix.~\ref{app:CovarianceMatrix}.

By applying the Takagi factorization of the squeezing matrix, which is analogous to the Bloch–Messiah decomposition of the covariance matrix; see Appendix.~\ref{app:subsec:GaussianStateDisentangled}, any pure Gaussian state can be brought into a product form of independent single-mode displacement and squeezing operations:
\begin{align} \label{eq:GeneralGaussianState}
    \vert \psi\rangle = \hat V \left[ \bigotimes_{n=1}^M \hat D_{\hat a_n}(\alpha_n)\, \hat S_{\hat a_n}(r_n) \right] \vert 0\rangle.
\end{align}
Here, $\hat D_{\hat a_n}(\alpha_n) = \exp(\alpha_n \hat a_n^\dag - \alpha_n^* \hat a_n)$ is the single-mode displacement operator with $\alpha_n \in \mathbb{C}$, and $\hat S_{\hat a_n}(r_n) = \exp\!\left(\frac{r_n}{2}(\hat a_n^{\dag 2} - \hat a_n^2)\right)$ is the single-mode squeezing operator with $r_n \geq 0$. The operator $\hat V$ is a mode transformation of the same form as $\hat U_\lambda$, mixing the modes linearly as $\hat V^\dag \hat a_n \hat V = \sum_{i=1}^M V_{ni}\, \hat a_i$, where $V_{ni}$ are the elements of a unitary $M\times M$ matrix.

The Eq.~\eqref{eq:GeneralGaussianState} represents the same state as Eq.~\eqref{eq:GeneralGaussianArbitraryBasis}, but expressed in the mode basis $\{\hat c_n\}_n$ defined by $\hat c_n = \hat V \hat a_n \hat V^\dag$, in which the state takes the disentangled form
\begin{align}
\vert\psi\rangle = \bigotimes_{n=1}^M \hat D_{\hat c_n}(\alpha_n)\,\hat S_{\hat c_n}(r_n)\,\vert 0\rangle,
\end{align}
where we have used the identity $\hat V e^{\hat A}\hat V^\dag = e^{\hat V\hat A\hat V^\dag}$. In this basis, the displacement and squeezing parameters transform as $\alpha_n = \sum_{i=1}^M (V^{-1})_{ni}\,\beta_i$ and $\mathrm{diag}(r_1,\ldots,r_M) = V^{-1}f\,(V^T)^{-1}$, so that the squeezing matrix — and correspondingly the covariance matrix (see Appendix.~\ref{app:CovarianceMatrix}) — becomes diagonal.

Now, the Gaussian state transforms under the mode transform in the  Schr\"odinger picture as
\begin{align}
    \vert \psi_\lambda\rangle = \hat U_\lambda\vert\psi\rangle = \hat U_\lambda\hat V \hat D \hat S \vert 0\rangle ,
\end{align}
where we defined $\hat D= \otimes_{n=1}^M \hat D_{\hat a_n}(\alpha_n)$ and $\hat S = \otimes_{n=1}^M \hat S_{\hat a_n}(r_n)$.
As the mode transform $\hat U_\lambda$ is a passive Gaussian unitary, the parameter-imprinted state $\vert \psi_\lambda\rangle$ remains Gaussian after the interaction.

\subsection{Quantum estimation theory} 

To determine the ultimate accuracy limit for unbiased estimation of $\lambda$ encoded in the state $\vert\psi_\lambda\rangle$ for a particular positive-operator valued measure (POVM) $\{\hat\Pi_u\}_u$ with $\sum_u \hat \Pi_u = \mathds{1}$, we introduce the   Cram\'er-Rao bound (CRB) \cite{liu2020quantum}
\begin{equation}
\mathrm{Var}[\lambda_{\text{est}}] \geq \frac{1}{  F(\vert\psi_\lambda\rangle,\{\hat\Pi_u\}_u)} ,
\end{equation}
which is a lower bound on the variance for any unbiased estimator $\lambda_{\text{est}}$, given by the inverse of the   Fisher information (FI) 
\begin{align}
    F(\vert\psi_\lambda\rangle,\{\hat\Pi_u\}_u) = \sum_u \frac{(\partial_\lambda p(u\vert\lambda))^2}{p(u\vert \lambda)} ,
\end{align}
where $p(u\vert\lambda) = \langle \psi_\lambda \vert\hat \Pi_u\vert\psi_\lambda\rangle$ is the probability of measurement outcomes indexed as $u$.

By optimizing over all possible measurements, we obtain the quantum Cram\'er-Rao bound (QCRB)\cite{liu2020quantum}
\begin{equation}
\mathrm{Var}[\lambda_{\text{est}}] \geq \frac{1}{  \max_{\{\Pi_u\}_u}F(\vert\psi_\lambda\rangle,\{\hat\Pi_u\}_u)}  = \frac{1}{\mathcal{F}(\vert\psi_\lambda\rangle)},
\end{equation}
where the quantum Fisher information (QFI) for pure states is given by
\begin{equation} \label{eq:QFIdefinition}
{\cal{F}}(\vert\psi_\lambda\rangle)=4\left[ \langle \partial_{\lambda} \psi_\lambda |  \partial_{\lambda}\psi_\lambda \rangle - |\langle \partial_{\lambda} \psi_\lambda | \psi_\lambda \rangle|^{2} \right].
\end{equation}
As the QFI determines the ultimate accuracy limit for unbiased estimation, we use it as a  metric to compare the performance of different quantum states.

\section{Mode parameter estimation for Gaussian states}\label{sec:ModeParamterEstimationGaussian}
In this section, we present a resource-normalized framework applicable to any physical system described by a finite number of bosonic modes. The corresponding quantum circuit for mode parameter estimation is shown in Fig.~\ref{fig:circuit}.

\subsection{QFI for mode parameter estimation with Gaussian states}
We now determine the ultimate accuracy limits for mode parameter estimation with multimode pure Gaussian quantum states. First, let us note that the derivative of the parameter imprinted state can be written as~\cite{Gessner:23}: 
\begin{align}
    \vert \partial_\lambda \psi_\lambda \rangle 
    &= -i \hat U_\lambda \sum_{n,m=1}^M  G_{nm}     \hat a_n^\dag  \hat a_m   \vert \psi \rangle .
\end{align}
Using this together with the transformation rules $\hat V^\dag a_m \hat V = \sum_{j} V_{mj} \hat a_j$, $\hat D^\dag \hat a_j\hat D = \hat a_j +\alpha_j$, $\hat S^\dag \hat a_j \hat S = c_j \hat a_j + s_j \hat a_i^\dag$, 
we obtain, as shown in Appendix.~\ref{app:QFIderivation}, the following expression for the derivative state
\begin{align}
    \vert\partial_\lambda \psi_\lambda\rangle &= -i\hat U_\lambda \hat V \hat D \hat S \sum_{i,j=1}^M  \tilde G_{ij}  \Big[   c_i s_j \hat a_i^\dag \hat a_j^\dag \vert 0\rangle \nonumber\\ &+ (s_j \alpha_i^* \hat a_j^\dag + c_i \alpha_j \hat a_i^\dag)\vert 0\rangle +(s_i s_j \delta_{ij} +\alpha_i^* \alpha_j)\vert 0\rangle \Big] \nonumber\\
     &=:  \vert\partial_\lambda \psi\rangle_2 + \vert\partial_\lambda \psi\rangle_1 + \vert\partial_\lambda \psi\rangle_0,
\end{align}
  where we denote $\tilde G := V^\dag G V$ and  $c_j := \cosh (r_j)$ and $s_j := \sinh (r_j)$.
Note that $\langle\partial_\lambda\psi\vert \partial_\lambda\psi_\lambda\rangle = \, _2\langle\partial_\lambda\psi\vert \partial_\lambda\psi_\lambda\rangle_2 +\, _1\langle\partial_\lambda\psi\vert \partial_\lambda\psi_\lambda\rangle_1 +\, _0\langle\partial_\lambda\psi\vert \partial_\lambda\psi_\lambda\rangle_0$.
Putting the pieces together, we arrive at the mode parameter QFI for Gaussian states
\begin{equation}\label{eq:QFIcomponents}
\begin{aligned} 
    &\mathcal{F}(\vert \psi_\lambda\rangle) = 4\sum_{i,j=1}^M \tilde G_{ij} c_j s_j \tilde G_{ji}^* c_is_i  +4\sum_{i,j=1}^M  \tilde G_{ij} s_j s_j \tilde G_{ji} c_ic_i \\
    &\quad\quad+\sum_{i,i',j=1}^M 4\Big[  \tilde G_{ij} s_j s_j \tilde G_{j i'} \alpha_{i'}\alpha_i^* +  \tilde G_{i' j} c_j c_j \tilde G_{ji}\alpha_i\alpha_{i'}^*   
     \\&\quad\quad\quad\quad\quad+ \tilde G_{i' j} c_j s_j \tilde G_{ji}^* \alpha_i^*\alpha_{i'}^*  + \tilde G_{i' j}^* s_j c_j \tilde G_{ji}\alpha_i \alpha_{i'} \Big].
\end{aligned} 
\end{equation}
We refer to Appendix.~\ref{app:QFIderivation} for the detailed calculation.
  This formula is rather cumbersome and we can put it into a more convenient form
\begin{align}
    \mathcal{F}(\vert \psi_\lambda\rangle) &=4\Big ( \text{Tr}[\tilde G C S \tilde G^*  C S] + \text{Tr}[\tilde G S ^2 \tilde G C^2]\\ 
     &\quad+ \text{Tr}[\tilde G (S^2+C^2) \tilde G \mathcal A] + 2\mathrm{Re} \left[\text{Tr}[\tilde G^* S C \tilde G \mathcal B]\right]\Big) ,\label{sec4:eq:QFIexplicitFormula}
\end{align}
where  $C= \text{diag} (c_1,\ldots ,c_M)$, $S = \text{diag}(s_1,\ldots, s_M)$, and $\mathcal A$ and $\mathcal B$ are matrices with components $\mathcal A_{ij}= \alpha_i \alpha_j^*$ and $\mathcal B_{ij} = \alpha_i \alpha_j$. Ref.~\cite{sorelli2024} has found an alternative expression for the QFI that is also valid for mixed Gaussian states given in terms of symplectic eigenvalues.

\subsection{Resources and optimality for coherent states}\label{sub:ResourcesOptimalityCoherent}
As the general expression for the QFI for Gaussian states is quite complicated, let us first build some intuition by considering the subclass of coherent states, which are often considered to be classical states of radiation. 

The coherent states are of the form $\vert\psi^{\text{coh}}\rangle = \hat V \hat D\vert0\rangle$, which implies no squeezing, \textit{i.e}., $r_i =0$.  The QFI in this case takes the simple form
\begin{align}
    \mathcal{F}(\vert\psi^{\text{coh}}\rangle ) = 4\sum_{i,i',j=1}^M \alpha_{i'}^* \tilde G_{i' j}  \tilde G_{ji}\alpha_i   = 4 \boldsymbol{\alpha}^\dag \tilde G^2  \boldsymbol{\alpha} .
\end{align}
This is already a much simpler expression than the general case. To gain even deeper insight, let us rewrite the above expression. For that, note that $\tilde G = V^\dag G V$ is Hermitian, so that it can be diagonalized by a unitary $W$, i.e., $W  \tilde G W^\dag = D=\text{diag}(g_1,\ldots,g_M)$, where $\{g_n\}_n$ are the real eigenvalues of the generator matrix. We assume throughout the paper that $g_1\leq g_2\leq\ldots \leq g_M$. We find in Appendix.~\ref{app:QFIcoherent}
\begin{align}\label{eq:QFIcoherentDiag}
     \mathcal{F}(\vert\psi^{\text{coh}}\rangle ) = 4 \boldsymbol{\alpha}^\dag W^\dag D^2 W\boldsymbol{\alpha} = \sum_{n=1}^{M} g_n^2 \langle \psi^\text{coh}\vert \hat b_n^\dag \hat b_n\vert\psi^{\text{coh}}\rangle .
\end{align}
The bosonic operators  $\hat b_n = \sum_{i} (B^\dag)_{ni} \hat a_i$ with unitary $M\times M$ matrix $B^\dag =WV^\dag$ and $[\hat b_n,\hat b_m^\dag]=\delta_{nm}$  here correspond to the special mode basis in which the generator appears diagonal, that is, 
\begin{align}
    \hat G= \sum_{i,j=1}^M G_{ij}\hat a_i^\dag \hat a_j = \sum_{n=1}^M g_n \hat b_n^\dag \hat b_n.
\end{align}
Let us seek the optimal coherent state that maximizes the QFI. We see that by, for example, letting the photon number of some mode $j$ with $g_j\ne 0$ grow without bound, i.e., $\langle \psi^\text{coh}\vert \hat b_j^\dag \hat b_j\vert\psi^{\text{coh}}\rangle  \rightarrow \infty$, we have  $\mathcal{F}(\vert\psi^{\text{coh}}\rangle ) \rightarrow \infty$. Hence, an unbounded increase in photon number leads to an arbitrarily large QFI. In this naive sense, the “optimal’’ coherent state would correspond to one with an infinite photon number. However, in practice, one has only access to a finite amount of photons, and thus it makes sense to consider the number of photons a (constraint) resource. With this in mind, one can seek the optimal coherent state that maximizes the QFI for a constrained amount of photons.

In many physical settings, the number of modes is infinite, $M = \infty$, and the generator matrix $G$ may likewise be unbounded. For mathematical clarity, we consider $M$ to be finite but arbitrarily large, interpreting it as a truncation cutoff, and allow the generator eigenvalues $\{g_1, \ldots, g_M\}$ to grow without bound as $M$ increases. This situation arises in the examples considered in Secs.~\ref{sec:TimeFrequencyEstimation} and \ref{sec:BeamDisplacementAndTiltEstimation}. In this regime, Eq.~\eqref{eq:QFIcoherentDiag} shows that the QFI can grow without bound even at fixed photon number, which underscores the necessity of introducing additional resource parameters associated with the modal distribution of the state.

We introduce two additional resources $\bar g$ and $\Delta g$ related to how the photon intensity is distributed among the modes $\{\hat b_n\}_n$ that diagonalize the generator $\hat G= \sum_{n}g_n\hat b_n^\dag \hat b_n$.
We define the normalized generator intensity distribution as
\begin{align}
    I_n = \frac{\langle\psi^\text{coh}\vert \hat b_n^\dag \hat b_n\vert\psi^{\text{coh}}\rangle}{N_S},
\end{align} 
where
 \begin{align}\label{eq:DefResourcePhot}
     N_S = \sum_{n\in \mathcal{I}_S} \langle\psi^\text{coh}\vert \hat b_n^\dag \hat b_n\vert\psi^{\text{coh}}\rangle
 \end{align}
 is the signal photon number.  The sum runs over all indices in  $\mathcal{I}_S := \{\, i = 1,\ldots,M \mid g_i \neq 0 \,\}$. The modes that belong to $\mathcal{I}_S$ are referred to   the signal modes. The complement $\mathcal{I}_I := \{\, i = 1,\ldots,M \mid g_i = 0 \,\}$ of $\mathcal{I}_S$  is the set of idler modes, as these modes are unaffected by the action of the mode transform $\hat U_\lambda$.
 We further define
 \begin{align}\label{eq:DefResourceMean}
     \bar g = \sum_{n=1}^M g_n I_n
 \end{align}
as the mean value of the normalized generator intensity and 
\begin{align}\label{eq:DefResourceVariance}
    \Delta g^2 = \sum_{n=1}^M (g_n-\bar g)^2 I_n
\end{align}
as the variance. With this, we see that the QFI takes the form
\begin{align}
      \mathcal{F}(\vert\psi^{\text{coh}}\rangle )  = 4 \left( \bar g^2 +\Delta g^2 \right) N_S . \label{eq:QFIcoherentGeneral}
\end{align}
In this form, it is clear that $\{N_S, \bar{g}, \Delta g\}$ play the role of resource parameters: increasing any one of them raises the QFI and thus improves metrological performance. Notably, any two coherent states sharing the same values of these parameters yield identical QFI, regardless of the detailed structure of the underlying modal intensity distribution. The metrological performance is therefore determined entirely by the first two moments $\bar{g}$ and $\Delta g$ of that distribution, together with the signal photon number $N_S$. As we show below, the full modal photon distribution becomes relevant once squeezing is introduced.

\subsection{Resources for Gaussian states}\label{sec:Resources}
Let us now generalize the insights from the previous section for coherent states to general pure Gaussian states. The  parameters $\{ \bar g,\Delta g, N_S\}$ defined in Eqs.\eqref{eq:DefResourcePhot}-\eqref{eq:DefResourceVariance} can be straightforwardly generalized by using 
\begin{align}
     I_n = \frac{\langle\psi \vert \hat b_n^\dag \hat b_n\vert\psi\rangle}{N_S},
\end{align}
now taking the expectation value of the photon number operator with respect to a general Gaussian state $\vert\psi\rangle$. Note that  $\bar g  =\langle \psi  \vert \hat G\vert\psi\rangle/N_S$, but importantly $\Delta g^2 \not\propto \langle \psi\vert  \hat G^2\vert\psi\rangle -\langle \psi\vert \hat G\vert\psi\rangle^2$. The variance $\Delta g^2$ is not a variance with respect to a quantum state, but a variance of the normalized generator intensity distribution $I_n$. We can also write the resource parameters in a basis-independent way 
\begin{align}
N_S &= \text{Tr}[\tilde P_S (S^2+\mathcal A)] \label{eq:DefResourceMeanPhotIndependent}\\
    \bar g &=\frac{\text{Tr}[ \tilde G (S^2 + \mathcal A)]}{N_S} \label{eq:DefResourceMeanBasisIndependent}\\
    \Delta g ^2 &= \frac{\text{Tr}[ \tilde G^2 (S^2 + \mathcal A)] }{N_S} -\bar g^2 \label{eq:DefResourceVarianceBasisIndependent}
\end{align}
as we show in Appendix.~\ref{app:ResourceParameters}. Here, $\tilde P_S = V^\dag P_S V$, where $P_S$ is the projector onto the support of $G$,  that is, onto the subspace spanned by eigenvectors of $G$ associated with all non-zero eigenvalues. Note that, for any pure state, the QFI takes the simple form $\mathcal F = 4\,\mathrm{Var}(\hat G)$. In general, $\mathrm{Var}(\hat G)$ depends on the particle-number statistics of the state (e.g., coherent, squeezed, or Fock-like, including its total photon number $N_S$) with the modal distribution of light over the eigenmodes of $\hat G$. For non-Gaussian states these two contributions need not factorize (see \cite{descamps_quantum_2023}). For pure Gaussian states, however, this factorization does hold, which is precisely why $\mathrm{Var}(\hat G)$ — and hence the QFI — can be written in terms of the three independent resources used throughout this work: the total photon number $N_S$, the mean generator eigenvalue $\bar g$, and the eigenvalue spread $\Delta g$. This separation makes explicit which contributions to the QFI scale at the Heisenberg limit ($N_S^2$) and which scale at the shot-noise limit ($N_S$). One of the central findings of this work is:
\begin{theorem} \label{TheoremBound}
The QFI $\mathcal{F}(\vert\psi_\lambda\rangle)$ of a general pure Gaussian quantum state $\vert \psi_\lambda\rangle = \hat U_\lambda \vert\psi\rangle$ subject to a mode transform   given in Eq.~\eqref{eq:UnitaryModeTransform} is upper bounded by
\begin{align}
    \mathcal{F}(\vert \psi_\lambda\rangle) \leq \mathcal{F}^{\mathrm{UB}}= \left( 8 \bar g^2  + 4 \Delta g^2  \right) N_S^2+  O(N_S) ,\label{eq:QFIbound}
\end{align}
where $\{N_S,\bar g,\Delta g\}$ are the resource parameters as defined in Eqs.~\eqref{eq:DefResourceMeanPhotIndependent}-\eqref{eq:DefResourceVarianceBasisIndependent}  of the input state $\vert\psi\rangle$.
\end{theorem}

 The proof is provided in Appendix.~\ref{App:ProofTheorem}. Establishing the above bound and demonstrating that it can be saturated confirms that the quantities $\{N_S, \bar{g}, \Delta g\}$ defined in Eqs.~\eqref{eq:DefResourceMeanPhotIndependent}--\eqref{eq:DefResourceVarianceBasisIndependent} are genuine resources for mode parameter estimation: larger values of any one of them raise the ultimate precision limit and thus directly enhance the ultimate metrological performance.

Let us give some context to interpret this result. 
In the following, we focus on the regime $N_S \gg 1$ and primarily consider the
$N_S^2$ contributions, which exhibit HL scaling.
In many cases of interest, the QFI of a Gaussian state takes the form 
\begin{align}
    \mathcal{F}(\vert\psi_\lambda\rangle ) = (c_{\bar g} \bar g^2 +c_{\Delta g}\Delta g^2)N_S^2+O(N_S) \label{eq:QFIoptimalityDef},
\end{align}
 where $c_{\bar g} \geq 0$ and $c_{\Delta g} \geq 0$ are dimensionless constants. Provided $c_{\bar g},c_{\Delta g}>0$, it is clear that for large enough $N_S$, such states will always outperform any coherent state whose QFI only scales as $\sim N_S$. For this reason, one often focuses only on the scaling with $N_S$, rather than on the prefactors. The latter become important when comparing two Gaussian states whose QFIs both scale as $N_S^2$.
Traditionally, Gaussian states with the same photon number $N_S$ have been compared by examining their total QFI values. However, such comparisons can be misleading because the QFI also depends on the modal resource parameters $\bar{g}$ and $\Delta g$: a state with larger $\bar{g}$ may appear superior even if it exploits its resources less efficiently. Our approach overcomes this limitation by introducing a resource-normalized framework in which Gaussian states can be compared on equal footing. Fixing the three resources $\{N_S, \bar{g}, \Delta g\}$, the coefficients $c_{\bar{g}}$ and $c_{\Delta g}$ encapsulate how the modal distribution of squeezing and displacement in a given Gaussian state affects its metrological performance, providing a direct and physically meaningful figure of merit for comparing competing strategies. Within this framework, the optimal states — those saturating the upper bound in the large-photon-number regime $N_S \gg 1$ — are characterized by $c_{\bar{g}} = 8$ and $c_{\Delta g} = 4$.\\

Note that, we restrict our analysis to pure Gaussian states without loss of generality. This bound can be derived explicitly and non-asymptotically for \emph{any} mixed Gaussian state $\hat\rho$ via purification. For mixed states $\hat\rho$, we can simply append $M$ auxiliary modes on which $\hat G$ does not act, i.e., $g_n=0$ for all auxiliary modes. A Gaussian purification $\vert\psi\rangle_A$ of $\hat\rho$ with $\mathrm{Tr}_A[\vert\psi\rangle_A\,{}_A\langle\psi\vert]= \hat\rho$ always exists. Note that $\vert\psi\rangle_A\,{}_A\langle\psi\vert$ and $\hat\rho$ have the same resources $\{N_S,\bar g, \Delta g\}$ which are calculated using the natural generalizations $I_n ={}_A\langle\psi\vert \hat b_n^\dag \hat b_n\vert\psi\rangle_A/N_S$ and $I_n = \mathrm{Tr}[\hat b_n^\dag \hat b_n \hat \rho]/N_S$ in Eqs.\eqref{eq:DefResourcePhot}-\eqref{eq:DefResourceVariance}.
Using the monotonicity of the QFI under completely positive trace-preserving maps we obtain  $\mathcal F (\hat\rho_\lambda) = \mathcal F (\mathrm{Tr}_A[\vert\psi\rangle_A\,{}_A\langle\psi\vert])  \leq \mathcal F (\vert\psi_\lambda\rangle_A)\leq (8\bar g^2+4 \Delta g^2)N_S^2+O(N_S)$.

\subsection{Optimal states}\label{subsec:OptimalStates}
We now turn to the discussion of states that achieve different forms of optimality as defined by Eq.~\eqref{eq:QFIbound}. As we demonstrate in Secs.~\ref{sec:TimeFrequencyEstimation} and \ref{sec:BeamDisplacementAndTiltEstimation}, the two modal resources $\bar{g}$ and $\Delta g$ play physically distinct roles in the concrete estimation scenarios we consider. The quantity $\bar{g}$ is associated with interferometric sensitivity: achieving the corresponding QFI scaling $\sim \bar{g}^2 N_S^2$ requires a phase-sensitive measurement, which typically performs best when substantial prior knowledge of the parameter is available. In contrast, the contribution $\sim \Delta g^2 N_S^2$ is accessible through phase-insensitive measurements, making $\Delta g$ the operationally relevant resource when such prior knowledge is unavailable. The relative importance of $\bar{g}$ and $\Delta g$ therefore depends on both the measurement strategy employed and the degree of prior information about the parameter.\\

We classify quantum states according to the type of optimality they achieve in the asymptotic scaling of the quantum Fisher information (QFI). 
\begin{definition}[Optimality classes of states]\label{def:StateOptimality}
Let $c_{\bar g}$ and $c_{\Delta g}$ be the coefficients appearing in the asymptotic expansion of the QFI in Eq.~\eqref{eq:QFIoptimalityDef}.
\begin{itemize}
    \item A state is said to be \emph{optimal} if $c_{\bar g} = 8$ and $c_{\Delta g} = 4$.
    \item A state is said to be \emph{mean optimal} if $c_{\bar g} = 8$ and $c_{\Delta g} < 4$.
    \item A state is said to be \emph{variance optimal} if $c_{\bar g} < 8$ and $c_{\Delta g} = 4$.
\end{itemize}
\end{definition}
Optimal states simultaneously achieve the maximal asymptotic scaling associated with both the $\bar g$ and the  $\Delta g$ contributions.
Mean-optimal (variance-optimal) states saturate only the former (latter) contribution.

\subsubsection{Optimal state}\label{subsubsec:Optimal}
Let us now present an optimal state. It is a two-mode state and takes on a simple form in the $\{ \hat b_n\}_n$ mode basis, which is the mode basis that diagonalizes the generator, i.e., $\hat G= \sum_n g_n \hat b_n^\dag \hat b_n$. It takes the form
\begin{align}
    \vert\psi\rangle = \hat S_{e^{-i\phi_i/2}\hat b_{i}}(r_i) \hat S_{e^{-i\phi_j/2}\hat b_{j}}(r_j) \vert 0\rangle , \label{eq:CompletlyOptimalState}
\end{align}
where $i\neq j$ are some fixed mode indices, $r_i,r_j$ are the squeezing strengths and  $\phi_n$ and $\phi_m$ are  freely adjustable squeezing angles. The different mode basis are represented in Fig.~\ref{fig:circuit}.

For desired mean and variance resource parameters $\bar g$ and $\Delta g$, respectively, we choose the squeezing parameters $r_i,r_j$, such that
\begin{align}
    s_i^2 &= \frac{N_S}{2} \left(1 - \frac{ \bar g}{\sqrt{\bar g^2+\Delta g^2}} \right) \label{eq:OptSqueezing1}\\
     s_j^2 &= \frac{N_S}{2} \left(1 + \frac{ \bar g}{\sqrt{\bar g^2+\Delta g^2}}\right)  \label{eq:OptSqueezing2}
\end{align}
holds.
We find that this state achieves an optimal QFI of
\begin{align}
    \mathcal{F}(\vert\psi_\lambda\rangle) = (\bar 8 g^2 + 4\Delta g^2)N_S^2+\underbrace{8 (\bar g^2 +\Delta g^2)N_S}_{=O(N_S)} 
\end{align}
for arbitrary squeezing angles $\phi_i,\phi_j$, if we can find suitable modes with indices $i,j$, such that the corresponding generator eigenvalues $g_i,g_j$ (with $g_i < g_j$ without loss of generality) satisfy
\begin{align}
     g_i&= \bar g -    \Delta g  \sqrt{ \frac{s_j^2}{s_i^2}}  \label{eq:opteigenval1}\\
     g_j &= \bar g  +    \Delta g  \sqrt{ \frac{s_i^2}{s_j^2}}, \label{eq:opteigenval2}
\end{align}
and  note that $\Delta g > 0$.

For given resource parameters $\{\bar{g}, \Delta g, N_S\}$ and a given generator $\hat G$ whose generator matrix has spectrum $\{g_1, \ldots, g_M\}$, it may not always be possible to find indices $i, j$ such that the corresponding eigenvalues $g_i, g_j$ satisfy Eqs.~\eqref{eq:opteigenval1} and \eqref{eq:opteigenval2}. In such cases, a strictly optimal state may not exist. However, in many physically relevant situations — including those discussed in Secs.~\ref{sec:TimeFrequencyEstimation} and \ref{sec:BeamDisplacementAndTiltEstimation} — the number of modes is large, $M \gg 1$, and the generator eigenvalues are densely, and possibly uniformly, distributed or even form a continuum. In such cases, indices $i, j$ can always be found for which $g_i$ and $g_j$ satisfy Eqs.~\eqref{eq:opteigenval1} and \eqref{eq:opteigenval2} to an excellent approximation, and the resulting state is therefore essentially optimal.

We also note that a state similar state to Eq.~\eqref{eq:CompletlyOptimalState} that makes use of two additional idler modes is also optimal as we discuss in Appendix.~\ref{app:IdlerAssistedOptimalState}. Additional optimal states may exist beyond those discussed here; however, we are not aware of any others.

\subsubsection{Variance-optimal state}\label{subsubsec:VarianceOptimal}
Let us now present the two variance-optimal states we found. The first variance-optimal state has the same form as the fully optimal state in Eq.~\eqref{eq:CompletlyOptimalState}, but employs the squeezing equally in both modes, \textit{i.e}., $r_i=r_j$.
We calculate the QFI in Appendix.~\ref{subsub:NoDisplacement} and find
\begin{align}\label{eq:VarianceOptimalState}
    \mathcal{F}(\vert\psi_\lambda\rangle) = 4(\bar g^2 + \Delta g^2)N_S^2+\underbrace{8 (\bar g^2 +\Delta g^2)N_S}_{=O(N_S)} .
\end{align}
We see that the prefactor of the variance term $\Delta g^2 N_S^2$ takes on its maximum value $4$, and thus the state is variance optimal. The prefactor of the mean term $\bar g^2N_S^2$ is $4$, only one half of its maximum value. Note that the state is optimal for arbitrary squeezing angles $\phi_i,\phi_j$. Also note that even in the limit $N_S\ll 1$, the state outperforms the coherent state by a factor of $2$.

We also identify another two-mode states, in which half of the energy budget is used to displace one mode and the other half used for squeezing the other one. The first mode is $\hat u_0= (\hat b_i+\hat b_j)/\sqrt{2}$, the second is $\hat u_1= (\hat b_i-\hat b_j)/\sqrt{2}$. We have $  \partial_\lambda\hat u_0^\dag (\lambda) = -i [ G_{00} \hat u_0^\dag (\lambda) +  G_{10} \hat u_1^\dag (\lambda) ]$ and $  \partial_\lambda\hat u_1^\dag (\lambda) = -i [ G_{10} \hat u_0^\dag (\lambda) +  G_{11} \hat u_1^\dag (\lambda) ]$, i.e., the derivative modes are again a linear combination of the two modes $\hat u_0,\hat u_1$. The state is
\begin{align}
    \vert\psi\rangle = \hat D_{e^{-i\phi_0/2}\hat u_{0}}(\alpha_0) \hat S_{e^{-i\phi_1/2}\hat u_{1}}(r_1) \vert 0\rangle , \label{eq:VarianceOptimalDisplaced}
\end{align}
with $\vert\alpha_i\vert^2=\sinh^2(r_j)$. For the correct phase setting, we find a QFI of $  2\bar g^2 + 4\Delta g^2N_S^2+O(N_S)$, which is also variance optimal.
We provide the detailed QFI calculation in Appendix.~\ref{subsub:WithDisplacement}. In fact, our calculation applies to general states with one mode displaced and the derivative of the mode squeezed, similar to the states discussed in Ref.~\cite{pinel2012ultimate}. In Appendix.~\ref{subsub:WithDisplacement}, we derive explicit conditions for when these states are variance optimal.
 In Sec.~\ref{subsec:VarianceOptimalStateTimeFrequency} we  further discuss these types of state for the example of time shift estimation.

\subsubsection{Mean-optimal state}\label{subsub:MeanOptimal}
Let us now present a mean-optimal state according to our definition above. We consider a single-mode squeezed state
\begin{align}\label{eq:MeanOptimalState}
    \vert\psi\rangle = \hat V \hat S_{\hat a_j}(r_j)\vert 0\rangle = \hat S_{\hat V\hat a_j\hat V^\dag} (r_j) \vert 0\rangle,
\end{align}
that is squeezed in a completely arbitrary mode $\hat V  \hat a_j \hat V^\dag$ with $j\in\{1,\ldots,M\}$, where $\{\hat a_n\}_n$ are the operators corresponding to some initial mode basis. It is straightforward to show that
\begin{align}
    \mathcal{F}(\vert\psi_\lambda\rangle) = (8\bar g^2 +0\cdot \Delta g^2)N_S^2+\underbrace{4 (\bar g^2 +\Delta g^2)N_S}_{=O(N_S)} .
\end{align}
Therefore, the state is mean optimal according to the definition given above. Interestingly,  the explicit shape of the squeezed mode $\hat V \hat a_j\hat V^\dag$ (with arbitrary $\hat V$) becomes unimportant, similar to the coherent state case discussed in Sec.~\ref{sub:ResourcesOptimalityCoherent}, whereas for the optimal and variance-optimal states, the squeezed modes must correspond to the specific modes $\hat{b}_i,\hat b_j$.
Note that the prefactor of the variance term is $c_{\Delta g} = 0$; consequently, single-mode squeezing cannot be fully optimal when $\Delta g > 0$. Note that the fully optimal state presented in Sec.~\ref{subsubsec:Optimal} reduces to a single-mode squeezed state in the regime $\bar g^2\gg \Delta g^2$.

\subsection{Optimal measurements}\label{sec:OptimalMeasurements}
Let us now discuss two classes of optimal measurement strategies, that is represented at the right of the quantum circuit in Fig.~\ref{fig:circuit}. 
\begin{definition}[Optimal measurements]\label{def:OptimalMeasurements}
Let $c_{\bar g}$ and $c_{\Delta g}$ be the coefficients appearing in the asymptotic expansion of the QFI in Eq.~\eqref{eq:QFIoptimalityDef}.
\begin{itemize}
    \item A measurement $\{ \hat \Pi_a \}_a$ is said to be \emph{optimal} for a particular state $\hat U_\lambda\vert\psi\rangle$, if $F(\hat U_\lambda\vert\psi\rangle,\{ \hat \Pi_a \}_a)=\mathcal{F}(\hat U_\lambda\vert\psi\rangle)$ .
     \item A measurement $\{ \hat \Pi_a \}_a$ is said to be \emph{variance optimal} for a particular state $\hat U_\lambda\vert\psi\rangle$ whose QFI can be expressed as in Eq.~\eqref{eq:QFIoptimalityDef}, if $F(\hat U_\lambda\vert\psi\rangle,\{ \hat \Pi_a \}_a)= c_{\Delta g} \Delta g^2 N_S^2 + c_{\bar g}' \bar g^2 N_S^2 + O(N_S)$, where $c_{\bar g}'\geq 0$ with $c_{\bar g}'\leq c_{\bar g}$.
\end{itemize}
\end{definition}

\subsubsection{Phase-sensitive optimal measurements}
We refer to multimode homodyne detection as a phase-sensitive measurement, since the detection scheme depends on adjustable phase parameters that can significantly affect its performance.

\begin{figure}[htbp]
\centering
\begin{tikzpicture}
  \node[inner sep=0] (img) {\includegraphics[width=\linewidth]{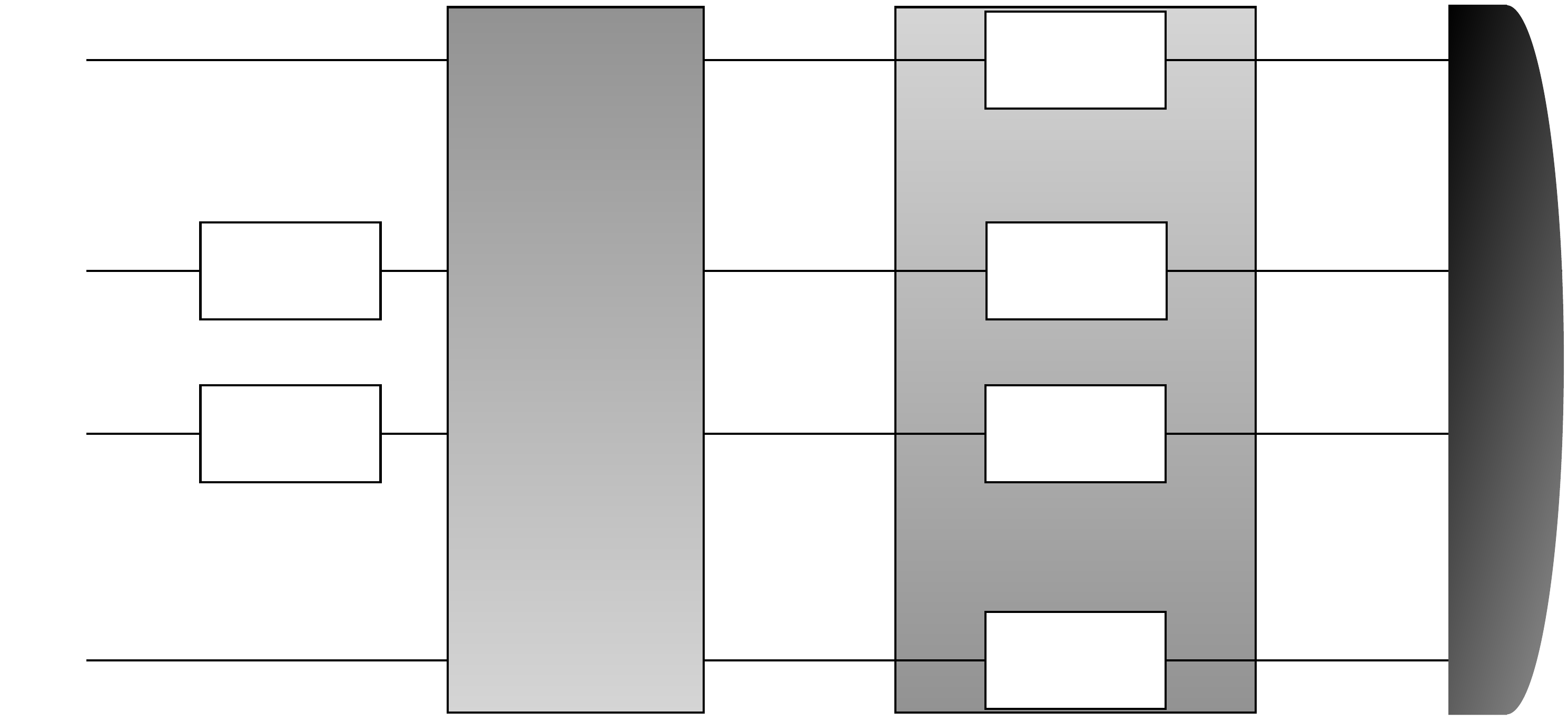}};
    \node at (-4,1.67) {$ \vert 0\rangle$};
    \node at (-3.6,1.87) {$  \hat a_1$};
    \node at (-0.25,1.87) {$  \hat b_1$};
    \node at (-4,1) {$\vdots$};
    \node at (-4,1.67) {$ \vert 0\rangle$};
     \node at (-4,0.5) {$ \vert 0\rangle$};
    \node at (-3.6,0.7) {$  \hat a_i$};
    \node at (-0.25,0.7) {$  \hat b_i$};
    \node at (-4,0.2) {$\vdots$};
    \node at (-4,-0.4) {$ \vert 0\rangle$};
    \node at (-3.6,-0.2) {$  \hat a_j$};
    \node at (-0.25,-0.2) {$  \hat b_j$};
    \node at (-4,-0.9) {$\vdots$};
    \node at (-4,-1.65) {$ \vert 0\rangle$};
    \node at (-3.6,-1.45) {$  \hat a_M$};
    \node at (-0.21,-1.45) {$  \hat b_M$};
    \node at (-1.2,-2.2) {$ \hat V=  \hat B$};
    \node at (1.7,-2.2) {$  \hat U_\lambda = e^{-i\lambda \sum_n g_n\hat b_n^\dag \hat b_n} $};
    \node at (-2.7,0.5) {$  \hat S_{\hat a_i}(r_i)$};
    \node at (-2.71,-0.4) {$  \hat S_{\hat a_j}(r_j)$};
    \node[scale=0.4] at (1.63,1.67) {$\exp({-i \lambda g_1  \hat b_1^\dag\hat b_1})$};
    \node[scale=0.4] at (1.63,0.5) {$\exp({-i \lambda g_i  \hat b_i^\dag\hat b_i})$};
     \node[scale=0.4] at (1.63,-0.4) {$\exp({-i \lambda g_j  \hat b_j^\dag\hat b_j})$};
     \node[scale=0.36] at (1.62,-1.65) {$\exp({-i \lambda g_M  \hat b_M^\dag\hat b_M})$};
\end{tikzpicture}
\caption{\label{fig:circuit}Quantum circuit representing the optimal mode parameter  estimation protocol. Two single-mode squeezed states are prepared in modes $\hat a_i,\hat a_j$ with squeezing angles $\phi_i=\phi_j=0$ (for illustrative clarity) after which the passive Gaussian unitary $\hat V =\hat B$ transfers this squeezing to modes $\hat b_i,\hat b_j$. In that mode basis, the mode parameter transform $\hat U_\lambda$ acts as a simple phase shift on each individual mode. After crossing the sample, the state is measured, with either phase-sensitive or insensitive detection.}
\end{figure}
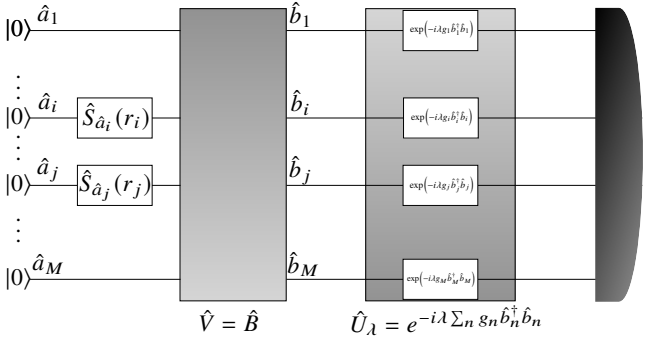

For the optimal and variance-optimal two-mode squeezed states presented in Secs.~\ref{subsubsec:Optimal} and \ref{subsubsec:VarianceOptimal}, we consider homodyne detection on both populated modes $i$ and $j$, corresponding to measuring the quadrature operators $\hat b_i e^{-i\varphi_i} + \hat b_i^\dag e^{+i\varphi_i}$ and $\hat b_j e^{-i\varphi_j} + \hat b_j^\dag e^{+i\varphi_j}$. Optimizing the measurement phases $\varphi_i, \varphi_j$ — to be distinguished from the squeezing angles $\phi_i, \phi_j$ — for the optimal state of Sec.~\ref{subsubsec:Optimal} yields the maximal Fisher information
\begin{align}
    \max_{\varphi_{i},\varphi_{j}} F_{\text{hom}} (\vert\psi_\lambda\rangle) = (8\bar{g}^2 + 4\Delta g^2)N_S^2 + O(N_S),
\end{align}
as derived in Appendix.~\ref{app:OptimalMeasurementVarianceCompletly}. For the variance-optimal state, the  optimization gives $\max_{\varphi_{i},\varphi_{j}} F_{\text{hom}} (\vert\psi_\lambda\rangle) = 4(\bar{g}^2 + \Delta g^2)N_S^2 + O(N_S)$. In both cases, two-mode homodyne detection is therefore an optimal measurement.

The optimal phase settings $\varphi_i^{\text{opt}}(\lambda)$ and $\varphi_j^{\text{opt}}(\lambda)$ depend on the parameter $\lambda$ to be estimated, so that achieving near-optimal performance in practice requires sufficiently accurate prior knowledge of $\lambda$, in the form of a prior estimate $\lambda_{\text{prior}}$. When such prior information is unavailable, a phase-insensitive measurement strategy, as discussed in Sec.~\ref{subsub:VarianceOptimalMeasurement}, is preferable.

 In Appendix.~\ref{app:Loss}, we consider the effect of photon loss and detector inefficiency, modeled by beam splitters of transmissivity $\eta_i, \eta_j$ placed immediately before the detectors. To obtain a compact result, we assume $\eta := \eta_i = \eta_j$ and work in the limit $s_i^2, s_j^2 \gg 1$, yielding
\begin{align}\label{homodyne}
      \max_{\varphi_{i},\varphi_{j}} F_{\text{hom}}(\vert\psi_\lambda\rangle) \approx \frac{2\eta}{1-\eta} \left[ \bar{g}^2 + \Delta g^2 \right] N_S.
\end{align}
This exceeds the QFI of a coherent state, $4\eta \left[ \bar{\omega}^2 + \Delta \omega^2 \right] N_S$, for transmissivities $\eta > 1/2$. We further show in Appendix.~\ref{app:OptimalMeasurementMean} that, for the mean-optimal state presented in Sec.~\ref{subsub:MeanOptimal}, single-mode homodyne detection is an optimal measurement; this analysis is restricted to the lossless case $\eta = 1$.

\subsubsection{Phase-insensitive variance-optimal measurements}\label{subsub:VarianceOptimalMeasurement}
We now turn to mode- and photon-number–resolved photon counting, a measurement that is \emph{phase-insensitive} in the sense that it requires no phase control, in contrast to homodyne detection. For the time and frequency estimation examples of Sec.~\ref{sec:TimeFrequencyEstimation}, as well as the beam displacement and tilt estimation of Sec.~\ref{sec:BeamDisplacementAndTiltEstimation}, the generators share the property that their eigenvalues $\{g_n\}_n$ are equally spaced. For these cases, we show that photon-number–resolved photon counting in a suitably chosen mode basis — the Fourier-dual of the generator eigenbasis $\{\hat b_i\}_i$ — constitutes a variance-optimal measurement in the sense of Definition~\ref{def:OptimalMeasurements}, provided the input state satisfies the conditions of Lemma~\ref{lemma:VarianceOptimal}.

Such phase-insensitive measurements are of significant practical value in low-information scenarios, such as lidar-based range (time) and velocity (frequency) estimation. In realistic lidar settings, the precision associated with the mean contribution $\bar{g}^2$ is generally unattainable due to phase ambiguity \cite{helstrom}, leaving only the variance contribution $\Delta g^2$ accessible. Identifying robust, phase-insensitive measurement strategies that are variance-optimal is therefore essential for such protocols.

\subsection{Previous works on optimal metrology and optimal mode parameter estimation}\label{sec:PreviousLiterature}
Let us first review optimal quantum metrology for systems in a finite-dimensional Hilbert space. Firstly in Ref.~\cite{giovannetti2006quantum}, it was shown that for estimating a parameter $\lambda$ encoded by a unitary $\hat U_\lambda = e^{-i\lambda \hat G}$ with a finite-dimensional generator $\hat G = \sum_{n\in\{n_{\text{min}},\ldots,n_{\text{max}}\}} g_n \vert n\rangle\langle n\vert$, the optimal state takes the form $(\vert n_{\text{min}}\rangle + \vert n_{\text{max}}\rangle)/\sqrt{2}$, where $n_{\text{min}}$ and $n_{\text{max}}$ are the indices corresponding to the minimal and maximal eigenvalues $g_{\text{min}}$ and $g_{\text{max}}$ of the generator. This result does not directly apply to our mode parameter estimation task, since each mode is associated with an \emph{infinite}-dimensional Fock space. If one truncates the Fock space of each mode at a finite photon number $N_{\text{cut}}$, the optimal state according to Ref.~\cite{giovannetti2006quantum} takes the form $\sim\bigl[(\hat a_{n_{\text{min}}}^\dag)^{N_{\text{cut}}} + (\hat a_{n_{\text{max}}}^\dag)^{N_{\text{cut}}}\bigr]\vert 0\rangle$. This state achieves a QFI of $4\Delta g^2 N_S^2$, as shown in Appendix.~\ref{app:PreviousLiterature}. It is therefore variance-optimal, but lacks the mean contribution $\bar{g}^2 N_S^2$, rendering it suboptimal within our framework. Our approach thus offers a genuine advantage: by avoiding any truncation of the infinite Fock space, it yields the true ultimate precision limit for Gaussian states.

Secondly, we discuss Ref.~\cite{matsubara2019} on optimal Gaussian metrology for generic multimode interferometric circuits. This work is relevant because an interferometric circuit can be interpreted as a mode transformation. The authors maximize the QFI for a fixed interferometric circuit under a constraint on the average photon number. Their analysis yields an optimal QFI of
$
8 g_{\max}^{2} N_S^{2} + O(N_S),
$
where $g_{\max}$ is the largest eigenvalue of the generator matrix. This performance is achieved with a single-mode squeezed state and is consistent with our results.
Unlike the previously discussed Ref.~\cite{giovannetti2006quantum}, this approach accommodates the infinite-dimensional  Fock spaces associated with each bosonic mode. However, in many mode-parameter-estimation tasks, such as time and frequency estimation (Sec.~\ref{sec:TimeFrequencyEstimation}) or beam-displacement and beam-tilt estimation (Sec.~\ref{sec:BeamDisplacementAndTiltEstimation}), one must consider infinitely many modes. The method of Ref.~\cite{matsubara2019} requires truncating the mode space, which leads to optimal states that depend on the chosen cutoff, as can be seen in Ref.~\cite{he2025b}. Such cutoff dependence is undesirable.
By contrast, our approach yields optimal states that are independent of any mode-space truncation. Combined with our resource-based interpretation, our framework constitutes a clear advantage.

Third, we review the literature that explicitly addresses mode-parameter–estimation tasks. Ref.~\cite{gessner2023b} presents a general framework for designing quantum-enhanced protocols in this setting. Ref.~\cite{boeschoten2025} extended this framework to multiparameter estimation, and Ref.~\cite{sorelli2024} gave explicit formulas for Gaussian states in terms of symplectic eigenvalues. The authors demonstrate that achieving Heisenberg scaling requires populating an initial mode together with an additional mode that has nonzero overlap with the derivative of the initial mode.     While there exist many quantum-enhanced probe states that provide a quantum advantage over classical strategies through superior photon-number scaling, Ref.~\cite{gessner2023b} leaves open the question of how to meaningfully compare different quantum-enhanced states with one another and what is the optimal state. In this work, we resolve this question for the subclass of pure Gaussian states by introducing the resource parameters $\bar g$ and $\Delta g$.

Finally, Ref.~\cite{pinel2012ultimate} investigates optimal Gaussian probe states for mode-parameter estimation in the regime where most signal photons are allocated to the displacement of a single mode $\hat a(\lambda)$. A small amount of squeezing, satisfying $\sinh^{2}(r)\ll |\alpha|^{2}$, is distributed between this mode and its derivative mode $\partial_\lambda \hat a(\lambda)$. When the squeezing is equally distributed, the resulting state achieves a QFI of
$ 
\mathcal{F} = 4 e^{2r} (\bar{g}^{2} + \Delta g^{2}) N_{S}.
$
The QFI thus scales linearly with $N_S$ (\textit{i.e}., SQL scaling), but with a constant prefactor enhancement $e^{2r} \ll N_S$ over a purely coherent-state strategy. This state is therefore suboptimal within our framework. However, as discussed in Sec.~\ref{subsubsec:VarianceOptimal} and demonstrated in Sec.~\ref{subsec:VarianceOptimalStateTimeFrequency}, when half of the signal photons are allocated to squeezing the derivative mode, the state of Ref.~\cite{pinel2012ultimate} can be variance-optimal, provided certain conditions are satisfied, as we show in Appendix.~\ref{subsub:WithDisplacement}. To achieve full optimality, however, the only state we are aware of is the one in Eq.~\eqref{eq:CompletlyOptimalState}, which allocates the entire signal photon budget to squeezing.
We note that Ref.~\cite{pinel2012ultimate} neither introduced $\bar{g}$ and $\Delta g$ as resource parameters nor considered general Gaussian states, focusing instead on a specific subclass. Our framework applies to arbitrary Gaussian states and, consistent with their findings, confirms that two modes are generally required to achieve optimal performance.

\section{First application of optimal mode estimation: Time-frequency shift estimation}\label{sec:TimeFrequencyEstimation}

In the previous section, we focused on a general estimation protocol involving a finite, discrete set of bosonic modes. We now turn to intrinsically continuous modes—such as time–frequency and transverse position–momentum—which are first decomposed into an infinite set of discrete modes and then truncated so that the results of the previous section can be applied.

\subsection{The model and the generators}\label{sec:TimeFrequencyModelGenerators}
We now consider the estimation of time and frequency shifts of quantum light. These two paradigmatic examples will be analyzed in detail in order to clearly demonstrate the capabilities and versatility of our framework. Note that we work exclusively within the framework of single-parameter estimation.

The estimation of time and frequency shifts is of central importance across a wide range of applications. For instance, in radar and lidar systems, a time delay encodes information about the distance to a target, while a Doppler-induced frequency shift provides access to the target’s velocity.

In Sec.~\ref{sub:ModeTransform}, we introduced so-called cavity modes whose extent is the entire quantization volume. Here, we consider spatio-temporal modes that are defined in the $xy$ plane at $z=0$ for some time interval (in our case we assume this interval to be infinite).
Under the paraxial and quasi-monochromatic approximation \cite{shapiro2009} a single spatial mode of the positive-part of the EM field in the  $xy$ plane at $z=0$ may be written as 
\begin{align}\label{eq:PositivePart}
    \hat E(t) = \frac{1}{\sqrt{2\pi}} \int_{-\infty}^\infty\mathrm{d}\omega \, e^{-i\omega t} \hat A(\omega) ,
\end{align}
where $[\hat E(t),\hat E(t')] = \delta (t-t')$ and $[\hat A(\omega),\hat A^\dag (\omega')] = \delta (\omega-\omega')$ are the commutation relations of continuous mode operators $\hat E(t)$ and $\hat A(\omega)$.
$\hat{E}^\dagger(t)$ may be interpreted as the creation operator for a photon localized at time $t$ in the $z=0$ plane, whereas $\hat{A}^\dagger(\omega)$ creates a photon of frequency $\omega$ that is not localized in time. Although negative frequency is not physical, we have extended the integration range to minus infinity. In practice, every considered quantum state will have a compact spectrum far away the zero frequency, typically in the GHz or THz range.\\

Now, the mode transforms $\hat U_\tau = e^{-i\tau \hat G_\tau}$ and $\hat U_\nu = e^{-i\nu \hat G_\nu}$ with generators
\begin{align}
    \hat G_\tau &= \int_{-\infty}^\infty \mathrm{d}\omega\, \omega \hat{A}^\dag (\omega) \hat A(\omega) \label{eq:GeneratorTimeShfit}\\
    \hat G_\nu &= \int_{-\infty}^\infty \mathrm{d} t\, t \hat E^\dag (t) \hat E(t) \label{eq:GeneratorFrequencyShfit}
\end{align}
induce a time shift $\tau$ and a frequency shift $\nu$ in the input state $\vert \psi\rangle$. Note that both of the generators are written in their continuous diagonal basis. Such generators have been called time and frequency operators \cite{fabre_time-frequency_2022}.Such generators are often referred to as time and frequency operators \cite{fabre_time-frequency_2022}. Both possess a continuous spectrum, with $\omega, t \in \mathds{R}$, as reflected by the presence of integrals rather than discrete sums. Consequently, this system corresponds to one with infinitely many modes. However, our formalism for optimal mode parameter estimation was developed for a finite number $M$ of discrete modes. To bridge this gap, we discretize the continuous degrees of freedom and approximate the system as finite-dimensional by introducing mode truncations.

As we show in  Appendix.~\ref{app:subsec:TruncationProcedure}, this procedure leads to the discritized generators of finitely many modes
\begin{align}
    \hat G_\tau &\simeq \sum_{n=1}^M \omega_n \hat a_n^\dag \hat a_n \label{eq:GenTauDiscrete}\\
    \hat G_\nu &\simeq \sum_{n=1}^M t_n \hat e_n^\dag \hat e_n \label{eq:GenNuDiscrete} ,
\end{align}
where $\hat e_n$ and $\hat a_n$ are discretized versions of $\hat E(t)$ and $\hat A(\omega)$ with $[\hat a_n,\hat a_m^\dag]=[\hat e_n,\hat e_m^\dag]=\delta_{nm}$.
Here, $\omega_n = \omega_{\text{min}} +  (n-1)\delta \omega$ and $t_n = t_{\text{min}}+ (n-1)\delta t$ are the  eigenvalues of the discretized generator, where $\delta\omega$ and $\delta t$ are the bin widths in the discitized frequency and time domain.
This discretization constitutes a good approximation as long as the states we consider can be well approximated using this discrete mode basis.
This can always be achieved (at least theoretically) by simply choosing the binning small enough and choosing the intervals $[\omega_{\text{min}},\omega_{\text{max}}]$ and $[t_{\text{min}},t_{\text{max}}]$ large enough, where $\omega_{\text{max}}= \omega_M$ and $t_{\text{max}}= t_M$. For numerical calculations, the truncation choice may be guided by Refs.~\cite{gessner2023b,sorelli2024,boeschoten2025}, which showed that is sufficient to consider a subspace spanned by finitely many modes to calculate the QFI.

Knowing that we can comfortably switch from a continuous to a discrete description, we will present the following results in the continuous description. 

\subsection{Resources for time and frequency estimation}
Now, according to the formalism in Sec.~\ref{sec:Resources}, for a general input state, the resource signal photon number is given by
\begin{align}
    N_S = \int_{-\infty}^\infty\mathrm{d}t \,  \langle \psi\vert \hat E^\dag (t)\hat E(t)\vert\psi\rangle = \int_{-\infty}^\infty\mathrm{d}\omega \,  \langle \psi\vert \hat A^\dag (\omega)\hat A(\omega)\vert\psi\rangle 
\end{align}
and the resources for time estimation are mean frequency $\bar\omega$ and bandwidth $\Delta\omega^2$ of the input state:
\begin{align}
    \bar \omega &= \int_{-\infty}^\infty\mathrm{d}\omega \, \omega \frac{\langle \psi\vert \hat A^\dag (\omega)\hat A(\omega)\vert\psi\rangle}{N_S}, \\
    \Delta\omega^2 &= \int_{-\infty}^\infty\mathrm{d}\omega \, (\omega-\bar\omega)^2 \frac{\langle \psi\vert \hat A^\dag (\omega)\hat A(\omega)\vert\psi\rangle}{N_S} .
\end{align}
They are the mean and variance of the normalized spectral intensity $\langle \psi\vert \hat A^\dag (\omega)\hat A(\omega)\vert\psi\rangle/N_S$. The resources for frequency estimation are central time $\bar t$ and time duration $\Delta t$, defined by
\begin{align}
    \bar t &= \int_{-\infty}^\infty\mathrm{d}t \, t \frac{\langle \psi\vert \hat E^\dag (t)\hat E(t)\vert\psi\rangle}{N_S} \\
    \Delta t^2 &= \int_{-\infty}^\infty\mathrm{d}t \, (t- \bar t)^2 \frac{\langle \psi\vert \hat E^\dag (t)\hat E(t)\vert\psi\rangle}{N_S} .
\end{align}
They are the mean and variance of the normalized temporal intensity $\langle \psi\vert \hat E^\dag (t)\hat E(t)\vert\psi\rangle/N_S$.

\subsection{Multimode Gaussian states for time and frequency estimation}

\subsubsection{Coherent states for time and frequency estimation}\label{subsec:TimeFrequencyCoherent}
Let us first consider the bosonic coherent state \cite{blow_continuum_1990,combescure_coherent_2012, roux_wigner_2020}, which is the quantum state describing by a phase-stabilized laser that can be written both in the time and frequency domains
\begin{align}
    \vert\psi^{\text{coh}}\rangle &= \exp\left(+ \alpha\int_{-\infty}^\infty\mathrm{d}t \, \Psi(t)\hat E^\dag(t)-h.c. \right) \vert0\rangle \\
    &= \exp\left( +\alpha\int_{-\infty}^\infty\mathrm{d}\omega \, \tilde \Psi(\omega)\hat A^\dag(\omega)-h.c. \right) \vert0\rangle
\end{align}
with eigenvalue relation $\hat E(t) \vert\psi^{\text{coh}}\rangle = \Psi(t)\vert\psi^{\text{coh}}\rangle$ and $\hat A(\omega) \vert\psi^{\text{coh}}\rangle = \tilde \Psi(\omega)\vert\psi^{\text{coh}}\rangle$.  Recall that multi-mode displaced states can always be written as single-mode displaced states in an appropriate mode basis \cite{fabre2020modes}, and here, the  creation operator of that single displaced mode is  $\int_{-\infty}^\infty \mathrm{d}t \, \Psi (t) \hat E^\dag (t) = \int_{-\infty}^\infty \mathrm{d}\omega\,\tilde \Psi (\omega) \hat A^\dag(\omega)$, where $\Psi(t)$ is an arbitrary normalized  complex function. The complex $\alpha$ is the displacement parameter.
We can easily verify that $\hat U_\tau\vert \psi^{\text{coh}}\rangle$ shifts the signal in time $\Psi(t)\rightarrow \Psi(t +\tau)$ and $\hat U_\nu\vert \psi^{\text{coh}}\rangle$ causes a frequency shift  $\tilde \Psi(\omega)\rightarrow \tilde \Psi(\omega -\nu)$.

To make this more concrete, consider a Gaussian mode function $\Psi(t) = \left(\frac{1}{2\sigma_t^2\pi}\right)^{1/4} e^{-\frac{1}{2}\left(\frac{t-t_0}{\sqrt{2}\,\sigma_t}\right)^2} e^{-i\omega_0(t-t_0)}$, plotted in Fig.~\ref{fig:CoherentState}. The Gaussian is characterized by the parameters $t_0$, $\omega_0$, and $\sigma_t$, which are directly related to the resource parameters via $\bar{t} = t_0$, $\bar{\omega} = \omega_0$, $\Delta t^2 = \sigma_t^2$, and $\Delta\omega^2 = \sigma_\omega^2 = 1/(4\sigma_t^2)$. Furthermore, we have $N_S = |\alpha|^2$.

\begin{figure}[htb]
    \centering
    \includegraphics[width=0.48\textwidth]{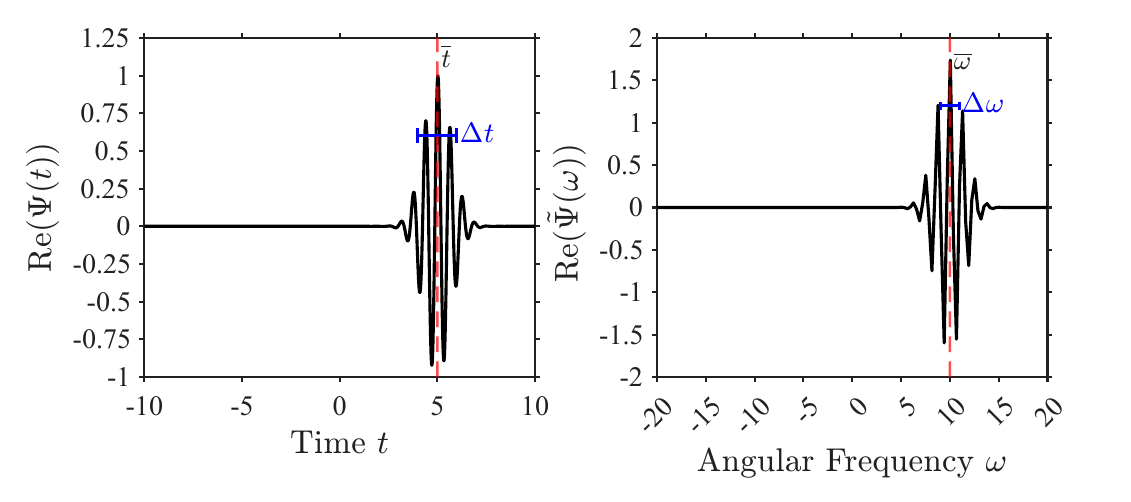}
    \caption{Real part of the coherent states signal $\Psi(t) = e^{-(t-t_{0})^{2}/\sigma_{t}^{2}}\,e^{i\omega_{0}t}$ (left) and its Fourier transform  $\tilde{\Psi}(\omega)$ (right). This figure illustrates a convenient representation of the time and frequency resources for classical and coherent states. The resources are the temporal and spectral bandwidths $\sigma_{t}$ and $\sigma_{\omega}$, and the central frequency and mean arrival time $\omega_{0}$ and $t_{0}$, here shown for a Fourier-limited pulse.}
    \label{fig:CoherentState} 
\end{figure}
Using Eq.~\eqref{eq:QFIcoherentGeneral}, we find the expression of the QFI for bosonic coherent states:
\begin{align}
    \mathcal{F}(\vert \psi_\tau^{\text{coh}}\rangle) &= 4 (\bar\omega^2+\Delta\omega^2) N_S \\
    \mathcal{F}(\vert \psi_\nu^{\text{coh}}\rangle) &= 4 (\bar t^2+\Delta t^2) N_S
\end{align}
which recovers well-known results from classical radar and lidar literature \cite{helstrom}, and are also presented in \cite{descamps_quantum_2023}. \\

In the radar and lidar literature, the quantities $\bar{\omega}$, $\Delta\omega$ and $\bar{ t}$, $\Delta t$ are well-established as the fundamental resources governing
the estimation of time delay and Doppler frequency shift, respectively~\cite{richards2005fundamentals,helstrom}. This decomposition traces back to Woodward's foundational analysis of the radar ambiguity function~\cite{Woodward1964}, which showed that the mean and spread of a
waveform's time-frequency support jointly determine the limits of range and velocity discrimination. In the quantum setting, these quantities enter the quantum Fisher information (QFI) as two structurally distinct contributions: a mean term ($4\bar{\omega}^2 N_S$ for time estimation, $4\bar{t}^2 N_S$ for frequency estimation) and a variance term ($4\Delta\omega^2 N_S$ for time estimation, $4\Delta t^2 N_S$ for frequency estimation)~\cite{giovannetti_quantum-enhanced_2004, zhuang_ultimate_2022}. These two terms are of different significance.\\

Consider time estimation, \textit{i.e}., target ranging. In any realistic narrowband waveform, the carrier frequency dominates the spectral spread: $\bar{\omega} \gg \Delta\omega$. Consequently, the mean term $4\bar{\omega}^2 N_S$ overwhelmingly dominates the total QFI. In principle, the Cram\'{e}r--Rao bound associated with this term is achievable, but only
via a phase-sensitive measurement, such as homodyne detection with a phase-locked local oscillator. Attaining this bound requires prior knowledge
of the time delay $\tau$ to within a phase cycle of the carrier,
$|\tau - \tau_{\mathrm{prior}}| \lesssim 1/\bar{\omega}$, imposing sub-wavelength accuracy on the prior. Such knowledge is unavailable in typical uncooperative-target scenarios and is only achievable in highly controlled interferometric setups.

By contrast, the variance term $4\Delta\omega^2 N_S$ is governed by the root-mean-square
bandwidth $\Delta\omega$ and accessible through incoherent or envelope-based
measurements which is precisely the time-of-flight paradigm underlying most practical ranging
systems~\cite{richards2005fundamentals}. Direct intensity detection, matched filtering of
the envelope, or photon counting all respond to the pulse shape rather than its absolute
carrier phase, and therefore require no phase prior. The Cram\'{e}r--Rao bound associated
with $4\Delta\omega^2 N_S$ is achievable with these prior-free
measurements and in radar waveform design the use of
high time-bandwidth product waveforms is motivated precisely by the need to maximize
$\Delta\omega$ without access to a phase reference~\cite{Woodward1964}.\\

Similar conclusions can be drawn for general Gaussian states. This observation highlights that the mean parameter $\bar g$ and the variance parameter $\Delta g$ play fundamentally different roles, motivating their treatment as distinct resources.

\subsubsection{Gaussian states with product structure in the Hermite-Gauss mode basis}\label{subsub:SuboptimalState}
Before discussing optimal states, let us examine some states that are commonly discussed in the literature — states that appear in product form in the Hermite–Gauss mode basis.

An arbitrary multimode squeezed vacuum state takes the form
\begin{align}\label{eq:squeezedStateTime}
\vert\psi^{\text{sq}}\rangle &= \exp\!\left(\frac{1}{2}\iint\mathrm{d}t \mathrm{d}t'\, f(t,t')\,\hat E^\dag(t)\hat E^\dag(t') - \mathrm{h.c.}\right)\vert 0\rangle \\
&= \exp\!\left(\frac{1}{2}\iint\mathrm{d}\omega\mathrm{d}\omega'\, \tilde{f}(\omega,\omega')\,\hat A^\dag(\omega)\hat A^\dag(\omega') - \mathrm{h.c.}\right)\vert 0\rangle, \label{eq:squeezedStateFrequency}
\end{align}
in the time and frequency domains, respectively. Here, $f(t,t')$ and $\tilde{f}(\omega,\omega')$ are the continuous analogues of the squeezing matrix in the time-domain and frequency-domain mode bases, respectively, and are known as the joint temporal amplitude (JTA) and joint spectral amplitude (JSA). Their modulus squares are known as the joint temporal intensity (JTI) and joint spectral intensity (JSI), respectively.

Now, similarly to the coherent state case, we can find a mode basis $\{\Psi_n \}_n$ in which the squeezed state takes on a simple (disentangled) form. 

Since   $f(t,t')$ is symmetric, it admits the expansion
\[
f(t,t')=\sum_{n=0}^{\infty} r_n\,\Psi_n(t)\Psi_n(t'),
\]
where  $r_n\ge   0$. 
Interpreting $f(t,t')$ as a bipartite wavefunction, this expansion corresponds to its Schmidt decomposition \cite{lamata_dealing_2005}. 
In finite dimensions, the analogous decomposition of a complex symmetric matrix is given by the Takagi factorization, see Appendix.~\ref{app:subsec:GaussianStateDisentangled}.
Note that in the frequency domain the Schmidt decomposition is $\tilde f(\omega,\omega ') =\sum_{n=0}^\infty r_n \tilde \Psi_n (\omega) \tilde \Psi_n (\omega')$ with $\tilde \Psi_n(\omega) = \int\mathrm{d}t \, \Psi_n(t) e^{+i\omega t}/\sqrt{2\pi}$. 

In this  basis $\{\Psi_n\}_n$, we can write 
\begin{align}\label{eq:SqueezedStateSchmidtModes}
    \vert\psi^{\text{sq}}\rangle = [\otimes_{n=1}^M \hat S_{\hat c_n} (r_n)]\vert 0\rangle .
\end{align}
with $ \hat c_n^\dag = \int\mathrm{d}t  \,\Psi_n (t)\hat E^\dag(t) $ being the mode basis in which the state appears disentangled. We  refer to $\{\hat c_n\}_n$ as the Schmidt basis.

To make it concrete, we consider the set of squeezed states that appear disentangled in the Hermite-Gauss (HG) mode basis
\begin{multline}
    \Psi_n (t) = \Phi_n (t;t_0,\omega_0,\sigma_t,\theta) 
      = \left( \frac{1}{2\sigma^2 \pi}\right)^{1/4} \sqrt{\frac{2^{-n}}{n!}} \\ \times H_n\left ( \frac{t-t_0}{\sqrt{2}\sigma_t} \right) e^{-\frac{1}{2}\left( \frac{t-t_0}{\sqrt{2}\sigma_t}\right)^2} e^{-i\omega_0 (t-t_0)} e^{-i\theta} \label{eq:DefinitionHGmode}.
\end{multline}
where $H_n (t) = (-1)^n e^{t^2} \frac{\mathrm{d}^n}{\mathrm{d}t^n} e^{-t^2}$ are the physicist's Hermite polynomials. The mode functions are centered around $t_0$ in the time domain and centered around $\omega_0$ in the frequency domain. Their width is  is characterized by $\sigma_t$ in the time domain and by $\sigma_\omega^2 = 1/(4\sigma_t^2)$  in the frequency domain. $\theta$ is a phase that will not impact what follows.

Let us now consider the QFI   for time estimation  of a state $  \vert\psi^{\text{sq}}\rangle = \otimes_n \hat S_{\hat c_n e^{-i\phi_n/2}} (r_n)\vert 0\rangle $, where we additionally introduced mode dependent squeezing angles $\phi_n$.  Let us assume the first two modes are equally populated, i.e., $r_0=r_1$ and $r_n = 0$ for $n\geq 2$ and that the squeezing angles satisfy $\phi_1-\phi_0 =\pi$. With this, we find in Appendix.~\ref{app:subsec:QFIofHGstate} the QFI for time estimation $\tau$: 
\begin{align}\label{eq:QFIHG}
   \mathcal{F}(\hat U_\tau\vert\psi ^{\text{sq}}\rangle) = (4 \bar \omega^2 + 2\Delta \omega^2 ) N_S^2 + O(N_S) .
\end{align}
It is evident that this specific state is suboptimal, as indicated by the prefactors of 4 and 2 preceding the mean and variance terms, respectively. Notably, both prefactors are only half of their maximal possible values. Since the Hermite–Gaussian modes are eigenfunctions of the Fourier transform, analogous results hold for frequency estimation.

\subsubsection{Optimal states for time and frequency estimation}\label{sec:OptimalStatesTimeFrequency}

We now turn to the problem of optimal multimode Gaussian states for time and frequency estimation. As before, we consider multimode squeezed vacuum states of the form given in Eq.~\eqref{eq:squeezedStateTime}, or equivalently in Eq.~\eqref{eq:squeezedStateFrequency}. In Sec.~\ref{subsec:OptimalStates}, we derived the optimal states analytically under the assumption that the generator possesses a quasi-continuous spectrum. Since this condition is satisfied for the discretized versions of the generators $\hat G_\tau$ and $\hat G_\nu$ in Eqs.~\eqref{eq:GenTauDiscrete} and \eqref{eq:GenNuDiscrete}, the general results of Sec.~\ref{subsec:OptimalStates} apply directly to the specific tasks of time and frequency estimation. As shown in Sec.~\ref{subsec:OptimalStates}, the optimal states take a particularly simple form in the basis $\{\hat b_n\}$ that diagonalizes the generator,
$\hat G = \sum_{n} g_n\, \hat b_n^\dag \hat b_n$,
in which the squeezing matrix $f$ is diagonal.

Let us focus on time estimation; the analysis for frequency estimation is entirely analogous by Fourier duality.

By analogy with the discrete case, we expect the squeezing matrix to be diagonal in the generator eigenbasis $\{\hat A(\omega)\}$, recalling that $\hat G_\tau = \int\mathrm{d}\omega\, \omega\, \hat A^\dag(\omega)\,\hat A(\omega).$
Accordingly, we make the following Ansatz for the squeezing function in the frequency domain:
\begin{align}
\tilde{f}(\omega,\omega') = \sum_{n\in\{i,j\}} r_n\, \delta(\omega-\omega_n)\,\delta(\omega'-\omega_n),
\end{align}
for some fixed modes $i, j \in \{1,2,\ldots,M\}$. However, this state is unphysical (non-normalizable) due to the presence of the Dirac delta functions. To obtain a physical state, we regularize the Dirac deltas by replacing them with Gaussians: $\delta(\omega-\omega_n) \to \tilde{\Phi}_0(\omega;\,t_n,\omega_n,\sigma_t,\theta)$, where $\tilde{\Phi}_0$ is the Fourier transform of the zeroth HG mode defined in Eq.~\eqref{eq:DefinitionHGmode}, which is a Gaussian in both the time and frequency domains.

We thus propose the following physical Ansatz, with squeezing matrix
\begin{align}
\tilde{f}(\omega,\omega') = \sum_{n\in\{i,j\}} r_n\,\tilde{\Phi}_0(\omega;\,t_n,\omega_n,\sigma_\omega,\theta)\,\tilde{\Phi}_0(\omega';\,t_n,\omega_n,\sigma_\omega,\theta). \label{eq:OptimalRegularizedState}
\end{align}
To achieve a good regularization, the width $\sigma_\omega$ of the Gaussians in the frequency domain should be much smaller than the separation between the two Gaussians, \textit{i.e}., $\sigma_\omega \ll |\omega_j - \omega_i|$. The remaining parameters $r_n$, $\omega_n$, and $t_n$ with $n \in \{i,j\}$ are to be chosen according to the type of optimality one wishes to attain.

Let us first consider a fully optimal state (see Sec.~\ref{subsec:OptimalStates} for the definition). In such a case, we have to choose the squeezing parameters $r_i,r_j$ and the central frequencies $\omega_i,\omega_j$ according to Eqs.~\eqref{eq:OptSqueezing1}-\eqref{eq:opteigenval2}.  The QFI of the corresponding state can be calculated by using the Schmidt decomposition of $\tilde f(\omega,\omega')$ given in Eq.~\eqref{eq:OptimalRegularizedState}. We show in Appendix. ~\ref{app:QFIregularized}  that the first two populated Schmidt modes $\tilde f(\omega,\omega')$ are approximately equal to $\tilde \Phi_0(\omega;t_i,\omega_i,\sigma_t,\theta)$ and $\tilde \Phi_0(\omega;t_j,\omega_j,\sigma_t,\theta)$. The QFI for temporal estimation can then be calculated, and we find
\begin{align}
    \mathcal{F}(\vert\psi_\tau\rangle) = 8\bar \omega^2 N_S^2 + 4 (\Delta \omega^2 - \sigma_\omega^2) N_S^2 + O(N_S) \label{eq:QFIofRegularizedStateTime}
\end{align}
which is optimal in the good regularization limit of $\sigma_\omega \ll \vert\omega_j-\omega_i\vert \sim \Delta\omega$. Note that when estimating a time parameter, we can freely  choose $t_i,t_j$ while $\omega_i,\omega_j$ have to be chosen according to Eqs.\eqref{eq:OptSqueezing1}-\eqref{eq:opteigenval2}. This freedom may be exploited for having a variance optimal state for the estimation of $\tau$ and $\nu$, as we later discuss in the next section.

\begin{figure*}[htbp]
    \centering
    \includegraphics[width=\linewidth]{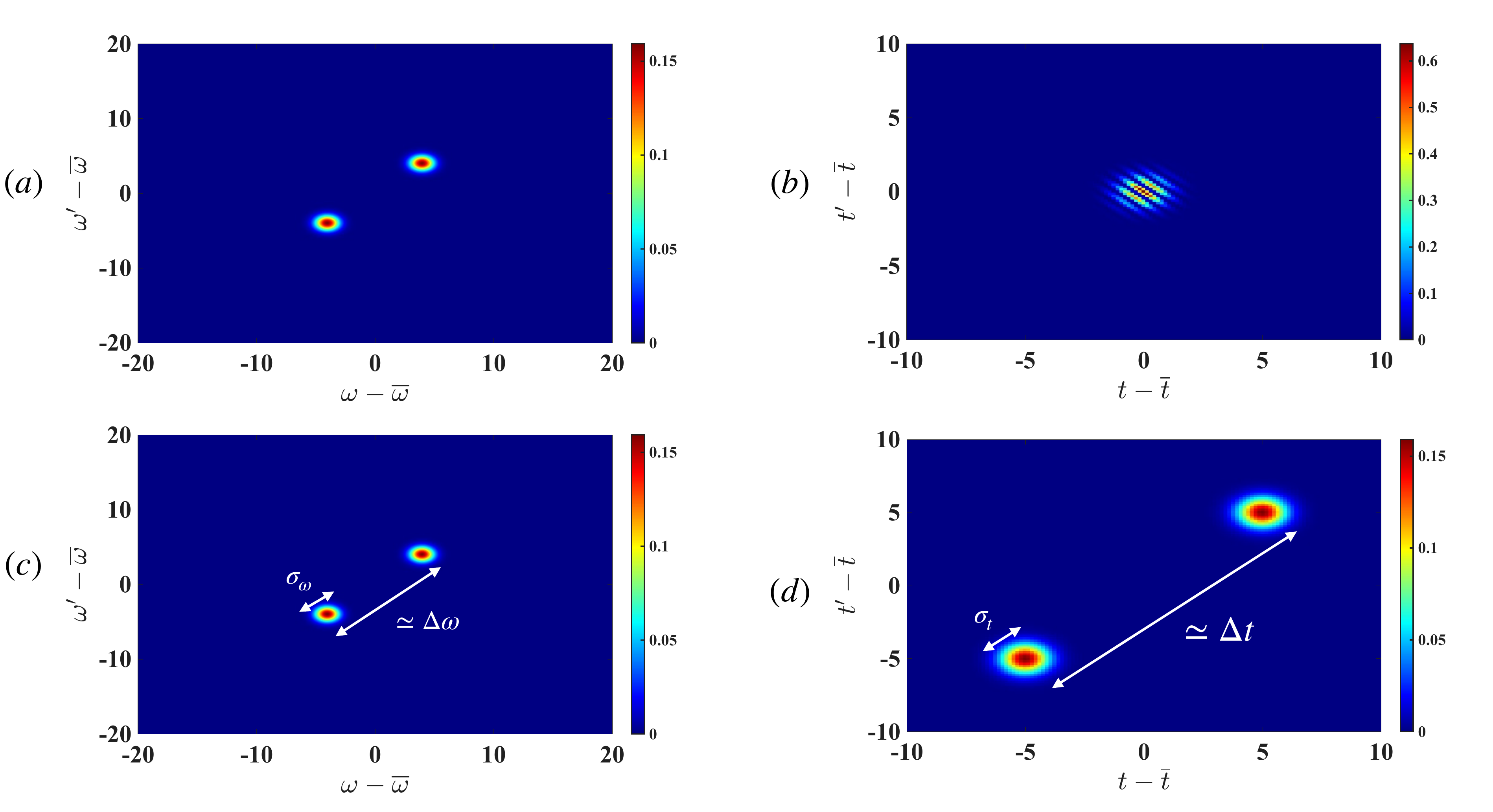}
    \caption{(a), (b) JSI and JTI for variance optimal states for time estimation, with equal squeezing parameters $r_{i}=r_{j}$, central time-of-arrival $t_i=t_j$, and central frequency such as $\vert\omega_j-\omega_i\vert\gg \sigma_\omega$. We remind that the resource $\Delta\omega$ is related to other parameters such as $\frac{\vert \omega_j-\omega_i\vert^2}{2} =  \Delta\omega^2-  \sigma_\omega^2\approx \Delta\omega^2$. (c), (d) JSI and JTI for simultaneous variance optimal state for time and frequency estimation, that corresponds to the parameters: equal squeezing parameter $r_{i}=r_{j}$,  with $\abs{t_{i}-t_{j}}\gg \sigma_{t}$, and  $\vert\omega_j-\omega_i\vert\gg \sigma_\omega$. We have chosen arbitrary central frequency and time $\overline{\omega},\overline{t}$ for the two distributions.}
    \label{fig:fplots1}
\end{figure*}

\subsubsection{Variance-optimal state for time and frequency estimation}\label{subsec:VarianceOptimalStateTimeFrequency}

We now focus on finding the variance optimal state. Let us consider first a time-shift estimation scenario. If we only strive for variance optimality, we choose an equal weighting of the squeezing parameters, \textit{i.e}.,  $r_i = r_j$, and the Gaussian centers in the frequency domain are chosen such that $\vert\omega_j-\omega_i\vert\gg \sigma_\omega$.
This choice yields $\frac{\vert \omega_j-\omega_i\vert^2}{2} =  \Delta\omega^2-  \sigma_\omega^2\approx \Delta\omega^2$ and $\frac{\omega_i+\omega_j}{2}= \bar \omega$. The state's QFI is 
\begin{equation}\label{QFItime}
    \mathcal{F}(\vert\psi_\tau\rangle) = 4\bar \omega^2 N_S^2   +   4(\Delta \omega^2 - \sigma_\omega^2) N_S^2 + O(N_S) ,
\end{equation}
and thus the state is  variance optimal in the good regularization limit $\Delta\omega^2-  \sigma_\omega^2\approx \Delta\omega^2$. As in the case discussed in the previous paragraph, the parameters $t_i,t_j$ can be freely chosen, and is not detrimental for temporal estimation. In the case of $t_i=t_j$, the state can only be optimal  for time estimation, but not for frequency estimation, as the squeezing matrix takes on diagonal structure only in the frequency domain and not the time domain. Similar observations for time-frequency shifts estimation with two-photon states were made in \cite{fabre_parameter_2021}. The JSI and JTI of the variance optimal state is represented in Fig.~\ref{fig:fplots1}, when $t_i=t_j$, and  $\vert\omega_j-\omega_i\vert\gg \sigma_\omega$.
If $\vert t_j-t_i\vert\gg \sigma_t$, the squeezing function takes on a diagonal structure  also in the temporal domain (JTI) and the state's QFI with respect to frequency estimation is
\begin{equation}\label{QFIfrequncy}
    \mathcal{F}(\vert\psi_\nu\rangle) = 4\bar t^2 N_S^2   +   4(\Delta t^2 - \sigma_t^2) N_S^2 + O(N_S) 
\end{equation}
and thus the state is  variance optimal in the good regularization limit $\Delta t^2-  \sigma_t^2\approx \Delta t^2$.
Thus, the state can be variance optimal for both time and frequency estimation. However, we are not aware of a measurement scheme that can jointly estimate both parameters.
We also point out that for specific resource parameter configurations it is possible to find states that are fully optimal for both time and frequency estimation.

The above states are two-mode squeezed states that do not make use of displacement. In Sec.~\ref{subsubsec:VarianceOptimal} we have discussed that it is also possible to find variance optimal states that make use of displacement. In fact, if we take the state regularized version of Eq.~\eqref{eq:VarianceOptimalDisplaced}  for time estimation,   with the energy being equally distributed between displacement and squeezing, we find
$
    \mathcal{F}_\tau = 2\bar \omega^2 N_S^2   +  4 (\Delta \omega^2 - \sigma_\omega^2) N_S^2 + O(N_S) .
$
That is, we have variance optimality, and the state's mean prefactor is reduced by a factor of $\frac12$ compared to the case in which both modes are squeezed equally.
We may also wonder if there are other possible mode choices that could attain variance optimality. We investigated a state being displaced in the first HG mode and squeezed in the second, and found that at best, the state achieves a QFI of $ 2\bar \omega^2 N_S^2   +  2 (\Delta \omega^2 - \sigma_\omega^2) N_S^2 + O(N_S)$, which is a factor $\frac12$ off of variance optimality. Another state based on high-order HG modes was discussed in Ref.~\cite{Reichert2026}.

\subsection{Optimal measurements for time and frequency estimation} \label{sec:OptimalMeasurementsTimeFrequency}
Let us discuss the measurements that can be deployed for time and frequency estimation using the optimal (or variance-optimal) regularized state with squeezing matrix given in Eq.~\eqref{eq:OptimalRegularizedState}. We focus on time estimation, from which the frequency estimation case follows directly by analogy.

\subsubsection{Phase-sensitive detection}
From Sec.~\ref{sec:OptimalMeasurements} and Appendix.~\ref{app:OptimalMeasurementVarianceCompletly}, phase-sensitive homodyne measurements in the mode basis in which the generator appears diagonally are optimal. In the case of a regularized state, the same analysis as in  Appendix.~\ref{app:OptimalMeasurementVarianceCompletly} can be used to show that homodyning the regularized modes $  \Phi_0(t+\tau_{\text{pr}};t_i,\omega_i,\sigma_t,\theta)$ and $  \Phi_0(t +\tau_{\text{pr}};t_j,\omega_j,\sigma_t,\theta)$ is optimal for time estimation.  The local oscillator (LO) modes must be shaped to match those of the two-mode squeezed vacuum state under investigation. For two-mode states with a significant frequency separation, a bichromatic local oscillator (BLO) — composed of two frequency components that match the signal modes — provides a critical practical advantage. This technique translates the two-mode squeezing information to a single, lower electronic frequency band, circumventing the need for high-bandwidth shot-noise-limited electronics and enabling the use of more sensitive, low-frequency detection systems. \cite{Marino:07}. Other schemes employ multipixel \cite{ferrini_compact_2013} or multiplexed \cite{eldan_multiplexed_2024} homodyne detection which necessitate shaping the temporal or spectral profile of a monochromatic LO to match an arbitrary mode can be achieved using a programmable waveshaper, as it was achieved for quantum frequency comb \cite{araujo_full_2014}. 

Finally, although phase-sensitive homodyne detection is fully optimal in principle for pure states, its implementation requires a substantial amount of prior information, owing to the need for precise tuning of the measurement apparatus and the presence of phase ambiguities. For instance, in gravitational wave interferometry, while the wave distorts spacetime transversally, optimal filtering and signal extraction depend critically on prior assumptions about source parameters to define the template bank and analysis bandwidth.

As discussed in Sec.~\ref{subsec:TimeFrequencyCoherent}, these practical lack of knowledge about the prior information often render the contribution to the QFI arising from the mean term, scaling as $\sim \bar g^{2} N_S^{2}$, effectively inaccessible. By contrast, for time-of-flight estimation in lidar systems, the contribution associated with the variance term, scaling as $\sim \Delta g^{2} N_S^{2}$, is not subject to these constraints. Consequently, without a prior estimate $\tau_{\text{pr}}$ sufficiently close to $\tau$, one must choose a relative phase of the homodyne detection $\phi$ suboptimally, leading to a direct and severe reduction in estimation precision.

\subsubsection{Phase-insensitive detection}
We now aim to identify variance-optimal measurements according to Definition~\ref{def:OptimalMeasurements}, namely measurements whose FI attains the contribution $\sim \Delta g^{2} N_S^{2}$ of the QFI optimally, without necessarily achieving the mean contribution $\bar g^2 N_S^2$ optimally. We identify a phase-insensitive measurement scheme, or direct detection, that is variance optimal if the state under consideration satisfies the variance-optimal measurement condition given in Appendix.~\ref{app:VarianceOptimalMeasurementCondition}. Specifically, we propose time-resolved and frequency-resolved photon counting as phase-insensitive measurements for time and frequency estimation, respectively. Mathematically, these measurements correspond to projections onto the states $(K!)^{-1/2}\bigotimes_{n=1}^{K} \hat E^\dagger(t_n)\lvert 0\rangle$ and $(K!)^{-1/2}\bigotimes_{n=1}^{K} \hat A^\dagger(\omega_n)\lvert 0\rangle$, respectively. We will focus on the sensitivity achievable with time-resolved photon counting for time estimation, but the discussion is similar for frequency estimation with frequency resolved photon counting. \\

In Appendix.~\ref{app:VarianceOptimalMeasurementForVarianceOptimalState}, we show that, for the variance-optimal state with squeezing matrix given in Eq.~\eqref{eq:OptimalRegularizedState} and parameters satisfying $r_i = r_j$, $\sigma_\omega\ll \vert\omega_j-\omega_i\vert$,
 time-resolved photon counting yields a FI for time estimation of
 \begin{equation}
  F_{\text{direct}}(\vert\psi_\tau\rangle)=  4 \Delta \omega^{2} N_S^{2} + O(N_S) 
 \end{equation}
when $t_i=t_j$. The condition of the Gaussian centers in the time domain being equal, \textit{i.e}., $t_i=t_j$, is needed so that the variance-optimal state satisfies the variance-optimal measurement condition given in Appendix.~\ref{app:VarianceOptimalMeasurementCondition}. 
Notably, the mean contribution $\sim \bar \omega^{2} N_S^{2}$ is absent, and only the variance term contributes. While this may appear to represent a substantial reduction in precision, in many realistic scenarios the  precision associated to the variance term is the only component that can be robustly achieved, as we have mentioned in Sec.~\ref{subsec:TimeFrequencyCoherent}.

Furthermore, we show in Appendix.~\ref{app:VarianceOptimalMeasurementForHGState} that the state  of the form given in Eq.~\eqref{eq:SqueezedStateSchmidtModes}, whose squeezing matrix appears diagonally in the HG mode basis, satisfies the variance-optimal measurement condition in Appendix.~\ref{app:VarianceOptimalMeasurementCondition}, and thus, time-resolved photon counting attains an FI of $2 \Delta \omega^{2} N_S^{2} + O(N_S)$.
 This establishes time-resolved photon counting as a variance-optimal measurement for a broad class of pure Gaussian quantum states.\\

In practice, this detection scheme relies on the implementation of photon counters. Photon-number-resolving~(PNR) detectors directly measure the photon-count distribution of the optical field incident upon them, giving access to the full statistics of the probe state at the detection stage. The practical construction of PNR detectors has followed two main routes, both achieving photon-number resolution of up to 100 photons: multiplexing multiple SNSPD detectors so that simultaneous arrival events are spatially or temporally separated~\cite{cheng_unveiling_2022}, or employing transition-edge sensors~(TES) \cite{eaton_resolution_2023}, which exhibit a timing jitter of order nanoseconds. By contrast, the sub-picosecond timing jitter achieved by optimised single-pixel SNSPDs---demonstrated at $2.6$~ps system jitter in ~\cite{korzh_demonstration_2020}. Therefore, the use of two-mode squeezed states of light as probes for measuring time and frequency shifts—combined with photon-number-resolving detectors with a resolution of 100 photons and low timing jitter—enables performance beyond the shot-noise scaling and is accessible with current technology.

\subsection{Literature on time-frequency parameter estimation}

A substantial body of literature is dedicated to quantum-enhanced estimation of time delays and frequency shifts, exploiting non-classical states of light including two-photon and squeezed states. These protocols are of particular relevance to quantum sensing applications such as radar and lidar. In these contexts, an estimated temporal delay directly corresponds to the range of a target, while a measured frequency shift, induced by the Doppler effect, provides information about its relative velocity.

\subsubsection{Two-photon states for time and frequency shifts estimation}

Two-photon states, which are non‑Gaussian and typically exhibit very low squeezing, offer valuable insights into the spectral and temporal engineering required for quantum metrological protocols.  Two‑photon states provide a clear advantage for probing sensitive samples \cite{nasr_quantum_2009,taylor_quantum_2016}. The measurement of time delays, classically performed using Optical Coherence Tomography– a Michelson interferometer – found its quantum counterpart in the Hong‑Ou‑Mandel (HOM) experiment. The HOM interferometer demonstrated the ability to measure sub‑picosecond time intervals \cite{PhysRevLett.59.2044}, laying the foundation for Quantum Optical Coherence Tomography \cite{nasr_demonstration_2003}. In \cite{lyons_attosecond-resolution_2018}  formally applied quantum metrology tools to the HOM interferometer, demonstrating  ten of attosecond‑scale resolution with entangled photon pairs. Subsequent spectral engineering of the biphoton state further enhanced precision. For instance, using a two‑mode frequency‑entangled state led to a factor‑of‑two improvement in resolution \cite{chen_hong-ou-mandel_2019} for a fixed total bandwidth. Theoretical work then showed that engineered time-frequency grid states could push resolution to the sub‑attosecond regime \cite{fabre_parameter_2021}.
The estimation of frequency shifts has been proposed in \cite{fabre_parameter_2021}
using generalized HOM interferometry with spectrally engineered two‑photon states. The framework was further generalized by \cite{descamps_time-frequency_2023}, who established the ultimate bounds for estimating a broader class of time‑frequency parameters, such as shear and rotation in chronocyclic phase space.\\

These results connect naturally to the findings of the present manuscript. In particular, the factor-of-four improvement in time-delay resolution achieved with a two-color frequency-entangled biphoton state \cite{chen_hong-ou-mandel_2019,fabre_parameter_2021} has a direct analogue in our framework: the same enhancement arises for two-mode squeezed vacuum states with the same spectral separation. This highlights that spectral engineering yields consistent metrological gains across different photon-number statistics — whether in the low-flux two-photon regime or the bright squeezed-light regime.

\subsubsection{Multimode squeezed states for time-frequency estimation}
The time or frequency accuracy enhancement achieved through spectral engineering of two-photon states, for a fixed amount of resources, is also demonstrably attainable using multimode squeezed states, as shown in this manuscript. This could be understood as two-mode squeezed state as a "Gaussification" of two-photon (two-mode) states, but such an explicit transformation will be let for further work. Let us now exhibit related works, and additional challenges to implement our idea. \\

The Gaussian state considered in Ref.~\cite{reichert24}, for example, is similar to the one discussed in Sec.~\ref{subsub:SuboptimalState} and is suboptimal: its variance-term prefactor is only "1" for both time and frequency estimation, which is below the maximal value of "4".  Ref.~\cite{reichert_quantum-enhanced_2022} studied an idler-assisted state with a Gaussian-shaped squeezing function, whose performance depends strongly on the specific form of the squeezing function and is likewise suboptimal in all cases. 

Strong coherent states with a small amount of squeezing in the derivative mode, that were shown to be optimal in  the limit $\vert \alpha\vert\gg \sinh^2 (r)$ Ref.~\cite{pinel2012ultimate}, were also considered theoretically for time estimation \cite{lamine2008quantum} and experimentally for frequency estimation \cite{cai2021quantum}. According to our framework, these approaches are suboptimal due to the linear scaling with signal photon number, but can be made variance optimal as shown in Sec.~\ref{subsec:VarianceOptimalStateTimeFrequency} by allocating half of the energy to squeezing. Another scenario where squeezing in the derivative mode is beneficial is the estimation of the temporal separation between two coherent or thermal pulses, which can be described in terms of symmetric and anti-symmetric temporal modes~\cite{sorelli2024}. In the regime of large temporal separations, where the pulse overlap is negligible, the estimation becomes primarily sensitive to variations in the mode shapes; accordingly, enhanced performance can be achieved by introducing squeezing in the derivatives of these symmetric and anti-symmetric modes.

A complementary perspective is provided by the bosonic metrology framework of Ref.~\cite{saharyan_resources_2025}, which disentangles the contributions of particle-number statistics and modal entanglement to the quantum Fisher information by making the phase reference explicit. In the continuous-variable regime, precision enhancement arises from quadrature squeezing of the appropriate collective mode, and for uncorrelated multimode probes, squeezing is the \emph{sole} source of quantum advantage — mode entanglement introduced by beam splitters provides no additional gain.

We should precise finally that achieving a resolution enhancement with squeezed states, as compared to classical states, requires a substantial degree of squeezing. In practice, this is challenging to realize due to multiple sources of noise, including injection losses, thermal effects, and other decoherence mechanisms. These factors, particularly photon losses over extended distances, rapidly degrade squeezing and induce thermalization, thereby limiting the feasibility of our current proposal for long-range applications such as radar and lidar. 

We provide now a simple estimate of the sensing distance. In such experimental settings, the total homodyne detection efficiency—accounting for photodiode quantum efficiencies, mode imbalance, optical losses, and equivalent electronic noise—can reach values as high as approximately 90\% (see for instance \cite{wang_provably-secure_2023}). Using a channel transmission under the form $\eta=10^{-\alpha L/10}$ where $L$ is the propagation length, and $\alpha\sim 0.2$ dB/km in optical fibers at telecom wavelength 1550nm, the total transmission is limited to $\sim 12$km (equivalent to $\eta>1/2$) for beating coherent states estimation for instance (see Eq.~(\ref{homodyne})).

\section{Second application of optimal mode estimation: Beam displacement and tilt estimation}\label{sec:BeamDisplacementAndTiltEstimation}

A further illustrative example is the estimation of beam displacement and tilt, as represented in Fig.~\ref{beamdisplacement}. Although its mathematical formulation is essentially equivalent to that of time and frequency estimation — and therefore does not warrant a repetition of the analysis presented in Sec.~\ref{sec:TimeFrequencyEstimation}, we emphasize that our framework offers a new perspective on this problem when compared to previous literature.

\subsection{Beam displacement estimation}
Let us consider a monochromatic aperture field $\hat{E}(x,y)\,e^{-i\omega t}$ of angular frequency $\omega$ in the detection plane $z=0$, which can be decomposed into spatial plane-wave modes as
\begin{align}
    \hat{E}(x,y) = \frac{1}{2\pi}\iint_{k_x^2+k_y^2<\omega^2}\mathrm{d}k_x\,\mathrm{d}k_y\, \hat{A}(k_x,k_y)\,e^{-i(k_x x + k_y y)}. \label{eq:SpatialField}
\end{align}
In the paraxial regime, the mode operators $\hat{E}(x,y)$ and $\hat{A}(k_x,k_y)$ in the spatial and spectral domains satisfy the commutation relations $[\hat{E}(x,y),\hat{E}^\dag(x',y')] = \delta(x-x')\delta(y-y')$ and $[\hat{A}(k_x,k_y),\hat{A}^\dag(k_x',k_y')] = \delta(k_x-k_x')\delta(k_y-k_y')$, and the integration domain in Eq.~\eqref{eq:SpatialField} may be extended to $\mathbb{R}^2$. A beam displacement is a transverse displacement $d$ of the beam in the aperture plane, corresponding to $\hat{E}(x,y) \to \hat{E}(x+d,y)$, which induces the transformation $\hat{A}(k_x,k_y) \to \hat{A}(k_x,k_y)\,e^{-ik_x d}$ on the spectral mode operators.

From this, one can readily deduce the unitary operator $\hat{U}_d = e^{-id\hat{G}_d}$ of the beam displacement transformation, with generator
\begin{align}
    \hat{G}_d = \iint_{\mathbb{R}^2}\mathrm{d}k_x\,\mathrm{d}k_y\, k_x\,\hat{A}^\dag(k_x,k_y)\,\hat{A}(k_x,k_y).
\end{align}
This generator is already expressed in diagonal form; in an arbitrary basis, one would require four integrals (or sums in the discrete case) instead of two.

Small tilts of the beam by an angle $\vartheta$ in the $xz$ plane about the origin correspond to the transformation $\hat{E}(x,y) \to \hat{E}(x,y)\,e^{-i\vartheta\frac{\omega}{c}x}$, from which one deduces the unitary operator $\hat{U}_\vartheta = e^{-i\vartheta\hat{G}_\vartheta}$ of the tilt transformation, with generator
\begin{align}
    \hat{G}_\vartheta = \frac{\omega}{c}\iint_{\mathbb{R}^2}\mathrm{d}x\,\mathrm{d}y\, x\,\hat{E}^\dag(x,y)\,\hat{E}(x,y),
\end{align}
which is likewise already expressed in the basis in which it takes diagonal form.

\begin{figure}[htbp]
\centering
\begin{tikzpicture}
  \node[inner sep=0] (img) {\includegraphics[width=0.8\linewidth]{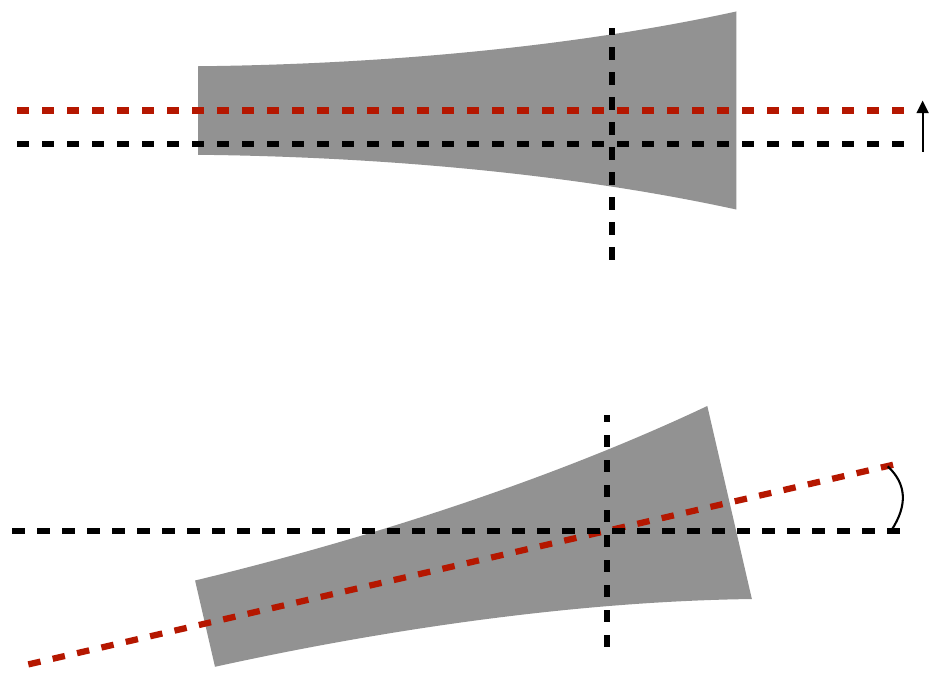}};
  \node at (-3.1,2.1) {$a)$};
  \node at (-3.1,-1.1) {$b)$};
  \node at (3.55,1.6) {$d$};
   \node at (3,-1.2) {$\vartheta$};
   \node at (1.2,0.1) {detection plane};
   \node at (-2.5,0.1) {reference axis};
   \draw[->] (-2.7,0.3) -- (-2.7,+1.4);
   \draw[->] (-2.7,-0.2) -- (-2.7,-1.3);
\end{tikzpicture}
\caption{\label{beamdisplacement}Two mode transforms are illustrated: $a)$ beam displacement, where the beam is shifted by a distance $d$ in   direction transverse to the detection plane, and $b)$, beam tilt, where the beam is rotated by an angle $\vartheta$ around an axis perpendicular to the reference axis. This schematic is widely inspired from \cite{delaubert2006quantum}. }
\end{figure}
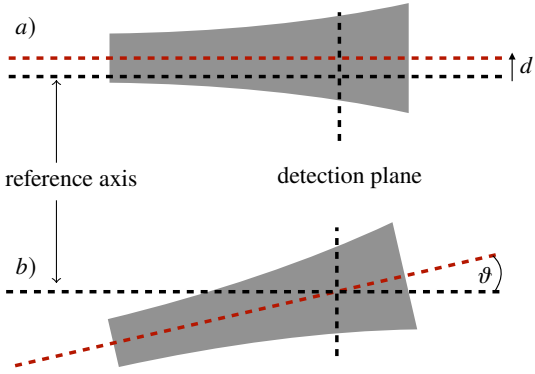

\subsection{Resources and optimal states for beam displacement and tilt estimation}
Let us now discuss the resources for beam displacement and tilt estimation. The common resource is the signal photon number
\begin{align}
    N_S &= \iint_{\mathbb{R}^2}\mathrm{d}k_x\,\mathrm{d}k_y\, \langle\psi\vert\hat{A}^\dag(k_x,k_y)\,\hat{A}(k_x,k_y)\vert\psi\rangle \\
    &= \iint_{\mathbb{R}^2}\mathrm{d}x\,\mathrm{d}y\, \langle\psi\vert\hat{E}^\dag(x,y)\,\hat{E}(x,y)\vert\psi\rangle.
\end{align}
We define $I^{\text{spec}}(k_x,k_y) = \langle\psi\vert\hat{A}^\dag(k_x,k_y)\,\hat{A}(k_x,k_y)\vert\psi\rangle/N_S$ as the normalized spectral photon-number intensity, and $I^{\text{spat}}(x,y) = \langle\psi\vert\hat{E}^\dag(x,y)\,\hat{E}(x,y)\vert\psi\rangle/N_S$ as the normalized spatial photon-number intensity. The resources for beam displacement estimation are then
\begin{align}
   \bar{k}_x &= \iint_{\mathbb{R}^2}\mathrm{d}k_x\,\mathrm{d}k_y\, k_x\, I^{\text{spec}}(k_x,k_y),\\
   \Delta k_x^2 &= \iint_{\mathbb{R}^2}\mathrm{d}k_x\,\mathrm{d}k_y\, (k_x-\bar{k}_x)^2\, I^{\text{spec}}(k_x,k_y),
\end{align}
corresponding to the mean and variance of the transverse wave vector in the $x$ direction. The resources for beam tilt estimation are
\begin{align}
   \bar{x}\,\frac{\omega}{c} &= \iint_{\mathbb{R}^2}\mathrm{d}x\,\mathrm{d}y\, x\, I^{\text{spat}}(x,y),\\
   \Delta x^2\,\frac{\omega^2}{c^2} &= \iint_{\mathbb{R}^2}\mathrm{d}x\,\mathrm{d}y\, (x-\bar{x})^2\, I^{\text{spat}}(x,y),
\end{align}
corresponding to the mean and variance of the spatial intensity distribution in the $x$ direction, each multiplied by the factor $\omega/c$.

Applying the results of Sec.~\ref{sec:ModeParamterEstimationGaussian}, the optimal QFI for displacement and tilt estimation is $(8\bar{k}_x^2 + 4\Delta k_x^2)N_S^2 + O(N_S)$ and $(\omega^2/c^2)(8\bar{x}^2 + 4\Delta x^2)N_S^2 + O(N_S)$, respectively. Completely analogously to Sec.~\ref{sec:OptimalStatesTimeFrequency}, an optimal state can be constructed by squeezing two regularized modes of the basis in which the generator is diagonal. For displacement estimation one may, for example, appropriately squeeze regularized versions of $\hat{A}(k_{x_i},0)$ and $\hat{A}(k_{x_j},0)$, and for tilt estimation, $\hat{E}(x_i,0)$ and $\hat{E}(x_j,0)$. The discussion of phase-sensitive and phase-insensitive measurements is entirely analogous to that of Sec.~\ref{sec:OptimalMeasurementsTimeFrequency}. The phase-insensitive measurements that achieve variance-optimal performance are spatially resolved photon counting for beam displacement estimation, and $k_x$-$k_y$-resolved photon counting for beam tilt estimation.

\subsection{Comparison with related works}
The recent work of Ref.~\cite{he2025b} also considered the problem of beam displacement sensing  using a quantum metrology framework. Their approach was to describe the problem in the spatial Hermite-Gaussian (HG) mode basis, compute the generator matrix in this basis, and evaluate the 
QFI via a formula analogous to Eq.~\eqref{eq:QFIcomponents} for a general single-mode squeezed state. Because the problem is intrinsically infinite-dimensional, they introduce a mode cutoff and then maximize the QFI subject to a photon-number constraint, without invoking additional 
resources. They find that the optimal state is related to the basis in which the truncated generator matrix appears diagonal, a result similar to our manuscript. Specifically, their optimal state is a single-mode squeezed state squeezed in the mode belonging to the generator-diagonalizing mode basis corresponding to the largest generator eigenvalue $k_{x,\text{max}}$, achieving a QFI of $8k_{x,\text{max}}^2 N_S^2 + O(N_S)$. This is consistent with our framework, as this state corresponds to the mean-optimal state discussed in Sec.~\ref{subsub:MeanOptimal}. The authors have also found that the phase-sensitive homodyne measurement of the populated mode is optimal \cite{he2025b}.

Our approach extends this analysis by introducing two additional resources, $\bar{k}_x$ and $\Delta k_x$, alongside $N_S$. Although we likewise truncate the mode space, a key advantage of our framework is that the optimal states we identify are \emph{independent} of the mode 
cutoff. More importantly, our analysis reveals that the state found in Ref.~\cite{he2025b} is not the only optimal state but belongs to a broader family: any state squeezed in a single arbitrary mode makes optimal use of the mean resource $\bar{k}_x$, with the prefactor 
of the mean term in the QFI universally equal to~$8$. Furthermore, by appropriately squeezing two modes in the generator-diagonalizing basis, one can also saturate the variance resource $\Delta k_x$. A particularly attractive special case is the variance-optimal state with 
$\bar{k}_x = 0$, for which phase-insensitive spatially-resolved photon counting achieves the optimal Fisher information $4\Delta k_x^2 N_S^2 + O(N_S)$, as shown in Sec.~\ref{sec:OptimalMeasurementsTimeFrequency}. Such a phase-insensitive measurement strategy may be experimentally more accessible than the  phase-sensitive homodyne detection proposed in~\cite{he2025b}: intensity measurements on HG modes via spatial mode demultiplexing with a multi-plane light conversion system~\cite{delaubert2006tem, boucher2020spatial, rouviere_ultra-sensitive_2024} already constitute a well-established practical toolbox that could be directly repurposed as a near-optimal platform for beam displacement sensing with non-classical states.

There are various related experimental works \cite{delaubert2006tem,delaubert2006quantum,treps2005quantum,boucher2020spatial} on beam displacement and tilt estimation using strong coherent states with squeezing in its derivative mode. These are examples of the general framework in Ref.~\cite{pinel2012ultimate} discussed in Sec.~\ref{sec:PreviousLiterature}. These approaches are practical due to the use of coherent light and homodyne measurements, which are readily available. However, according to our framework, they are suboptimal due to their linear scaling with photon number. However, they can be made variance optimal as shown in Sec.~\ref{subsec:VarianceOptimalStateTimeFrequency} by allocating half of the energy to squeezing. \\

Finally, this beam-displacement setting provides also a natural physical illustration of the results of Sec.~\ref{subsec:OptimalStates}: a relevant example is super-resolution imaging, where no scaling improvement over the SQL is possible whenever the initially populated mode is orthogonal to its own derivative [19]. In our framework this condition is exactly $\bar g=0$, since $\bar g N_S=\langle\psi|\hat G|\psi\rangle$ and $\langle\psi|\partial_\lambda\psi\rangle=-i\bar g$. For a symmetric input mode such as the fundamental Gaussian HG$_0$ — typical of a diffraction-limited point-spread function — the mean transverse wavevector vanishes by symmetry, $\bar k_x=\bar g=0$, so the derivative of HG$_0$ is precisely the orthogonal HG$_1$ mode. Recovering Heisenberg scaling requires accessing the $\Delta g^2$ channel via the two-mode squeezed construction of Sec.~\ref{subsec:OptimalStates}, which correlates HG$_0$ with the derivative mode HG$_1$ — reproducing, within our resource framework, the well-known necessity of accessing modes beyond the direct image to beat the Rayleigh limit in super-resolution imaging \cite{tsang_quantum_2016,gessner2023b}.

\section{Outlook and Discussion}\label{conclusion}

In this paper, we have established that, for pure Gaussian states, the QFI of any mode parameter estimation protocol is upper-bounded by a function of three resource parameters: the mean and variance of the generator intensity distribution over the natural mode basis induced by the transformation, denoted $\bar{g}$ and $\Delta g$, and the mean signal photon number $N_S$. We have shown that an appropriately configured two-mode squeezed vacuum state is optimal, in the sense that its QFI saturates this upper bound for any fixed values of $\bar{g}$, $\Delta g$, and $N_S$. Crucially, we have further demonstrated that multimode homodyne detection on these two modes constitutes an optimal measurement, in the sense that it achieves the QCRB and thus extracts all the information available in the optimal probe state. This resource formulation applies across estimation tasks: we have demonstrated its optimality explicitly for time delays and frequency shifts, and for tilt  and beam displacement. We have summarized the sensitivities that can be reached with different quantum states in Table~\ref{metrology}, that provide a complete characterization of optimal quantum strategies for mode parameter estimation.

\begin{table}[h]
\centering
\begin{tabular}{l  c  r}
\hline\hline
\textbf{State} & \textbf{QFI} $\mathcal F(\vert\psi_\lambda\rangle)$   \\ 
\hline
Coherent state $[\otimes_n\hat D_{\hat a_n}(\alpha_n)]\vert 0\rangle$ 
    &  $(4 \bar g^2 + 4 \Delta g^2) N_S$ \\
    Single-mode squeezed state $\hat S_{\hat c_i}\vert 0\rangle$ & $(8\bar g^2 +0\Delta g^2)N_S^2 + O(N_S)$ \\ 
    Variance-optimal state  Eq.~\eqref{eq:VarianceOptimalDisplaced} & $(2\bar g^2  + 4\Delta g^2) N_S^2 + O(N_S)$ \\
    Variance-optimal state  Eq.~\eqref{eq:VarianceOptimalState} & $(4\bar g^2  + 4\Delta g^2)
    N_S^2 + O(N_S)$ \\
    Optimal state in Eq.~\eqref{eq:CompletlyOptimalState} & $(8\bar g^2  + 4\Delta g^2) N_S^2 + O(N_S)$ \\
    
\hline\hline
\end{tabular}
\caption{\label{metrology}Quantum Fisher Information (QFI) for mode parameter estimation with various quantum states. The definition of the various types of optimality are given in Sec.~\ref{def:StateOptimality}.}
\end{table}

As perspectives, performing mode estimation protocol, at constant level resources, for non-Gaussian states, is a natural step as the investigation started for $N$-frequency entangled single photons \cite{descamps_quantum_2023} and multimode photon substracted state \cite{lopetegui_detection_2025}. The difficulty that seems to arise however, when dealing with considering the modal structure of non-Gaussian states, the modal and particle-number statistics can be intertwined \cite{descamps_quantum_2023}, which renders more difficult the resource comparison. On the other hand, from a foundational perspective, an important question would concern the role of superselection rules in setting the ultimate limits of precision \cite{saharyan_resources_2025}. Besides, the QCR bounds may not be adapted for non-Gaussian statistics in the non-
asymptotic regime \cite{montenegro_review_2025}, and other bounds were developed \cite{PhysRevLett.108.230401}.  Then, the other perspectives are to determine the optimal strategies for the loss and noise regime encountered in realistic lidar and radar scenarios \cite{giovannetti2001quantum,giovannetti2002positioning,shapiro2007quantum,maccone2023gaussian,huang2021quantum,zhuang2017entanglement,zhuang_ultimate_2022,reichert2023quantum,wei2024quantum}, where thermal noise is expected to degrade mode parameter estimation below the Heisenberg limit while potentially preserving a constant-factor quantum advantage over the shot-noise scaling. Note that, in these scenarios, the modal resource parameters also appear naturally, suggesting that our model's resource parameters play a role that extends beyond the idealized setting.
Finally, our formalism can be applied across a wide range of experimental scenarios, from estimating parameter-shift mode transformations like time and frequency shifts to analyzing nonlinear phase effects that arise in dispersive propagation, diffraction, or even reflection from accelerating or time-varying surfaces.

\section*{Acknowledgment}
M.~R. acknowledges insightful discussions with Francesco Albarelli. M.~R. acknowledges support from UPV/EHU Ph.D. Grant No. PIF21/289. M.~S. and M.~R. acknowledge support from HORIZON-CL4-2022-QUANTUM01-SGA project 101113946 OpenSuperQ-Plus100 of the EU Flagship on Quantum Technologies, the Spanish Ramon y Cajal Grant RYC-2020-030503-I, and the “Generaci\'on de Conocimiento” project Grant No. PID2024-156808NB-I00 funded by MI-CIU/AEI/10.13039/501100011033, by “ERDF Invest in your Future” and by FEDER EU. We also acknowledge support from the Elkartek project KUBIBIT -kuantikaren berrikuntzarako ibilbide teknologikoak (ELKARTEK25/79). This work has also been partially supported by the Ministry for Digital Transformation and the Civil Service of the Spanish Government through the QUANTUM ENIA project call – Quantum Spain project, and by the European Union through the Recovery, Transformation and Resilience Plan – NextGenerationEU within the framework of the Digital Spain 2026 Agenda.

\section*{Disclosures}

The authors declare no conflicts of interest.

\section*{Data availability}

Data underlying the results presented in this paper are not publicly available at this time but may be obtained from the authors upon reasonable request.

\bibliography{bibliosqueezed}

\newpage
\appendix 
\onecolumngrid   

\section{Mode transformation}\label{app:modetransformation}

\subsection{The mode transform and different bases}

Let us summarize the different mode bases used in this manuscript and the most important relations for the calculations that follow.

Throughout the manuscript, $\{\hat a_n\}_n$ denotes an arbitrary mode basis. The mode basis $\{\hat c_n\}_n$, defined by $\hat c_n = \hat V \hat a_n\hat V^\dag$, is the special basis in which the squeezing matrix — and correspondingly the covariance matrix (Appendix~\ref{app:CovarianceMatrix}) — is diagonal. Equivalently, a multimode squeezed vacuum state expressed in this basis factorizes as a tensor product of single-mode squeezed vacuum states.

The action of $\hat V$ on $\{\hat a_n\}_n$ is given by the transformation rules:
\begin{align}
    \hat V^\dag \hat a_n \hat V &= \sum_{i=1}^M V_{ni}\,\hat a_i, \label{eq:PassiveTransformation1}\\
    \hat V^\dag \hat a_n^\dag \hat V &= \sum_{i=1}^M V_{ni}^*\,\hat a_i^\dag, \label{eq:PassiveTransformation2}\\
    \hat V\,\hat a_j \hat V^\dag &= \sum_{n=1}^M (V^\dag)_{jn}\, \hat a_n = \sum_{n=1}^M V_{nj}^*\, \hat a_n, \label{eq:PassiveTransformation1Schrodinger}\\
    \hat V\,\hat a_j^\dag \hat V^\dag &= \sum_{n=1}^M (V^\dag)_{jn}^*\, \hat a_n^\dag = \sum_{n=1}^M V_{nj}\, \hat a_n, \label{eq:PassiveTransformation2Schrodinger}
\end{align}
where $V_{nm}$ are the elements of the corresponding $M\times M$ unitary matrix $V$. Analogous relations hold for $\hat U_\lambda$: $\hat U_\lambda^\dag \hat a_n\hat U_\lambda = \sum_{i=1}^M U_{ni}(\lambda)\, \hat a_i$. We also define $\hat a_n(\lambda) := \hat U_\lambda \hat a_n \hat U_\lambda^\dag$.

The generator $\hat G$ of the unitary mode transformation $\hat U_\lambda = e^{-i\lambda \hat G}$ can be expressed in any mode basis. In the bases $\{\hat a_n\}_n$ and $\{\hat c_n\}_n$ it reads
\begin{align}
    \hat G = \sum_{n,m=1}^M G_{nm}\, \hat a_n^\dag \hat a_m = \sum_{i,j=1}^M \tilde G_{ij}\, \hat c_i^\dag \hat c_j,
\end{align}
where $\tilde G_{ij} = (V^\dag G V)_{ij}$, obtained by applying Eqs.~\eqref{eq:PassiveTransformation1}--\eqref{eq:PassiveTransformation2} together with $\hat c_n = \hat V \hat a_n\hat V^\dag$. The generator matrix $G$ is Hermitian and therefore diagonalizable, so it is always possible to find a mode basis $\{\hat b_i\}_i$, defined by $\hat b_i = \hat B \hat a_i \hat B^\dag$, in which the generator takes the diagonal form
\begin{align}
    \hat G = \sum_{i=1}^M g_i\, \hat b_i^\dag \hat b_i,
\end{align}
where $\{g_1,\ldots, g_M\}$ are the real eigenvalues of $G$, ordered without loss of generality as $g_1 \leq g_2 \leq \cdots \leq g_M$. The mode transformation $\hat B$ acts on the basis $\{\hat a_n\}_n$ as
\begin{align}
    \hat B\, \hat a_n \hat B^\dag = \sum_{i=1}^M (B^\dag)_{ni}\, \hat a_i, \label{app:eq:DefBunitary}
\end{align}
where $B_{ij}$ are the elements of a $M\times M$ unitary matrix $B$, and one has $G = B\,\mathrm{diag}(g_1,\ldots,g_M)\,B^\dag$.

\begin{table}[h]
\centering
\begin{tabular}{cccc}
\hline
\textbf{Basis} & \textbf{Notation} & \textbf{Property} & \textbf{Relation}\\
\hline
Arbitrary mode basis 
    & $\{\hat a_n\}_n$  & - & -
    \\[4pt]

Generator Diagonalizing basis 
    & $\{\hat b_n\}_n$   & $\hat G = \sum_{n=1}^M g_n \hat b_n^\dag \hat b_n$ & $\hat b_n = \hat B\hat a_n\hat B^\dag$
   \\[4pt]

Schmidt/Takagi basis 
    & $\{\hat c_n\}_n$  & $\hat V\hat S\hat V^\dag = [\otimes_{n=1}^M \hat S_{\hat c_n}(r_n)]$ & $\hat c_n = \hat V\hat a_n\hat V^\dag$
    \\
\hline
\end{tabular}
\caption{The three mode bases  that are used in this manuscript in the discrete and finite mode setting.}
\end{table}

\subsection{Bringing Gaussian states into disentengled form}\label{app:subsec:GaussianStateDisentangled}
We stated in the main text that any pure Gaussian state can be written in the product form of Eq.~\eqref{eq:GeneralGaussianState}, that we will now show. A singular value decomposition of the complex symmetric squeezing matrix $f$ appearing in Eq.~\eqref{eq:GeneralGaussianArbitraryBasis} takes the general form
\begin{align}
    f = V_L\, D\, V_R^\dag,
\end{align}
where $D$ is a diagonal matrix with non-negative entries and $V_L$, $V_R$ are unitary matrices. Since $f$ is symmetric, i.e., $f = f^T$, the decomposition specializes to
\begin{align}
    f = V\, D\, V^T,
\end{align}
known as the Takagi factorization, with components $f_{nm} = \sum_{j=1}^M V_{nj}\,r_j\, V_{mj}$.

Consider the multimode squeezing operator. Substituting the Takagi factorization and regrouping terms gives
\begin{align}
    &\exp \left( \frac{1}{2} \sum_{n,m=1}^M f_{nm}\,\hat a_n^\dag \hat a_m^\dag + \mathrm{h.c.} \right) \notag\\
    &= \exp \left( \frac{1}{2} \sum_{n,m=1}^M \sum_{j=1}^M V_{nj}\,r_j\, V_{mj}\,\hat a_n^\dag \hat a_m^\dag + \mathrm{h.c.} \right) \\
    &= \exp \left( \frac{1}{2} \sum_{j=1}^M r_j \left(\sum_{n=1}^M V_{nj}\,\hat a_n^\dag\right)\!\left(\sum_{m=1}^M V_{mj}\,\hat a_m^\dag\right) + \mathrm{h.c.} \right)\\
    &= \exp \left( \frac{1}{2} \sum_{j=1}^M r_j\, (\hat c_j^\dag)^2 + \mathrm{h.c.} \right),
\end{align}
where we defined $\hat c_j := \hat V\hat a_j\hat V^\dag$ as the new mode basis, and used $\sum_{n=1}^M V_{nj}\,\hat a_n^\dag = \hat V \hat a_j^\dag \hat V^\dag = \hat c_j^\dag$.

The displacement operator can likewise be rewritten in the $\{\hat c_j\}_j$ basis. Using $\hat a_n^\dag = \sum_{j=1}^M V_{nj}^*\,\hat c_j^\dag$, which follows from Eq.~\eqref{eq:PassiveTransformation2}, one obtains
\begin{align}
    \exp\!\left( \sum_{n=1}^M \beta_n\, \hat a_n^\dag - \mathrm{h.c.} \right)
    &= \exp\!\left( \sum_{n=1}^M \beta_n \sum_{j=1}^M V_{nj}^*\,\hat c_j^\dag - \mathrm{h.c.} \right) \\
    &= \exp\!\left( \sum_{j=1}^M \alpha_j\,\hat c_j^\dag - \mathrm{h.c.} \right),
\end{align}
with $\alpha_j = \sum_{n=1}^M (V^\dag)_{jn}\,\beta_n$.

Combining these results, the state in Eq.~\eqref{eq:GeneralGaussianArbitraryBasis} becomes
\begin{align}
    \vert \psi\rangle
    = \exp \left( \sum_{j=1}^M \alpha_j\,\hat c_j^\dag - \mathrm{h.c.} \right)
      \exp \left( \frac{1}{2} \sum_{j=1}^M r_j\,(\hat c_j^\dag)^2 - \mathrm{h.c.} \right) \vert 0\rangle
    = \bigotimes_{j=1}^M \hat D_{\hat c_j}(\alpha_j)\,\hat S_{\hat c_j}(r_j)\,\vert 0\rangle,
\end{align}
which, expressed back in the original mode basis $\{\hat a_n\}_n$, reads
\begin{align}
   \vert \psi\rangle = \hat V\!\left[ \bigotimes_{j=1}^M \hat D_{\hat a_j}(\alpha_j)\,\hat S_{\hat a_j}(r_j) \right] \vert 0\rangle.
\end{align}

\subsection{Relating the squeezing matrix to the covariance matrix}\label{app:CovarianceMatrix}
Here, we will relate the squeezing matrix $f$ to the more common covariance matrix $\Sigma$ that one encounters when working with quadrature operators
\begin{align}
    \hat q_n = \frac{\hat c_n+\hat c_n^\dag}{\sqrt{2}} , \quad\quad\quad \hat p_n = \frac{\hat c_n^\dag - \hat c_n^\dag}{\sqrt{2}i}
\end{align}
instead of annihilation $\hat c_n$ and creation $\hat c_n^\dag$ operators. Note that here we work with $\{ \hat c_n\}_n$ (with $\hat c_n = \hat V\hat a_n\hat V^\dag$)  as in this basis the covariance takes on a simple diagonal form. The quadrature operators satisfy the commutation relation $[\hat x_n,\hat p_m] = \hbar i\delta _{nm}$.
Let us introduce 
\begin{align}
    \hat{\mathbf{r}} = (\hat x_1,\hat p_1,\hat x_2,\hat p_2,\ldots,\hat x_M,\hat p_M)^T .
\end{align}
The components of the $2M\times 2M$ covariance matrix are defined as
\begin{align}
    \Sigma_{kl} = \text{tr}[\vert \psi\rangle\langle\psi\vert \{ (\hat r_k-\langle\hat r_k \rangle), (\hat r_l-\langle\hat r_l \rangle) \}]
\end{align}
where $\{\hat A,\hat B\}= AB+BA$.

Now, consider a general pure Gaussian state
\begin{align}
    \vert\psi\rangle &= \hat V \left[\otimes_{n=1}^M\hat D_{\hat a_n}(\alpha_n) \hat S_{\hat a_n}(r_n)\right]\vert 0\rangle \\
    &= \left[\otimes_{n=1}^M\hat D_{\hat V\hat a_n\hat V^\dag}(\alpha_n) \hat S_{\hat V\hat a_n\hat V^\dag}(r_n)\right]\vert 0\rangle \\
       &= \left[\otimes_{n=1}^M\hat D_{\hat c_n}(\alpha_n) \hat S_{\hat c_n}(r_n)\right]\vert 0\rangle .
\end{align}
The associated $M\times M$ squeezing matrix expressed in the $\{\hat c_n\}_n$ mode basis is
\begin{align}
    f &=     \text{diag}(r_1,r_2,\ldots, r_M) 
\end{align}
The $2M\times 2M$ covariance matrix expressed in the $\{\hat c_n\}_n$ mode basis
\begin{align}
       \Sigma =  \bigoplus_{n=1}^M \begin{pmatrix}
     e^{+ 2r_n } & 0\\
     0 & e^{-2r_n }
    \end{pmatrix}  .
\end{align}
Note that our squeezing operators and the mode basis $\{\hat c_n\}_n$ are defined in such a way that the $p$ quadratures are squeezed while the $q$ quadratures are anti-squeezed.

\section{QFI of a multi-mode Gaussian quantum state subject to a mode transform}
In what follows and generally throughout all the manuscript, all sums run from $1$ to $M$ if not explicitly indicated otherwise.

\subsection{The general case}\label{app:QFIderivation}
The derivative of a mode transform is given by
\begin{align}
    \partial_\lambda \hat U_\lambda = \partial_\lambda e^{-i\lambda \hat G}=-i \hat U_\lambda \hat G.
\end{align}

Let us continue examining the derivative of the state and simplify it
\begin{align}
   i\vert \partial_\lambda \psi_\lambda\rangle &= \hat U_\lambda \hat G  \vert \psi\rangle =  \hat U_\lambda \hat G  \hat V \hat D\hat S \vert 0\rangle =  \hat U_\lambda \hat V \hat V^\dag\hat G  \hat V \hat D\hat S \vert 0\rangle  \\
   &=\hat U_\lambda \hat V \sum_{nm} G_{nm}  \hat V^\dag \hat a_n^\dag \hat V\hat V^\dag \hat a_m \hat V \hat D\hat S\vert 0 \rangle \\
   &= \hat U_\lambda \hat V  \sum_{ij}\sum_{nm} V_{ni}^* G_{nm}  V_{mj}  \hat a_i^\dag   \hat a_j  \hat D\hat S\vert 0 \rangle \\
   &=\hat U_\lambda \hat V  \sum_{ij}  \tilde G_{ij}    \hat a_i^\dag   \hat a_j  \hat D\hat S\vert 0 \rangle \\
    &=\hat U_\lambda \hat V \hat D \sum_{ij}  \tilde G_{ij}   \hat D^\dag\hat a_i^\dag \hat D\hat D^\dag   \hat a_j  \hat D\hat S\vert 0 \rangle \\
    &= \hat U_\lambda \hat V \hat D \sum_{ij}  \tilde G_{ij} (\hat a_i^\dag +\alpha_i^*) (\hat a_j + \alpha_j) \hat S\vert 0 \rangle \\
     &= \hat U_\lambda \hat V \hat D \sum_{ij}  \tilde G_{ij} (\hat a_i^\dag \hat a_j + \hat a_j  \alpha_i^* + \hat a_i^\dag \alpha_j + \alpha_i^*\alpha_j) \hat S\vert 0 \rangle \\
     &= \hat U_\lambda \hat V \hat D \hat S \sum_{ij}  \tilde G_{ij} (\hat S^\dag \hat a_i^\dag \hat S \hat S^\dag \hat a_j \hat S+ \hat S^\dag \hat a_j\hat S \ \alpha_i^* + \hat S^\dag \hat a_i^\dag \hat S \alpha_j + \alpha_i^*\alpha_j)  \vert 0 \rangle \\
     &= \hat U_\lambda \hat V \hat D \hat S \sum_{ij}  \tilde G_{ij} \Big( [c_i \hat a_i^\dag + s_i \hat a_i][ c_j \hat a_j + s_j \hat a_j^\dag] \\
      & \quad\quad\quad\quad\quad\quad\quad\quad\quad\quad+ [c_j \hat a_j + s_j \hat a_j^\dag]   \alpha_i^* + [c_i \hat a_i^\dag +s_i \hat a_i] \alpha_j + \alpha_i^*\alpha_j \Big)  \vert 0 \rangle ,
\end{align}
where we have used $\hat V^\dag a_m \hat V = \sum_{j} V_{mj} \hat a_j$ in the third line, in the fourth line we denoted the Hermitian matrix $\tilde G_{ij} = \sum_{nm} V_{ni}^* G_{nm} V_{mj} = (V^\dag G V)_{ij}$, in the sixth line we used $\hat D^\dag \hat a_j\hat D = \hat a_j +\alpha_j$, and in the 9th and 10th lines we used $\hat S^\dag \hat a_j \hat S = c_j \hat a_j + s_j \hat a_j^\dag$, where we used the shorthand notation $c_j := \cosh (r_j)$ and $s_j := \sinh (r_j)$. Now we can further simplify this by using $\hat a_i \vert 0\rangle = 0$ and $\hat a_i\hat a_j^\dag \vert 0\rangle = \delta_{ij} \vert 0\rangle$, and we obtain
\begin{align}
    \vert\partial_\lambda \psi_\lambda\rangle &= -i\hat U_\lambda \hat V \hat D \hat S \sum_{ij}  \tilde G_{ij} \Big[  c_i s_j \hat a_i^\dag \hat a_j^\dag \vert 0\rangle + (s_j \alpha_i^* \hat a_j^\dag + c_i \alpha_j \hat a_i^\dag)\vert 0\rangle +(s_i s_j \delta_{ij} +\alpha_i^* \alpha_j)\vert 0\rangle \Big] \\
    &=:  \vert\partial_\lambda \psi_\lambda\rangle_2 + \vert\partial_\lambda \psi_\lambda\rangle_1 + \vert\partial_\lambda \psi_\lambda\rangle_0 .
\end{align}
The derivative state is now in a form such that all Gaussian operators are on the left side. We see 
that the derivative states split up into three different components, a two-photon, one-photon, and zero-photon component. We can examine these individually due to $\langle\partial_\lambda\psi_\lambda\vert \partial_\lambda\psi_\lambda\rangle = \, _2\langle\partial_\lambda\psi_\lambda\vert \partial_\lambda\psi_\lambda\rangle_2 +\, _1\langle\partial_\lambda\psi_\lambda\vert \partial_\lambda\psi_\lambda\rangle_1 +\, _0\langle\partial_\lambda\psi_\lambda\vert \partial_\lambda\psi_\lambda\rangle_0$. 

So, let us start with the two-photon term:
\begin{align}
    \, _2\langle\partial_\lambda\psi_\lambda\vert \partial_\lambda\psi_\lambda\rangle_2 &= \sum_{i',j',i,j} \tilde G_{i'j'}^* \tilde G_{ij} c_{i'} s_{j'} c_i s_j \langle 0\vert \hat a_{i'} \hat a_{j'}\hat a_i^\dag \hat a_j^\dag\vert 0\rangle  \\
    &= \sum_{ij} \tilde G_{ji}^* \tilde G_{ij} c_j s_i c_i s_j + \sum_{ij} \tilde G_{ij}^* \tilde G_{ij} c_i s_j c_i s_j \\
    &= \sum_{ij} s_i \tilde G_{ij} c_j s_j \tilde G_{ji}^* c_i +\sum_{ij} c_i \tilde G_{ij} s_j s_j \tilde G_{ji} c_i ,
\end{align}
where we used  $\langle 0\vert \hat a_{i'} \hat a_{j'}\hat a_i^\dag \hat a_j^\dag\vert 0\rangle = \delta_{i',j} \delta_{j' i} + \delta_{i' i}\delta_{j'j}$

Now, we come to the one-photon term
\begin{align}
    \vert \partial_\lambda\psi_\lambda\rangle_1 &= -i\hat U_\lambda \hat V \hat D \hat S \sum_{ij}   \Big[    \tilde G_{ij} s_j \alpha_i^* \hat a_j^\dag + \tilde G_{ij} c_i \alpha_j \hat a_i^\dag    \Big] \vert 0\rangle \\
    &= -i\hat U_\lambda \hat V \hat D \hat S \sum_{ij}   \Big[   \tilde G_{ij} s_j \alpha_i^*   + \tilde G_{ji} c_j \alpha_i     \Big]\hat a_j^\dag \vert 0\rangle  \\
    &=: -i\hat U_\lambda \hat V \hat D \hat S \sum_{ij} F_{ij} \hat a_j^\dag \vert 0\rangle,
\end{align}
where we defined $F_{ij} = \tilde G_{ij} s_j \alpha_i^*   + \tilde G_{ji} c_j \alpha_i    $.
With this, we find 
\begin{align}
     \, _1\langle\partial_\lambda\psi_\lambda\vert \partial_\lambda\psi_\lambda\rangle_1 &= \sum_{i' j'ij} F_{i'j'}^* F_{ij} \langle 0\vert \hat a_{j'} \hat a_j^\dag \vert 0\rangle \\
     &= \sum_{i'ij} F_{i'j}^* F_{ij} \\
     &= \sum_{i'ij} \left(  \tilde G_{i'j}^* s_j \alpha_{i'}   + \tilde G_{ji'}^* c_j \alpha_{i'}^*     \right) \left( \tilde G_{ij} s_j \alpha_i^*   + \tilde G_{ji} c_j \alpha_i     \right) \\
     &= \sum_{i'ij} \Big[ \alpha_i^* \tilde G_{ij} s_j s_j \tilde G_{j i'} \alpha_{i'} + \alpha_{i'}^* \tilde G_{i' j} c_j c_j \tilde G_{ji}\alpha_i  \\
     &\quad \quad\quad  + \alpha_{i'}^* \tilde G_{i' j} c_j s_j \tilde G_{ji}^* \alpha_i^* + \alpha_{i'} \tilde G_{i' j}^* s_j c_j \tilde G_{ji}\alpha_i \Big] .
\end{align}

For the zero-photon term we find
\begin{align}
    \, _0\langle\partial_\lambda\psi_\lambda\vert \partial_\lambda\psi_\lambda\rangle_0 &= \sum_{i'j'}  \tilde G_{i'j'}^*     (s_{i'} s_{j'} \delta_{i'j'} +\alpha_{i'} \alpha_{j'}^* ) \sum_{ij}  \tilde G_{ij}    (s_i s_j \delta_{ij} +\alpha_i^* \alpha_j) \\
    &= \vert\langle \psi_\lambda \vert \partial_\lambda \psi_\lambda\rangle\vert^2 .
\end{align}

Now, by adding all of these terms together, we  obtain the QFI. We have
\begin{align}
    \mathcal{F}(\vert\psi_\lambda\rangle) &= 4\left(\langle\partial_\lambda \psi_\lambda\vert \partial_\lambda \psi_\lambda\rangle - \vert \langle \psi_\lambda \vert \partial_\lambda \psi_\lambda\rangle \vert^2 \right) \\
    &= 4\left( \, _2\langle\partial_\lambda\psi\vert \partial_\lambda\psi_\lambda\rangle_2 +\, _1\langle\partial_\lambda\psi\vert \partial_\lambda\psi_\lambda\rangle_1 +\, _0\langle\partial_\lambda\psi\vert \partial_\lambda\psi_\lambda\rangle_0 - \vert \langle \psi_\lambda \vert \partial_\lambda \psi_\lambda\rangle \vert^2 \right) \\
    &= 4\left( \, _2\langle\partial_\lambda\psi\vert \partial_\lambda\psi_\lambda\rangle_2 +\, _1\langle\partial_\lambda\psi\vert \partial_\lambda\psi_\lambda\rangle_1 \right) .
\end{align}
Plugging in the expressions from above, we obtain Eq.~\eqref{eq:QFIcomponents} of the main text for the QFI.

\begin{align}  
    &\mathcal{F}(\vert \psi_\lambda\rangle) = 4\sum_{ij}  \tilde G_{ij} c_j s_j \tilde G_{ji}^* c_i s_i +4\sum_{ij} \tilde G_{ij} s_j s_j \tilde G_{ji} c_i  c_i\\
    &\quad\quad+\sum_{i'ij} 4\Big[  \tilde G_{ij} s_j s_j \tilde G_{j i'} \alpha_{i'}\alpha_i^* +  \tilde G_{i' j} c_j c_j \tilde G_{ji}\alpha_i\alpha_{i'}^*   
     \\&\quad\quad\quad\quad\quad+  \tilde G_{i' j} c_j s_j \tilde G_{ji}^* \alpha_i^*\alpha_{i'}^* +  \tilde G_{i' j}^* s_j c_j \tilde G_{ji}\alpha_i \alpha_{i'} \Big] .
\end{align}

The QFI in this form is quite cumbersome, and it is useful to introduce some new notation.  We define the matrices $C= \text{diag} (c_1,\ldots ,c_M)$, $S = \text{diag}(s_1,\ldots, s_M)$, and  the $M\times M$ matrix  $\mathcal A$ with components $\mathcal  A_{ij}= \alpha_i \alpha_j^*$ and  $\mathcal  B_{ij} = \alpha_i \alpha_j$. Also note that $\tilde G = V^\dag G V$. The QFI then takes the form
\begin{align} \label{app:eq:QFIexplicitFormula1}
     \mathcal{F}(\vert\psi_\lambda\rangle) &= 4\Big ( \text{Tr}[\tilde G C S \tilde G^*  C S] + \text{Tr}[\tilde G S ^2 \tilde G C^2] \\\label{app:eq:QFIexplicitFormula2}
     &\quad\quad\quad + \text{Tr}[\tilde G S^2 \tilde G \mathcal  A]+ \text{Tr}[\tilde G C^2 \tilde G \mathcal A] + \text{Tr}[\tilde G C S \tilde G^* \mathcal B^*]+ \text{Tr}[\tilde G^* S C \tilde G \mathcal B]\Big)  \\
     &= 4\Big ( \text{Tr}[\tilde G C S \tilde G^*  C S] + \text{Tr}[\tilde G S ^2 \tilde G C^2] 
     + \text{Tr}[\tilde G (S^2+C^2) \tilde G \mathcal A] + 2\mathrm{Re} \left[\text{Tr}[\tilde G^* S C \tilde G \mathcal B]\right]\Big) .
\end{align}
This can be seen by using the following relation
\begin{align}
    \text{tr}[ABCDEF]= \sum_{a} (ABCDEF)_{aa}= \sum_{a,b,c,d,e,f} A_{ab}B_{bc}C_{cd}D_{de}E_{ef}F_{fa}
\end{align}
for some arbitrary $M\times M$ matrices $A,B,C,D,E,F$. 
If $B,C,E,F$ are diagonal, we have
\begin{align}
    \text{tr}[ABCDEF]&=\sum_{a,b,c,d,e,f} A_{ab}B_{bc}\delta_{bc} C_{cd}\delta_{cd}D_{de}E_{ef} \delta_{ef}F_{fa} \delta_{fa} \\
    &= \sum_{a,b,d,e,f} A_{ab}B_{bb}  C_{bd}\delta_{bd}D_{de}E_{ef} \delta_{ef}F_{fa} \delta_{fa}\\
    &= \sum_{a,b,e,f} A_{ab}B_{bb} C_{bb} D_{be}E_{ef} \delta_{ef}F_{fa} \delta_{fa}\\
    &= \sum_{a,b,e} A_{ab}B_{bb} C_{bb} D_{be}E_{ee}  F_{ea} \delta_{ea}\\
    &= \sum_{a,b} A_{ab}B_{bb} C_{bb} D_{ba}E_{aa}  F_{aa} .
\end{align}
With this we confirm the validity of the first two terms in Eq.~\eqref{app:eq:QFIexplicitFormula1}. 

For the last four terms, we consider
\begin{align}
    \text{tr}[ABCDE] =\sum_{a} (ABCDE)_{aa} = \sum_{a,b,c,d,e} A_{ab}B_{bc}C_{cd}D_{de}E_{ea}
\end{align}
Now, if $B,C$ are diagonal, we have
\begin{align}
    \text{tr}[ABCDE] &= \sum_{a,b,c,d,e} A_{ab}B_{bc}\delta_{bc}C_{cd}\delta_{cd}D_{de}E_{ea} \\
    &=  \sum_{a,b,d,e} A_{ab}B_{bb} C_{bd}\delta_{bd}D_{de}E_{ea}\\
    &=  \sum_{a,b,e} A_{ab}B_{bb} C_{bb} D_{be}E_{ea}.
\end{align}
With this, we confirm the four terms in Eq.~\eqref{app:eq:QFIexplicitFormula2}.

As we will later see, bringing the QFI into this form will be helpful for proving our main result.

\subsection{QFI of a coherent state}\label{app:QFIcoherent}
As in the main text, the QFI of the coherent state is

\begin{align}
    \mathcal{F}(\vert\psi^{\text{coh}}\rangle ) = 4\sum_{i'ij} \alpha_{i'}^* \tilde G_{i' j}  \tilde G_{ji}\alpha_i   = 4 \boldsymbol{\alpha}^\dag \tilde G^2  \boldsymbol{\alpha} = 4 \boldsymbol{\alpha}^\dag W^\dag D^2 W\boldsymbol{\alpha} = 4\sum_{n}  D_{nn}^2 \vert (W\boldsymbol{\alpha})_n\vert^2 .
\end{align}
Now,  $W = B^\dag V$ is the unitary that diagonalizes $\tilde G$ (and also $\tilde G^2$). Recall that $G = B D B^\dag$.
\begin{align}
   \vert (W\boldsymbol{\alpha})_n\vert^2 &= \sum_{ij} W_{nj}^* W_{ni} \alpha_j^* \alpha_i = \sum_{ij} W_{nj}^* W_{ni}  \langle \psi^{\text{coh}}\vert \hat V \hat a_j^\dag \hat a_i \hat V^\dag\vert \psi^{\text{coh}}\rangle  \\
   &= \sum_{ij} W_{nj}^* W_{ni}  \langle \psi^{\text{coh}}\vert \hat V \hat a_j^\dag \hat V^\dag \hat V \hat a_i \hat V^\dag\vert \psi^{\text{coh}}\rangle = \sum_{ij} W_{nj}^* W_{ni}  \langle \psi^{\text{coh}}\vert \sum_q V_{qj} \hat a_q^\dag  \sum_r V_{ri}^* \hat a_r\vert \psi^{\text{coh}}\rangle\\ 
   &= \sum_{ijqr} V_{qj}W_{nj}^* W_{ni} V_{ri}^* \langle \psi^{\text{coh}}\vert   \hat a_q^\dag      \hat a_r\vert \psi^{\text{coh}}\rangle = \sum_{qr} (\sum_j V_{qj} (W^\dag)_{jn}) (\sum_i W_{ni} (V^\dag)_{ir})  \langle \psi^{\text{coh}}\vert   \hat a_q^\dag      \hat a_r\vert \psi^{\text{coh}}\rangle \\
   &= \sum_{qr} (V  (B^\dag V)^\dag)_{qn} ( B^\dag V  V^\dag)_{nr}  \langle \psi^{\text{coh}}\vert   \hat a_q^\dag      \hat a_r\vert \psi^{\text{coh}}\rangle =      \langle \psi^{\text{coh}}\vert  (\sum_q B_{qn}  \hat a_q^\dag )      (\sum_r B_{nr}^\dag\hat a_r)\vert \psi^{\text{coh}}\rangle\\
   &= \langle \psi^{\text{coh}}\vert  \hat b_n^\dag  \hat b_n \vert\psi^{\text{coh}}\rangle.
\end{align}
In the first line we used $\alpha_j\alpha_i = \langle \psi^{\text{coh}}\vert \hat V \hat a_j^\dag \hat a_i \hat V^\dag\vert \psi^{\text{coh}}\rangle $, recall that $\vert\psi^{\text{coh}}\rangle = \hat V [\otimes_i\hat D_{\hat a_i}(\alpha_i)]\vert 0\rangle$. In the second line we used Eq.~\eqref{eq:PassiveTransformation1Schrodinger} and Eq.~\eqref{eq:PassiveTransformation2Schrodinger}. In the fourth line we used $W= B^\dag V$.

From this follows
\begin{align}
    \mathcal{F}(\vert\psi^{\text{coh}}\rangle ) &= 4\sum_{n}  D_{nn}^2 \vert (W\boldsymbol{\alpha})_n\vert^2 \\
    &= 4 \sum_n g_n^2 \langle \psi^{\text{coh}}\vert  \hat b_n^\dag  \hat b_n \vert\psi^{\text{coh}}\rangle .
\end{align}

\section{Resources}\label{app:ResourceParameters}
In the main text we have given the definition of the resources in terms of the intensity $I_n$. It will prove useful to express the resources in a basis independent way.

Let us first start with the signal photon number
\begin{align}
    N_S &= \sum_{i\in \mathcal{I}_S}  \langle \hat b_i ^\dag  \hat b_i  \rangle = \sum_{n,m}\sum_{i \in \mathcal{I}_S} (B^\dag)_{in}^* (B^\dag)_{im}  \langle \hat a_n^\dag \hat a_m \rangle  \\
    &=  \sum_{n,m}\sum_{i \in \mathcal{I}_S} B_{ni} (B^\dag)_{im}  \langle \hat a_n^\dag \hat a_m \rangle \\
    &= \sum_{n,m} (P_S)_{nm}  \langle 0 \vert \hat S^\dag\hat D^\dag \hat V^\dag \hat a_n^\dag \hat a_m \hat V \hat D\hat S\vert 0\rangle   \\
     &= \sum_{n,m} (V^\dag P_S V)_{nm}  \langle 0 \vert \hat S^\dag\hat D^\dag  \hat a_n^\dag \hat a_m   \hat D\hat S\vert 0\rangle   \\
     &= \sum_{n,m} (\tilde P_S)_{nm}   \left( s_n^2 \delta_{nm} + \alpha_m \alpha_n^*  \right)   \\
    &=\text{Tr}[\tilde P_S (S^2+\mathcal A) ] \\
    &=\text{Tr}[\tilde P_S (S^2+\mathcal A) \tilde P_S] .
\end{align}
In the third line we defined the projector onto the eigenspace of $ G$ with non-zero eigenvalues $P_S = \sum_{i\in \mathcal{I}_S} B_{ni} (B^\dag)_{im}  = \sum_{i,j}  B_{ni} \chi _{\mathcal{I}_S} (i)\delta_{ij}(B^\dag)_{im}$, where $\chi _{\mathcal{I}_S} (i) =1$ if $i\in \mathcal{I}_S$ and $0$ otherwise. That is $P_S G = G P_S = P_SGP_S =G$.
In the fourth line, we  defined $\tilde P_S = V^\dag P_S V$, which is the projector onto the eigenspace of $\tilde G$ with non-zero eigenvalues. We have $\tilde P_S \tilde G  = \tilde G \tilde P_S = \tilde P_S \tilde G \tilde P_S = \tilde G$.
 
For the mean resource $\bar g$, we have
\begin{align}
    \bar g \cdot N_S &= \sum_{i\in \mathcal{I}_S} g_i  \langle \hat b_i^\dag  \hat b_i  \rangle \\
    &=  \sum_{i} g_i  \langle \hat b_i^\dag  \hat b_i  \rangle \\
    &= \sum_{n,m} G_{nm} \langle \hat a_n^\dag \hat a_m \rangle = \sum_{n,m} G_{nm} \langle 0 \vert \hat S^\dag\hat D^\dag \hat V^\dag \hat a_n^\dag \hat a_m \hat V \hat D\hat S\vert 0\rangle  \\
    &=\sum_{n m} (V^\dag G V)_{nm}  \langle 0 \vert \hat S^\dag\hat D^\dag   \hat a_{n}^\dag \hat a_{m}  \hat D\hat S\vert 0\rangle  \\
     &=\sum_{n m} \tilde G_{nm}   \left( s_n^2 \delta_{nm} +\alpha_n^* \alpha_m \right) \\
     &=  \text{Tr}[ \tilde G (S^2 +\mathcal  A)] 
\end{align}
and thus
\begin{align}
    \bar g =\frac{\text{Tr}[ \tilde G (S^2 +\mathcal  A)]}{N_S} .
\end{align}
We see that this quantity only depends on how the modes are distributed, not on the total number of photons. 

And finally, for the variance resource parameter $\bar g$ we have
\begin{align}
    \left( (\Delta g)^2 +\bar g^2 \right) N_S &= \sum_{i\in \mathcal{I}_S} g_i^2  \langle \hat b_i^\dag  \hat b_i  \rangle \\
    &=  \sum_{i} g_i^2  \langle \hat b_i^\dag  \hat b_i  \rangle \\
    &= \sum_{n,m} (G^2)_{nm} \langle \hat a_n^\dag \hat a_m \rangle = \sum_{n,m} (G^2)_{nm} \langle 0 \vert \hat S^\dag\hat D^\dag \hat V^\dag \hat a_n^\dag \hat a_m \hat V \hat D\hat S\vert 0\rangle  \\
    &=\sum_{n m} (V^\dag G^2 V)_{nm}  \langle 0 \vert \hat S^\dag\hat D^\dag   \hat a_{n}^\dag \hat a_{m}  \hat D\hat S\vert 0\rangle  \\
     &=\sum_{n m} (\tilde G^2)_{nm}   \left( s_n^2 \delta_{nm} +\alpha_n^* \alpha_m \right) \\
     &= \text{Tr}[ \tilde G^2 (S^2 + \mathcal A)] .
\end{align}
With this, we get
\begin{align}
    (\Delta g)^2 = \frac{\text{Tr}[ \tilde G^2 (S^2 + \mathcal A)] }{N_S} -\bar g^2 .
\end{align}

\section{Proof of Theorem~\ref{TheoremBound}}\label{App:ProofTheorem}

We now prove Theorem~\ref{TheoremBound}, which requires showing that $\mathcal{F}(\vert\psi_\lambda\rangle) \leq \left(8\bar{g}^2 + 4\Delta g^2\right)N_S^2 + O(N_S)$.

We begin by establishing two lemmas that will enable a concise proof of Theorem~\ref{TheoremBound}.

\begin{lemma}\label{sec4:lemma1}
The QFI given in Eq.~\eqref{sec4:eq:QFIexplicitFormula} is upper bounded by
\begin{align}
    \mathcal{F}(\vert\psi_\lambda\rangle) \leq 4\left(2\,\mathrm{Tr}[\tilde{G}S^2\tilde{G}S^2] + 4\,\mathrm{Tr}[\tilde{G}S^2\tilde{G}\mathcal{A}]\right) + O(N_S).
\end{align}
\end{lemma}

The proof relies on the Hermiticity of $\tilde{G}$ and, while straightforward, is somewhat lengthy; it is provided in Appendix~\ref{app4:lemm1}. The purpose of this intermediate bound is to reduce the expression to one containing only two terms, which simplifies the subsequent analysis.
    
The second lemma is:
\begin{lemma} \label{sec4:lemma2}
The inequality 
\begin{align}
  0  \leq 4 \mathrm{Tr}[H Q]^2  + 4 \mathrm{Tr}[H^2 Q]\mathrm{Tr}[Q]  -8  \mathrm{Tr}[HQHQ]
\end{align}
holds for $H$ Hermitian and $Q$ Hermitian and positive semi-definite. 
\end{lemma}
The proof of this lemma~\ref{sec4:lemma2} has been delegated to the Appendix~\ref{app4:lemm2}.
\\
We are now ready to complete the proof of Theorem~\ref{TheoremBound}.
\begin{proof}
We use the lemma~\ref{sec4:lemma2} and set   $ H = \tilde G$ and $Q =\tilde P_S( S+\mathcal A ) \tilde P_S := S' + \mathcal A'$, where we denote $S'= \tilde P_S S \tilde P_S$ and $\mathcal A'=\tilde P_S A\tilde P_S$. With this, we obtain
\begin{align}
   &0\leq   4 \text{Tr}[H Q]^2  + 4 \text{Tr}[H^2 Q]\text{Tr}[Q]     - 8 \text{Tr}[HQH Q] \\
    &=  4 \text{Tr}[\tilde G (S' + \mathcal A')]^2  + 4 \text{Tr}[\tilde G^2 (S' + \mathcal A')]\text{Tr}[ S' + \mathcal A' ]     - 8 \text{Tr}[\tilde G(S' + \mathcal A') \tilde G(S' + \mathcal A')]\\
    &= \left( 8\bar g^2 + 4(\Delta g)^2 \right) N_S^2   - 8  \text{Tr}[\tilde G S'   \tilde G S'  ] - 2\cdot 8 \text{Tr}[\tilde GS'\tilde G \mathcal A'] -8 \text{Tr}[\tilde G  \mathcal A'\tilde G  \mathcal A']
\end{align}
where in the last line we used $4 \text{Tr}[\tilde G (S' + \mathcal A')]^2  = 4 \text{Tr}[\tilde G (S + \mathcal A)]^2 = 4 \bar g^2 N_S^2$ and $4 \text{Tr}[\tilde G^2 (S' + \mathcal A')]\text{Tr}[ S' + \mathcal A' ] = 4 \text{Tr}[\tilde G^2 (S + \mathcal A)]\text{Tr}[ S + \mathcal A ] = 4 (\bar g^2 + (\Delta g)^2)N_S^2$, which follows from $\tilde P_S \tilde G = \tilde G \tilde P_S = \tilde G$. We also used in the last line
 $ 
    \text{Tr}[\tilde G(S'+\mathcal A')\tilde G(S'+\mathcal A')]  = \text{Tr}[\tilde GS'\tilde GS'] + 2\text{Tr}[\tilde GS' \tilde G \mathcal A'] + \text{Tr}[\tilde G \mathcal A'\tilde G \mathcal A'] 
$  
 exploiting linearity and the cyclic property of the trace.  By now bringing the first two terms with the minus sign to the left side, we obtain:
 \begin{align} 
    8  \text{Tr}[\tilde G S'   \tilde G S'  ] + 2\cdot 8 \text{Tr}[\tilde GS'\tilde G\mathcal  A'] \leq \left( 8\bar g^2 + 4(\Delta g)^2 \right) N_S^2 - 8 \text{Tr}[\tilde G \mathcal  A'\tilde G  \mathcal A'] . \label{sec4:eq:inequalityproof}
 \end{align}
 We now use the lemma~\ref{sec4:lemma1}. For this, we add the appropriate terms of order $L:=4\cdot 3   (\bar g^2+ \Delta g^2)N_S  - 4\text{Tr}[\tilde G^2 S^2 ]$ on both sides of Eq.~\eqref{sec4:eq:inequalityproof} and   with $\text{Tr}[\tilde G \mathcal  A'\tilde G  \mathcal A'] = \text{Tr}[\tilde G \mathcal A \tilde G  \mathcal A]$, $\text{Tr}[\tilde G S'   \tilde G S'  ] = \text{Tr}[\tilde G S   \tilde G S  ] $ and $\text{Tr}[\tilde GS'\tilde G \mathcal A']  = \text{Tr}[\tilde GS\tilde G \mathcal A]  $  we obtain: 

 \begin{align} 
    8  \text{Tr}[\tilde G S   \tilde G S  ] + 2\cdot 8 \text{Tr}[\tilde GS\tilde G \mathcal A] +  L\leq \left( 8\bar g^2 + 4(\Delta g)^2 \right) N_S^2 - 8 \text{Tr}[\tilde G  A\tilde G  \mathcal A]  +   L.
 \end{align}
From this and   lemma~\ref{sec4:lemma1} follows directly:
    \begin{align}
    \mathcal{F}(\vert\psi_\lambda\rangle) &\leq  \left( 8\bar g^2 + 4(\Delta g)^2 \right) N_S^2 - 8 \text{Tr}[\tilde G  A\tilde G  \mathcal A] +L\label{sec4:eq:tightbound}\\
    &\leq  \left( 8\bar g^2 + 4(\Delta g)^2 \right) N_S^2  + L  \label{sec4:eq:normalbound} \\
     &= \left( 8\bar g^2 + 4(\Delta g)^2 \right) N_S^2  + 4\cdot 3   (\bar g^2+ \Delta g^2)N_S  - 4\text{Tr}[\tilde G^2 S^2 ] .\label{sec4:eq:tightbound1}
\end{align}
In the last line, we used $\text{Tr}[\tilde G  \mathcal A \tilde G  \mathcal A ]\geq 0$. This follows from the fact that $\mathcal A$ is Hermitian and positive semi-definite. Therefore, $\mathcal A^{1/2}$ exists and is also Hermitian. Thus, we have $\text{Tr}[\mathcal A^{1/2}\tilde G\mathcal A^{1/2} \mathcal A^{1/2}\tilde  G \mathcal A^{1/2}] = \text{Tr}[(\mathcal A^{1/2}\tilde  G\mathcal A^{1/2})^2] \geq 0$. This follows from the fact that $\mathcal A^{1/2}\tilde G\mathcal A^{1/2}$ is Hermitian, and the trace of the square of a Hermitian matrix is always non-negative. This completes the proof.
\end{proof}

An important observation is that Eq.~\eqref{sec4:eq:tightbound} implies that Eq.~\eqref{sec4:eq:normalbound} can only be attained if $\mathrm{Tr}[\tilde{G}\mathcal{A}\tilde{G}\mathcal{A}] = 0$. While this might suggest that displacement always decreases performance, there may exist special non-zero displacement configurations for which $\mathrm{Tr}[\tilde{G}\mathcal{A}\tilde{G}\mathcal{A}] = 0$, so that the possibility of an optimal state with non-zero displacement cannot be excluded \footnote{One can show that $\mathrm{Tr}[\tilde{G}\mathcal{A}\tilde{G}\mathcal{A}] = (\boldsymbol{\alpha}^\dag\tilde{G}\boldsymbol{\alpha})^2$. When $\tilde{G}$ has both positive and negative eigenvalues, one can find $\boldsymbol{\alpha} \neq 0$ such that $\mathrm{Tr}[\tilde{G}\mathcal{A}\tilde{G}\mathcal{A}] = (\boldsymbol{\alpha}^\dag\tilde{G}\boldsymbol{\alpha})^2 = 0$.}. Indeed, in Sec.~\ref{subsub:WithDisplacement} we identified a variance-optimal state that allocates half of its energy budget to displacement.

In the case of zero displacement, the bound in Eq.~\eqref{sec4:eq:tightbound1} reduces to
\begin{align}
\mathcal{F}(\vert\psi_\lambda\rangle) \leq \left(8\bar{g}^2 + 4(\Delta g)^2\right)N_S^2 + 8(\bar{g}^2 + \Delta g^2)N_S,
\end{align}
which is tight even at the level of the linear term in $N_S$, as can be verified by evaluating the QFI of the optimal state given in Eq.~\eqref{app:QFIoptimalState}. Finally, in the limit of small squeezing, a multimode squeezed vacuum state reduces to a superposition of the vacuum and a two-photon state, suggesting that the present bound may also be applicable to two-photon states — a direction worth pursuing in future work.

\subsection{Proof of Lemma~\ref{sec4:lemma1}}\label{app4:lemm1}
\begin{proof}
First, recall that the QFI is 
\begin{align} 
     \mathcal{F}(\vert\psi_\lambda\rangle)= 4\Big [ \text{Tr}[\tilde G C S \tilde G^*  C S] + \text{Tr}[\tilde G S ^2 \tilde G C^2]  
     + \text{Tr}[\tilde G S^2 \tilde G A]+ \text{Tr}[\tilde G C^2 \tilde G A] + \text{Tr}[\tilde G C S \tilde G^* B^*]+ \text{Tr}[\tilde G^* S C \tilde G B]\Big] . 
\end{align}
We have six terms and, by bounding some of them, we derive an upper bound with only two terms quadratic in $N_S$.
Let us start with the terms that are only due to squeezing. We have
\begin{align}
     \text{Tr}[\tilde G C S \tilde G^*  C S] &= \sum_{ij} \tilde G_{ij} c_j s_j \tilde G_{ji}^* c_i s_i =  \sum_{ij} \tilde G_{ij}^2 c_j s_j c_i s_i \\
     &=\sum_i \tilde G_{ii}^2 a_i a_i + \sum_{i\neq j} \tilde G_{ij}^2 a_j a_i \\
    &= \sum_i \tilde G_{ii}^2 a_i a_i + \sum_{i<j} (\tilde G_{ij}^2 +\tilde G_{ji}^2 ) a_i a_j \\
     &= \sum_i \tilde G_{ii}^2 a_i a_i + \sum_{i<j} (\tilde G_{ij}^2 + (\tilde G_{ij}^2)^* ) a_i a_j \\
      &= \sum_i \tilde G_{ii}^2 a_i a_i + \sum_{i<j} 2\text{Re}[\tilde G_{ij}^2 ] a_i a_j \\
        &= \sum_i \tilde G_{ii}^2 a_i a_i + \sum_{i\neq j} \text{Re}[\tilde G_{ij}^2 ] a_i a_j  
\end{align}
where we introduced $a_i = c_i s_i$ and used $ \text{Re}[ \tilde G_{ij}^2 ] =  \text{Re}[\tilde G_{ji}^2 ]$ in the last line. Note, here $a_i$ does not stand for an annihilation operator. The above clearly shows that $\text{Tr}[\tilde G C S \tilde G^*  C S]$ is real-valued.  We continue
\begin{align}
   \text{Tr}[\tilde G C S \tilde G^*  C S] \leq \vert \text{Tr}[\tilde G C S \tilde G^*  C S] \vert &\leq \Big\vert\sum_i \tilde G_{ii}^2 a_i a_i \Big\vert + \Big\vert \sum_{i\neq j} \text{Re}[\tilde G_{ij}^2 ] a_i a_j   \Big\vert \\
    &\leq \sum_i \tilde G_{ii}^2 a_i a_i + \sum_{i\neq j} \big\vert \text{Re}[\tilde G_{ij}^2 ]\big\vert a_i a_j  \\
    &\leq \sum_i \tilde G_{ii}^2 a_i a_i + \sum_{i\neq j} \vert  \tilde G_{ij}  \vert^2 a_i a_j  \\
    &=   \sum_{i, j} \vert \tilde G_{ij}  \vert^2 a_i a_j =     \sum_{i, j} \vert \tilde G_{ij}  \vert^2 s_i c_i s_j c_j\\
    &=  \sum_{i, j} \vert \tilde G_{ij}  \vert^2 s_i \sqrt{s_i^2+1} s_j \sqrt{s_j^2+1} \\
    &\leq \sum_{ij} \vert \tilde G_{ij}\vert^2  \left[ s_i^2 s_j^2 + \frac{1}{2} (s_i^2+s_j^2) \right]  \\
    &= \sum_{ij} \vert \tilde G_{ij}\vert^2 s_i^2 s_j^2 + \frac{1}{2}\sum_{ij} \vert \tilde G_{ij}\vert^2 (s_i^2+ s_j^2) \\
    &= \mathrm{Tr}[\tilde G S ^2 \tilde G S^2] + 2\frac{1}{2} \mathrm{Tr}[G^2 S^2] , \label{app:eq:InequalityChainLastLine}
\end{align}
where we used $\text{Re}[z^2  ] \leq \vert z\vert^2$ in the third line. In the fifth line we used $c_i= \sqrt{s_i^2+1}$. In the sixth line we used $\text{LHS}=s_i \sqrt{s_i^2+1} s_j \sqrt{s_j^2+1} \leq s_i^2 s_j^2 + \frac{1}{2} (s_i^2+s_j^2) =\text{RHS}$. This can be  seen by squaring both sides (both sides are non-negative), that is, the left hand side squared is $\text{LHS}^2=s_i^2 (s_i^2+1) s_j^2 (s_j^2+1) =  (s_i^4+s_i^2)  (s_j^4+s_j^2) = s_i^4s_j^4  +s_i^4s_j^2+ s_i^2s_j^4+s_i^2s_j^2$. The right hand side squared is 
\begin{align}
    \text{RHS}^2 = &[s_i^2 s_j^2 + \frac{1}{2} (s_i^2+s_j^2)]^2 = s_i^4s_j^4 + 2\frac{1}{2} s_i^2s_j^2 (s_i^2+s_j^2)+\frac{1}{4}(s_i^2+s_j^2)^2 \\
    &= s_i^4s_j^4 + s_i^4s_j^2 + s_i^2s_j^4 +\frac{1}{4}(s_i^2+s_j^2)^2 .
\end{align}
Now, let us rewrite $\text{LHS}^2\leq \text{RHS}^2$  as
\begin{align}
    0 &\leq \text{RHS}^2 - \text{LHS}^2 = \frac{1}{4}(s_i^2+s_j^2)^2 - s_i^2s_j^2 = \frac{1}{4} (s_i^2-s_j^2)^2 ,
\end{align}
which is obviously true as the term after the last equals sign is manifestly non-negative. And thus $\text{LHS}  \leq \text{RHS}$ from which Eq.~\eqref{app:eq:InequalityChainLastLine} finally follows.

Now, let us come to the terms that are due to displacement. Note that $  \text{Tr}[\tilde G C S \tilde G^* \mathcal B^*]^* = \text{Tr}[(\tilde G C S \tilde G^* \mathcal B^*)^\dag] = \text{Tr}[\tilde G^* C S \tilde G \mathcal B]$. Thus, we have
\begin{align}
    \text{Tr}[\tilde G C S \tilde G^* \mathcal B^*]+ \text{Tr}[\tilde G^* S C \tilde G \mathcal B ] &= 2 \text{Re}\left[ \text{Tr}[\tilde G C S \tilde G^* \mathcal B^*] \right] \\
    & \leq 2  \left \vert \text{Tr}[\tilde G C S \tilde G^* \mathcal B^*] \right\vert \\
    &=     2\left\vert \sum_{iji'} \alpha_{i'}^* \tilde G_{i' j} c_j s_j \tilde G_{ji}^* \alpha_i^* \right\vert \\
    &= 2\left\vert \sum_{j} c_j s_j (\sum_{i'} \tilde G_{i' j}\alpha_{i'}^*)  (\sum_i\tilde G_{ij} \alpha_i^*  )\right\vert \\
    &=2 \left\vert \sum_{j} c_j s_j (\sum_i\tilde G_{ij} \alpha_i^*  )^2 \right\vert\\ 
    &\leq 2 \sum_{j} c_j s_j \left\vert\sum_i\tilde G_{ij} \alpha_i^*  \right\vert^2 \\
    &=2 \sum_{j} c_j s_j  \sum_i\tilde G_{ij} \alpha_i^*     \sum_{i'}\tilde G_{i'j}^* \alpha_{i'} \\
    &= 2\sum_{iji'} \tilde G_{ij} c_j s_j \tilde G_{ji'} \alpha_{i'} \alpha_i^* \\
    &= 2 \text{Tr}[\tilde G CS \tilde G \mathcal A] \\
    &\leq  2 \text{Tr}[\tilde G C^2 \tilde G \mathcal A] \\
    &= 2 \text{Tr}[\tilde G S^2 \tilde G \mathcal A] + 2 \text{Tr}[\tilde G ^2 \mathcal A]  .
\end{align}
With
\begin{align}
    \text{Tr}[\tilde G C^2 \tilde G \mathcal A] = \text{Tr}[\tilde G (S^2+1) \tilde G \mathcal A] = \text{Tr}[\tilde G S^2 \tilde G \mathcal A] + \text{Tr}[\tilde G^2 \mathcal A] 
\end{align}
and 
\begin{align}
    \text{Tr}[\tilde G S ^2 \tilde G C^2] = \text{Tr}[\tilde G S ^2 \tilde G S^2] + \text{Tr}[\tilde G^2 S ^2  ]
\end{align}
we finally obtain with all of the results above
\begin{align}
    \mathcal{F}(\vert\psi_\lambda\rangle)& = 4\Big [ \text{Tr}[\tilde G C S \tilde G^*  C S] + \text{Tr}[\tilde G S ^2 \tilde G C^2] \\
     &\quad \quad \quad\quad+ \text{Tr}[\tilde G S^2 \tilde G \mathcal A]+ \text{Tr}[\tilde G C^2 \tilde G\mathcal A] + \text{Tr}[\tilde G C S \tilde G^* \mathcal B^*]+ \text{Tr}[\tilde G^* S C \tilde G \mathcal B]\Big] \\
     & \leq  4\Big [ \mathrm{Tr}[\tilde G S ^2 \tilde G S^2] +  \mathrm{Tr}[G^2 S^2] + \text{Tr}[\tilde G S ^2 \tilde G S^2] + \text{Tr}[\tilde G^2 S ^2  ] \\
      &\quad + \text{Tr}[\tilde G S^2 \tilde G \mathcal A]+ \text{Tr}[\tilde G S^2 \tilde G \mathcal A] +  \text{Tr}[\tilde G^2 \mathcal A]  + 2 \text{Tr}[\tilde G S^2 \tilde G \mathcal A] + 2 \text{Tr}[\tilde G^2 \mathcal A]  \Big]  \\
      &= 4\cdot \Bigg( 2 \mathrm{Tr}[\tilde G S^2 \tilde G S^2] + 4\mathrm{Tr}[\tilde G S^2 \tilde G \mathcal A]  +   3  \text{Tr}[\tilde G^2 (S^2+\mathcal A )] - \text{Tr}[\tilde G^2 S^2 ] \Bigg)  \\
      &= 4\cdot \left( 2 \mathrm{Tr}[\tilde G S^2 \tilde G S^2] + 4\mathrm{Tr}[\tilde G S^2 \tilde G \mathcal A]  \right)  +   4\cdot 3   (\bar g^2+ \Delta g^2)N_S  - 4\text{Tr}[\tilde G^2 S^2 ]
\end{align}
This completes the proof.
\end{proof}

\subsection{Proof of Lemma~\ref{sec4:lemma2}}\label{app4:lemm2}
\begin{proof}
    As $Q$ is Hermitian and positive semi-definite, there is a unitary $W$, such that $Q = W^\dag D W$, where $D= \text{diag}(d_1,\ldots, d_M)$ is diagonal and we have $d_i \geq 0$. With this, let us consider the right-hand side of the to be proved inequality which we divide by $4$ and obtain
    \begin{align}
        \text{RHS}/4 &= \mathrm{Tr}[H Q]^2  +  \mathrm{Tr}[H^2 Q]\mathrm{Tr}[Q]  -2 \mathrm{Tr}[HQHQ] \\
         &=\mathrm{Tr}[H W^\dag D W]^2  +  \mathrm{Tr}[H^2 W^\dag D W]\mathrm{Tr}[W^\dag D W]  - 2  \mathrm{Tr}[HW^\dag D WHW^\dag D W] \\
         &= \mathrm{Tr}[H'  D ]^2  +  \mathrm{Tr}[H'^2   D  ]\mathrm{Tr}[  D  ]  -2  \mathrm{Tr}[H'   D H' D  ]  ,
    \end{align}
    where we introduced $H'=W H W^\dag$, which is Hermitian. Let us now continue by rewriting the above in terms of components. For the first term, we have to calculate the square of the following
    \begin{align}
        \text{Tr}[H' D] = \sum_k (H' D)_{kk} = \sum_k \sum_i H_{ki}' D_{ik} \delta_{ik} = \sum_i H_{ii}' d_i .
    \end{align}
    For the second term, we have 
    \begin{align}
        \mathrm{Tr}[H'^2   D  ] &=  \sum_k (H'^2 D)_{kk} = \sum_k \sum_i (H'^2)_{ki} D_{ik} \delta_{ik} = \sum_i (H'^2)_{ii} d_i = \sum_i \sum_n H_{in}' H_{ni}' d_i \\
        &= \sum_{in} \vert H'_{in}\vert^2 d_i .
    \end{align}
    For the third term, we have
    \begin{align}
        \text{Tr}[H'DH'D] &= \sum_{j} (H'DH'D)_{jj} = \sum_{j}\sum_i (H'D)_{ji} (H'D)_{ij } \\
        &= \sum_{j}\sum_i \sum_r H_{jr}' D_{ri} \sum_q H_{iq}' D_{qj} \\
        &= \sum_{ji} H_{ji}' D_{ii} H_{ij}' D_{jj}\\
        &= \sum_{ij} \vert H'_{ij}\vert^2 d_i d_j
    \end{align}
    Putting these three terms together, we obtain for the right-hand side
    \begin{align}
        \text{RHS}/4 &= \sum_i H'_{ii} d_i \sum_j H'_{jj} d_j + \sum_{in} \vert H_{in}'\vert^2 d_i \sum_j d_j -2 \sum_{ij} \vert H'_{ij}\vert^2 d_i d_j \\
        &= \sum_{ij} \left[ H_{ii}' H_{jj}'+ \sum_n \vert H_{in}'\vert^2 - 2 \vert H_{ij}'\vert^2\right] d_i d_j \\
        &= \sum_{i} \left[ H_{ii}'^2+ \sum_n \vert H_{in}'\vert^2 - 2 \vert H_{ii}'\vert^2\right] d_i^2 \\
        & \quad+\sum_{i\neq j} \left[ H_{ii}' H_{jj}' + \sum_n \vert H_{in}'\vert^2 - 2 \vert H_{ij}'\vert^2\right] d_i d_j .
    \end{align}
    Now, note that we have for the sum over $n$ in the third line
    \begin{align}
        \sum_{i}\sum_n \vert H_{in}'\vert^2 d_i^2 = \sum_i \left[  \vert H_{ii}'\vert^2  +  \sum_{\substack{n\\ n \neq i}}\vert H_{in}'\vert^2  \right] d_i^2
    \end{align}
    and for the sum over $n$ in the fourth line, we have 
    \begin{align}
        \sum_{i\neq j} & \sum_n \vert H_{in}'\vert^2 d_i d_j = \sum_{i\neq j} \sum_n\left( \frac{\vert H_{in}'\vert^2}{2} + \frac{\vert H_{jn}'\vert^2}{2} \right) d_i d_j \\
        &=\sum_{i\neq j}  \left( \frac{\vert H_{ii}'\vert^2}{2} + \frac{\vert H_{ji}'\vert^2}{2} + \frac{\vert H_{ij}'\vert^2}{2} + \frac{\vert H_{jj}'\vert^2}{2} \right) d_i d_j + \sum_{i\neq j} \sum_{\substack{n\\ n\neq i,j}}\left( \frac{\vert H_{in}'\vert^2}{2} + \frac{\vert H_{jn}'\vert^2}{2} \right) d_i d_j \\
        &=\sum_{i\neq j}  \left( \frac{\vert H_{ii}'\vert^2}{2} + \vert H_{ij}'\vert^2+ \frac{\vert H_{jj}'\vert^2}{2} \right) d_i d_j + \sum_{i\neq j} \sum_{\substack{n\\ n\neq i,j}}\left( \frac{\vert H_{in}'\vert^2}{2} + \frac{\vert H_{jn}'\vert^2}{2} \right) d_i d_j 
    \end{align}
    Plugging these expressions into RHS yields
    \begin{align}
        \text{RHS}/4 &= \sum_{i} \underbrace{\left[ H_{ii}'^2+  \vert H_{ii}'\vert^2 - 2 \vert H_{ii}'\vert^2\right]}_{=0} d_i^2 \\
        & \quad+\sum_{i\neq j} \left[ H_{ii}' H_{jj}' +   \frac{\vert H_{ii}'\vert^2}{2} + \vert H_{ij}'\vert^2 + \frac{\vert H_{jj}'\vert^2}{2} - 2 \vert H_{ij}'\vert^2\right] d_i d_j \\
        & \quad + \sum_i \sum_{\substack{n\\ n \neq i}}\vert H_{in}'\vert^2    d_i^2 + \sum_{i\neq j} \sum_{\substack{n\\ n\neq i,j}}\left( \frac{\vert H_{in}'\vert^2}{2} + \frac{\vert H_{jn}'\vert^2}{2} \right) d_i d_j  \\
        &= \sum_{i\neq j} \Bigg( \frac{H_{ii}'}{2} +  \frac{H_{jj}'}{2}\Bigg)^2 d_i d_j  +\sum_{i\neq j} - \vert H_{ij}'\vert^2 d_i d_j \\
        & \quad + \sum_i \sum_{\substack{n\\ n \neq i}}\vert H_{in}'\vert^2    d_i^2 + \sum_{i\neq j} \sum_{\substack{n\\ n\neq i,j}}\left( \frac{\vert H_{in}'\vert^2}{2} + \frac{\vert H_{jn}'\vert^2}{2} \right) d_i d_j 
    \end{align}
    where we used $H_{ii}'^2 = \vert H_{ii}'\vert^2$ in the first line, which is true to Hermiticity. Now, let us note that
    \begin{align}
        \sum_i &\sum_{\substack{n\\ n \neq i}}\vert H_{in}'\vert^2    d_i^2 = \sum_{i\neq n} \vert H_{in}'\vert^2    d_i^2 = \sum_{i\neq j} \vert H_{ij}'\vert^2 d_i^2 \\
        &=  \sum_{i\neq j} \frac{\vert H_{ij}'\vert^2 d_i^2}{2} +  \sum_{i\neq j} \frac{\vert H_{ji}'\vert^2 d_j^2}{2} \\
        &= \sum_{i\neq j} \vert H_{ij}' \vert^2 \frac{d_i^2 + d_j^2}{2} .
    \end{align}
    With this, we finally obtain
    \begin{align}
    \text{RHS}/4   &= \sum_{i\neq j} \Bigg( \frac{H_{ii}'}{2} +  \frac{H_{jj}'}{2}\Bigg)^2 d_i d_j  +\sum_{i\neq j}  \vert H_{ij}'\vert^2 \frac{d_i^2 - 2d_i d_j + d_j^2}{2} \\
        & \quad  + \sum_{i\neq j} \sum_{\substack{n\\ n\neq i,j}}\left( \frac{\vert H_{in}'\vert^2}{2} + \frac{\vert H_{jn}'\vert^2}{2} \right) d_i d_j \\
         &= \sum_{i\neq j} \underbrace{\Bigg( \frac{H_{ii}'}{2} +  \frac{H_{jj}'}{2}\Bigg)^2 d_i d_j}_{\geq0}  +\sum_{i\neq j}  \underbrace{\vert H_{ij}'\vert^2 \frac{(d_i- d_j)^2}{2}}_{\geq0} \\
        & \quad  + \sum_{i\neq j} \sum_{\substack{n\\ n\neq i,j}} \underbrace{\left( \frac{\vert H_{in}'\vert^2}{2} + \frac{\vert H_{jn}'\vert^2}{2} \right) d_i d_j}_{\geq0}  \\
        &\geq 0.
    \end{align}
    Note that $H_{ii}'$ is real (due to Hermiticity of $H$), and $d_i \geq 0$ (due to positive semi-definiteness of $Q$). Thus, the three above terms are in a manifestly non-negative form. Thus, we have $\text{RHS} \geq 0$, which concludes the proof.
\end{proof}

\section{Optimal states}

\subsection{Mean-optimal state}\label{app:MeanOptimalState}
We will now discuss the scenario in which we aim to maximize the prefactor of the term proportional to $\bar g^2$ in the QFI. We refer to such a states as mean optimal.

As an Ansatz, we choose a single-mode squeezed state with zero displacement. Thus, the input state we consider takes the form 
\begin{align}
    \vert \psi\rangle =  \hat V \hat S_{\hat a_n}(r_n)\vert 0\rangle =   \hat S_{\hat V\hat a_n \hat V^\dag}(r_n)\vert 0\rangle.
\end{align}
Thus, the single mode $\hat c_n = \hat V\hat a_n \hat V^\dag$ is squeezed with squeezing parameter $r_n >0$. The other modes are in vacuum, i.e., $r_i=0$ for $i\neq n$.
Also note that $\hat V$ shall be completely arbitrary.

The calculation of the QFI of such a state is straightforward.
\begin{align}
    \mathcal{F}(\vert\psi_\lambda\rangle) &= 4\text{Tr}[\tilde G C S \tilde G^*  C S] +4 \text{Tr}[\tilde G S ^2 \tilde G C^2] \\
    &= 4\sum_{ij} s_i \tilde G_{ij} c_j s_j \tilde G_{ji}^* c_i +4\sum_{ij} c_i \tilde G_{ij} s_j s_j \tilde G_{ji} c_i \\
    &= 4 s_{n} \tilde G_{n n} c_{n} s_{n} \tilde G_{nn}^* c_{n} +4\sum_{i} c_i \tilde G_{in} s_{n} s_{n} \tilde G_{ni} c_i  \\
    &=  4 \vert\tilde G_{n n}\vert^2 s_{n}^2c_{n}^2    +4 \vert\tilde G_{n n}\vert^2 s_{n}^2c_{n}^2 + \sum_{ \substack{i\\ i\neq n}}  \vert\tilde G_{in}\vert^2 s_{n}^2   c_i^2 \\
    &= 8 \vert\tilde G_{n n}\vert^2 s_{n}^2 (s_{n}^2+1) +4\sum_{ \substack{i\\ i\neq n}}  \vert\tilde G_{in}\vert^2 s_{n}^2 \cdot  1 \\
    &= 8 \vert\tilde G_{n n}\vert^2   s_{n}^4 + 4\sum_{ i}  \vert\tilde G_{in}\vert^2 s_{n}^2 .\label{app:eq:MeanOptimalStateQFIlast}
\end{align}
In the third line, we used  $s_i = 0$ for $i\neq n$, in the fourth line we separated the $n$ term from the sum, in the fifth line we used  $c_i^2 = s_i^2+1$ for all $i$ and  $c_i = 1$ for $i\neq n$. 

Now, we want to relate this result for the QFI to the resources $N_S$, $\bar g$, and $\Delta g$.
We have
\begin{align}
    N_S &= \text{Tr}[\tilde P_SS^2]  = \sum_i (\tilde P_S S^2)_{ii} = \sum_{i,j} (\tilde P_S) _{ij} (S^2)_{ji} = \sum_{i,j} (\tilde P_S) _{ij}  s_i^2 \delta_{ji} = \sum_{i}(\tilde P_S) _{ii} s_i^2 \\
    &= (\tilde P_S) _{nn} s_n^2 .
\end{align}
Recall that $\tilde P_S$ is the projector onto the subspace spanned by the eigenvectors of $\tilde G$ associated to non-zero eigenvalues. In the (common) case where $\tilde G$ has only non-zero eigenvalues, $\tilde P_S$ is simply the identity.

For the mean resource parameter we have
\begin{align}
    \bar g N_S &= \text{Tr}[\tilde G S^2] = \sum_{ij} \tilde G_{ij} s_i^2\delta _{ij} = \sum_i \tilde G_{ii} s_i^2  = \tilde G_{nn} s_{n}^2 
\end{align}
We also have
\begin{align}
    \left( \Delta g^2 +\bar g^2 \right) N_S &= \text{Tr}[\tilde G^2 S^2] = \sum_{nm} (\tilde G^2)_{nm} s_n^2\delta _{nm} =  \sum_{n} (\tilde G^2)_{nn} s_n^2 \\
    &= \sum_{in} \tilde G_{ni} \tilde G_{in} s_n^2 = \sum_{in}\vert \tilde G_{in} \vert^2 s_n^2 \\
    &= \sum_{i}\vert \tilde G_{in} \vert^2 s_{n}^2 .
\end{align}
By examining the QFI in Eq.~\eqref{app:eq:MeanOptimalStateQFIlast} and the expressions for  $\bar g N_S$ and $\left( \Delta g^2 +\bar g^2 \right) N_S$, we recognize that the QFI can be rewritten in terms of the resources parameters:
\begin{align}
    \mathcal{F}(\vert\psi_\lambda\rangle) &= 8 \bar g^2 N_S^2 + 0\cdot \Delta g^2 N_S^2 + 4 \bar g^2 N_S + 4(\Delta g)^2 N_S \\
    &= 8 \bar g^2 N_S^2 + 0\cdot \Delta g^2 N_S^2+ O(N_S) .
\end{align}
Thus, an arbitrary single mode squeezed state attains the best prefactor of $8$ in the $\bar g^2$ term, but the worst prefactor of $0$ for the $(\Delta g)^2$ term. It is an optimal state if $\Delta g=0$ or almost optimal in the case of $\bar g^2 \gg (\Delta g)^2$.

\subsection{Variance-optimal state}\label{app:VarianceOptimalState}
\subsubsection{Without displacement}\label{subsub:NoDisplacement}
Now, let us come to a different state that maximizes the prefactor in front of the variance term $\Delta g^2$.

Again, as an Ansatz, we use a state with zero displacement.  It takes the form
\begin{align}
    \vert\psi\rangle &=  \underbrace{\hat B e^{-i \frac{-\phi_n}{2} \hat a_i^\dag \hat b_i}e^{-i \frac{-\phi_j}{2} \hat a_j^\dag \hat a_j} }_{\hat V} \hat S_{ \hat a_{i}}(r) \hat S_{ \hat a_{j}}(r) \vert 0\rangle =   e^{-i \frac{-\phi_n}{2} \hat B\hat a_i^\dag \hat b_i \hat B^\dag}e^{-i \frac{-\phi_j}{2} \hat B\hat b_j^\dag \hat b_j\hat B^\dag}    \hat S_{ \hat B\hat a_{i}\hat B^\dag}(r) \hat S_{ \hat B\hat a_{j}\hat B^\dag}(r) \vert 0\rangle \\
     &=   e^{-i \frac{-\phi_n}{2} \hat b_i^\dag \hat b_i}e^{-i \frac{-\phi_j}{2} \hat b_j^\dag \hat b_j}    \hat S_{  \hat b_{i} }(r) \hat S_{  \hat b_{j} }(r) \vert 0\rangle \\
    &=  \hat S_{e^{-i\phi_i/2}\hat b_{i}}(r) \hat S_{e^{-i\phi_j/2}\hat b_{j}}(r) \vert 0\rangle
\end{align}
Thus, we squeeze two mode in the mode basis $\{\hat b_n\}_n$ in which the generator $\hat G$ appears diagonal. Note that we squeeze both modes $i$ and $j$ equally with squeezing parameters and with squeezing angles $\phi_i,\phi_j$. In the second line we used Eq.~\eqref{app:eq:DefBunitary}. Let us also note we assume $g_i,g_j\neq 0$.

The QFI and the resources of such a state can be calculated quite easily due to the state's simple structure. For the QFI we have
\begin{align}
    \mathcal{F}(\vert\psi_\lambda\rangle) &= 4\text{Tr}[\tilde G C S \tilde G^*  C S] +4 \text{Tr}[\tilde G S ^2 \tilde G C^2] \\
    &= 4\sum_{ij} s_i \tilde G_{ij} c_j s_j \tilde G_{ji}^* c_i +4\sum_{ij} c_i \tilde G_{ij} s_j s_j \tilde G_{ji} c_i \\
    &= 8 \sum_i g_i^2  s_i^2 c_i^2  = 8 g_{i}^2  s_{i}^2 c_{i}^2 +  8 g_{j}^2  s_{j}^2 c_{j}^2  \\
    &= 8 \tilde g_{i}^2 s_{i}^2 (s_{i}^2+1) +  8 g_{j}^2  s_{j}^2 (s_{j}^2+1)  \label{app:eq:VarianceOptimalQFI} .
\end{align}
Note that in the second line we used $V = B  T$, where $T_{nm} = e^{-i \frac{-i\phi_n}{2}} \delta_{nm}$, where $\phi_n = 0$ for $n\neq i,j$. With this, we find $\tilde G = V^\dag G V =  T^\dag B^\dag G B T = T^\dag  D T = D= \text{diag}(g_1,\ldots,g_M)$.

Now, we want to express the QFI in terms of the resources $N_S$, $\bar g$, and $\Delta g$. We have
\begin{align}
    N_S = \text{Tr}[\tilde P_SS^2] =  \text{Tr}[ S^2] = s_{i}^2 + s_{j}^2 \label{app:eq:VarianceOptimalPhoton},
\end{align}
where $\tilde P_SS^2 = S^2$ follows from our initial assumption $g_i,g_j\neq 0$.
For the mean resource, we have
\begin{align}
    \bar g N_S = \text{Tr}[\tilde G S^2] = \sum_{nm} \tilde G_{nm} s_n^2\delta _{nm} &= \sum_n \tilde G_{nn} s_n^2  = \tilde G_{ii} s_{i}^2 + \tilde G_{jj} s_{j}^2 \\
    &= g_{i}  s_{i}^2+ g_{j}  s_{j}^2  \label{app:eq:VarianceOptimalMean}
\end{align}
and  for the variance resource we have
\begin{align}
    \left(  \Delta g^2 +\bar g^2 \right) N_S &= \text{Tr}[\tilde G^2 S^2] = \sum_{nm} (\tilde G^2)_{nm} s_n^2\delta _{nm} =  \sum_{n} (\tilde G^2)_{nn} s_n^2 \\
    &= \sum_{rn} \tilde G_{nr} \tilde G_{rn} s_n^2 = \sum_{rn}\vert \tilde G_{rn} \vert^2 s_n^2 \\
    &= \sum_{r}\vert \tilde G_{ri} \vert^2 s_{i}^2 + \sum_{r}\vert \tilde G_{rj} \vert^2 s_{j}^2 \\
     &=  g_{i}^2   s_{i}^2 + g_{j}^2  s_{j}^2  \label{app:eq:VarianceOptimalVariance}
\end{align}
Let us now set $s_{i}^2= s_{j}^2= N_S/2$. From this follows
\begin{align}
       \Delta g ^2 +\bar g^2    = \frac{g_{i}^2    + g_{j}^2    }{2}
\end{align}
and
\begin{align}
    \mathcal{F}(\vert\psi_\lambda\rangle) &= 8\left(g_{i}^2    + g_{j}^2   \right) \frac{N_S}{2}\left(\frac{N_S}{2}+1\right)   \\
    &= 2  \cdot 2\frac{g_{i}^2   + g_{j}^2  }{2} N_S (N_S + 2)  \\
    &= 4 \left(  \Delta g ^2 +\bar g^2 \right) N_S^2 + 8 \left( \Delta g^2 +\bar g^2 \right) N_S  \\
    &= 4 \left(  \Delta g^2 +\bar g^2 \right) N_S^2 + O(N_S) .
\end{align}
We see that the factor $4$ in front of the bandwidth term $ (\Delta g)^2$  is optimal, while the factor $4$ in front of the $\bar g^2$ term is suboptimal, it is only $50\%$ of the optimal factor $8$. Thus, the state is variance optimal

We also make an important observation for the $N_S \ll 1$ scenario. In this case, the performance of our state is a factor $2$ better than that of the classical coherent state.

\subsubsection{With displacement}\label{subsub:WithDisplacement}

Let us now come to a class of states  that makes use of squeezing and displacement and does not rely on the generator-induced mode basis $\{\hat b_i\}_i$ in contrast to the previous state in \ref{subsub:NoDisplacement}. It relies on populating an initial mode and the corresponding derivative mode, and as we show here, under certain conditions this state can be variance optimal. It has been widely discussed in the literature~\cite{pinel2012ultimate,gessner2023b}.

First recall that in an arbitrary basis $\{ \hat a_n\}_n$ we have
$
    \partial_\lambda\hat a_m^\dag (\lambda) = \partial_\lambda \hat U_\lambda \hat a_m^\dag \hat U_\lambda^\dag =  -i\sum_{n} G_{nm} \hat a_n^\dag (\lambda) .
$
We see  that the derivative mode $\partial_\lambda\hat a_m^\dag$ can be expressed in terms of the initial modes.

Now, let us choose an arbitrary initial mode $\hat u_0$. It is always possible to find a  mode basis $\{\hat u_0,\hat u_1,\ldots \}$ such that the derivative mode of $\hat u_0$ is composed of only two modes
\begin{align}
    \partial_\lambda\hat u_0^\dag (\lambda) = -i [ G_{00} \hat u_0^\dag (\lambda) +  G_{10} \hat u_1^\dag (\lambda) ].
\end{align}
We refer to $\hat u_1$ as the derivative mode of $\hat u_0$, as it is the only mode (other than $\hat u_0$ itself) that has overlap with $\partial_\lambda\hat u_0 (\lambda)$.
Now, suppose we displace  $\hat u_0$ and squeeze $\hat u_0$ and $\hat u_1$, that is, we consider the state
\begin{align}
    \hat V\hat D_{\hat u_0}(\alpha_0) \hat S_{\hat u_0}(r_0)\hat S_{\hat u_1}(r_1)\vert 0\rangle ,
\end{align}
where $\hat V= e^{- i\frac{-\phi_0}{2}\hat u_0^\dag \hat u_0}e^{-i \frac{-\phi_1}{2}\hat u_1^\dag \hat u_1}$ and we have $V_{ni} = e^{-i(-\phi_n)/2}\delta_{ni}$, where $\phi_n$ is related to the squeezing angle.

Let us calculate the QFI and the resources of such a state. (assume $G$ is full rank). Recall that $\tilde G = V^\dag G V$.
The QFI of the state is then
\begin{align}
\mathcal{F}=&4 \Big[  \tilde G_{00}^2 \left(\vert\alpha_0 \vert^2
   \left(c_0^2+s_0^2\right)+c_0 s_0
   \left(\alpha_0^2 +\alpha_0^{*2}+2 c_0 s_0\right)  \right)\\
   &+ \tilde  G_{11}^2 \cdot 2 c_1^2 s_1^2
   \\
   &+\vert \tilde  G_{01}\vert^2 \Big( (s_1^2 +c_1^2)\vert\alpha_0\vert^2+ s_1^2c_0^2 +c_1^2 s_0^2 \\
   &\quad\quad\quad\quad\quad+s_1 c_1  (e^{2i\gamma} \{ \alpha_0^2+c_0s_0\}+e^{-2i\gamma} \{ \alpha_0^{*2}+c_0s_0\})\Big)\Big]
\end{align}
where we used $\tilde G_{01} = G_{01} e^{i\frac{\phi_1-\phi_0}{2}}$ and denoted $\gamma = \frac{\phi_1-\phi_0}{2}+\arg[G_{01}]$ and note that $\tilde G_{00} = G_{00}$ and $\tilde G_{11} = G_{11}$. Choosing $\gamma = 0$  and $\alpha_0\in \mathds{R}$ clearly maximizes the QFI, which is why we will use this parameter setting in the following.

For the resources we have
\begin{align}
    N_S = s_0^2+s_1^2+\vert \alpha_0\vert^2
\end{align}
\begin{align}
    \bar g N_S &= \sum_{n,m}  \tilde G_{nm} (s_n^2 \delta_{nm}+\alpha_n^*\alpha_m) \\
    &= \sum_{n}  \tilde  G_{nn}  s_n^2    + \sum_{n,m}  \tilde G_{nm}  \alpha_n^*\alpha_m \\
    &=\tilde G_{00}s_0^2 + \tilde G_{11}s_1^2 + \tilde G_{00} \vert \alpha_0\vert^2
\end{align}

We also have
\begin{align}
   ( \Delta g^2 + \bar g^2)N_S &= \sum_{n,m}   (\tilde G^2)_{nm} (s_n^2 \delta_{nm}+\alpha_n^*\alpha_m)  \\
   &=  ( \tilde G^2)_{00} s_0^2 + (\tilde G^2)_{11} s_1^2 + (\tilde G^2)_{00}  \vert \alpha_0\vert^2 \\
   &=  \sum_{i} \tilde G_{0i}\tilde G_{i0} s_0^2 + \sum_{i} \tilde G_{1i}\tilde G_{i1} s_1^2 +  \sum_{i} \tilde G_{0i} \tilde G_{i0}   \vert \alpha_0\vert^2 \\
   &= \left[ \vert \tilde G_{00}\vert^2 + \vert \tilde G_{01}\vert^2 \right] (s_0^2+\vert \alpha\vert^2) + \left[ \vert \tilde G_{10}\vert^2 +\vert \tilde G_{11}\vert^2+ \sum_{i\geq 2}\vert G_{1i}\vert^2\right] s_1^2 
\end{align}

Now, assuming $s_0^2=0$, $\vert \alpha_0\vert ^2 = s_1^2 = N_S/2$  and considering the limit $N_S\gg 1$, we obtain for the resources $ \bar g    
     =    (\tilde G_{00}    +  \tilde  G_{11})/2   $, $ \Delta g^2 =\vert \tilde G_{01}\vert^2 +\sum_{i\geq 2} \vert G_{1i}\vert^2/2  + (\tilde G_{00}-\tilde G_{11})^2 /4$ and for the QFI
     \begin{align}
         \mathcal{F} \approx 2 \tilde G_{11}^2  N_S^2 + 4  \left( \Delta g^2- \frac{1}{4}  (\tilde G_{00}-\tilde G_{11})^2  - \frac12 \sum_{i\geq 2}\vert G_{1i}\vert^2\right)   N_S^2 + O(N_S) .
     \end{align}
We see that if $\tilde G_{00}-\tilde G_{11} = G_{00}- G_{11}=0$ and $- \frac12 \sum_{i\geq 2}\vert G_{1i}\vert^2 =0$, then the state is variance optimal with 
\begin{align}
         \mathcal{F} \approx 2  \bar g^2 N_S^2 + 4    \Delta g^2    N_S^2 + O(N_S) .
     \end{align}
     It is easy to see that the modes $\hat u_0= (\hat b_i+\hat b_j)/\sqrt{2}$, the second is $\hat u_1= (\hat b_i-\hat b_j)/\sqrt{2}$ with half of the energy used for displacing one mode and the other to squeeze the other, satisfy both conditions, and thus this setting achieves the variance-optimal QFI given above.
The performance associated to the variance term of the QFI can be attained with a homodyne measurement of the mode $\hat u_1(\lambda)$.

\subsection{Completely optimal state}\label{app:CompletlyOptimalState}
In Sec.~\ref{subsubsec:Optimal} of the main text, we identified the optimal state. We now verify that this state is indeed optimal. Recall that it takes the form
\begin{align}
    \vert\psi\rangle = \hat{S}_{e^{-i\phi_i/2}\hat{b}_{i}}(r_i)\,\hat{S}_{e^{-i\phi_j/2}\hat{b}_{j}}(r_j)\,\vert 0\rangle,
\end{align}
which generalizes the variance-optimal state of Sec.~\ref{app:VarianceOptimalState} by allowing the two squeezing strengths to differ, \textit{i.e}., $r_i \neq r_j$.

First, recall that we assume $g_i< g_j$ and $g_i\neq 0$ and $g_j\neq 0$. Now, assume that $i,j$ exists such that the corresponding eigenvalues are
\begin{align}
    g_i &= \bar g -  \vert \Delta g\vert \sqrt{\frac{s_j^2}{s_i^2}} \\
    g_j &=  \bar g + \vert \Delta g\vert \sqrt{\frac{s_i^2}{s_j^2}} 
\end{align}
and the average photon number of mode $i,j$ are chosen as
\begin{align}
    s_i^2  &=  \frac{N_S}{2} \left( 1 - \frac{\bar g  }{\sqrt{\bar g^2+(\Delta g)^2}} \right)  \\
    s_j^2  &=  \frac{N_S}{2} \left( 1 + \frac{\bar g  }{\sqrt{\bar g^2+(\Delta g)^2}} \right) 
\end{align}
where $N_S = s_i^2+s_j^2$ which follows from $g_i,g_j\neq 0$.

The calculation for the QFI and resources is analogous to the one in the previous Sec.~\ref{app:VarianceOptimalState} and we find
\begin{align}  
    N_S &= \text{Tr}[S^2] = s_{i }^2 + s_{j}^2\\
    \bar g N_S   &=  g_{i}  s_{i}^2+ g_{j}  s_{j}^2   \\
    \left( \Delta g ^2 +\bar g^2 \right) N_S  &=  g_{i}^2 s_{i}^2 +    g_{j}^2  s_{j}^2 
\end{align}
and 
\begin{align}  
    \mathcal{F}(\vert\psi_\lambda\rangle)  = 8   g_{i}^2  s_{i}^2 (s_{i}^2+1) +  8   g_{j}^2 s_{j}^2 (s_{j}^2+1) . \label{app:OptimalStateQFIeigenvalues}
\end{align}

Let us first show that everything is self consistent.
We start with the mean resource:
\begin{align}
    \bar g N_S   &= \left( \bar g -  \vert \Delta g\vert \sqrt{\frac{s_j^2}{s_i^2}} \right) s_i^2 +  \left(  \bar g + \vert \Delta g\vert \sqrt{\frac{s_i^2}{s_j^2}}\right) s_j^2 \\
    &=  \left( \bar g s_i^2 - \vert \Delta g\vert \sqrt{s_j^2 s_i^2} \right) + \left( \bar g s_j^2 + \vert \Delta g\vert \sqrt{s_i^2 s_j^2} \right)  \\
    &= \bar g \left( s_i^2 +s_j^2 \right)
\end{align}
Thus, we have shown self consistence for the mean resource.

Let us continue with the variance resource.
\begin{align}
      \left( \Delta g ^2 +\bar g^2 \right) N_S  &=  g_{i}^2 s_{i}^2 +    g_{j}^2  s_{j}^2 = \left( \bar g -  \vert \Delta g\vert \sqrt{\frac{s_j^2}{s_i^2}} \right)^2 s_i^2 +  \left(  \bar g + \vert \Delta g\vert \sqrt{\frac{s_i^2}{s_j^2}}\right)^2 s_j^2 \\
      &= \left( \bar g^2 - 2 \bar g\vert \Delta g\vert \sqrt{\frac{s_j^2}{s_i^2}} + \Delta g^2 \frac{s_j^2}{s_i^2}\right) s_i^2 +  \left( \bar g^2 + 2 \bar g\vert \Delta g\vert \sqrt{\frac{s_i^2}{s_j^2}} + \Delta g^2 \frac{s_i^2}{s_j^2}\right) s_j^2 \\
      &= \bar g \left( s_i^2+s_j^2 \right) + \Delta g^2 \left( s_i^2+s_j^2 \right) .
\end{align}
Thus, we have shown self consistence for the variance resource.

Now, it only remains to show that the QFI reaches its optimal value. We have
\begin{align}  
    \mathcal{F}  &= 8   g_{i}^2  s_{i}^2 (s_{i}^2+1) +  8   g_{j}^2 s_{j}^2 (s_{j}^2+1) \\
    &= 8\left( \bar g -  \vert \Delta g\vert \sqrt{\frac{s_j^2}{s_i^2}} \right)^2 s_i^2 (s_i^2+1) + 8 \left(  \bar g + \vert \Delta g\vert \sqrt{\frac{s_i^2}{s_j^2}}\right)^2 s_j^2 (s_j^2+1) \\
    &= 8\left( \bar g^2 - 2 \bar g\vert \Delta g\vert \sqrt{\frac{s_j^2}{s_i^2}} + \Delta g^2 \frac{s_j^2}{s_i^2}\right) s_i^2 (s_i^2 +1) + 8 \left( \bar g^2 + 2 \bar g\vert \Delta g\vert \sqrt{\frac{s_i^2}{s_j^2}} + \Delta g^2 \frac{s_i^2}{s_j^2}\right) s_j^2  (s_j^2+1) \\
    &= 8 \bar g^2 \left[  s_i^2  (s_i^2+1)+ s_j^2  (s_j^2+1)\right] + 8\Delta g^2 \left[ s_j^2 (s_i^2+1) +s_i^2 (s_j^2+1) \right] +  8\cdot 2 \bar g \vert \Delta g \vert \sqrt{s_i^2 s_j^2} \left[ s_j^2 + 1 - (s_i^2+1)\right] \\
    &= 8 \bar g^2 \left[  s_i^4 +  s_j^4\right] + 8\cdot 2\Delta g^2   s_i^2 s_j^2   +  8\cdot 2 \bar g \vert \Delta g \vert \sqrt{s_i^2 s_j^2} \left[ s_j^2  - s_i^2\right] + 8 \bar g^2 N_S + 8 \Delta g^2 N_S + 0 .
\end{align}
 To make further progress, let us consider the following terms:
 Let us abbreviate $q = \frac{  \bar g }{\sqrt{\bar g^2+(\Delta g)^2}}$
 \begin{align}
     s_i^4 + s_j^4 &= \left( \frac{N_S}{2} \right)^2 \left(  1 - q\right)^2 + \left( \frac{N_S}{2} \right)^2 \left(  1 + q\right)^2 \\
     &= \frac{N_S^2}{2} \left[ 1+ q^2 \right]
 \end{align}
 Next:
 \begin{align}
     s_i^2 s_j^2 &= \frac{N_S}{2} \left( 1 - q \right) \frac{N_S}{2} \left( 1 + q \right) = \left(\frac{N_S}{2} \right)^2 \left( 1-q^2 \right)\\
     &=  \frac{N_S^2}{4}  \left( 1-q^2 \right)
 \end{align}
 Lastly, let us consider
\begin{align}
    s_j^2 -s_i^2 &= \frac{N_S}{2} \left( 1 + q \right) - \frac{N_S}{2} \left( 1 - q \right) = \frac{N_S}{2} \left(  1 + q - 1 + q\right) = \frac{N_S}{2}  2 q\\
    &= + N_S q
\end{align}
These results can be plugged into the QFI and we obtain:
\begin{align}
    &= 8 \bar g^2  \frac{N_S^2}{2} [1+q^2]+ 8\cdot 2\Delta g^2    \frac{N_S^2}{4}(1-q^2)   +  8\cdot 2 \bar g \vert \Delta g \vert \sqrt{\frac{N_S^2}{4}(1-q^2) }  (+) N_S q + 8( \bar g^2   +  \Delta g^2) N_S \\
    &= 4 \bar g^2  N_S^2 [1+q^2]+ 4 \Delta g^2     N_S^2(1-q^2)   +   4\cdot 2 \bar g \vert \Delta g \vert N_S^2\sqrt{(1-q^2) }    q + 8( \bar g^2   +  \Delta g^2) N_S \\
    &= 4 \bar g^2  N_S^2  + 4 \Delta g^2     N_S^2  +4 \bar g^2  N_S^2 q^2    -4 \Delta g^2     N_S^2q^2 +   4\cdot 2 \bar g \vert \Delta g \vert N_S^2\sqrt{(1-q^2) }      q + 8( \bar g^2   +  \Delta g^2) N_S \\
    &= 8 \bar g^2  N_S^2  + 4 \Delta g^2     N_S^2  +4 \bar g^2  N_S^2 (q^2-1)    -4 \Delta g^2     N_S^2q^2 +   4\cdot 2 \bar g \vert \Delta g \vert N_S^2\sqrt{(1-q^2) }     q + 8( \bar g^2   +  \Delta g^2) N_S
\end{align}
Now, the first two terms in the last line above corresponds to optimal QFI for given resources. To prove optimality, we thus have to show that the other terms containing $q$ expressions cancel each other out to zero. So let us consider:
\begin{align}
     4 &\bar g^2  N_S^2 (q^2-1)    -4 \Delta g^2     N_S^2q^2 +   4\cdot 2 \bar g \vert \Delta g \vert N_S^2\sqrt{(1-q^2) }     q  \\
      &= -4 N_S^2 \left[ \bar g  \frac{\vert\Delta g\vert}{\sqrt{\bar g^2+\Delta g^2}}  - \vert\Delta g\vert \frac{\bar g }{\sqrt{\bar g^2+(\Delta g)^2}}\right]^2 .
\end{align}
 With this, we finally arrive at
\begin{align}
    \mathcal{F} (\vert\psi_\lambda\rangle)=  8 \bar g^2  N_S^2  + 4 \Delta g^2     N_S^2 + 8( \bar g^2   +  \Delta g^2) N_S . \label{app:QFIoptimalState}
\end{align}

\subsection{An idler-assisted optimal state}\label{app:IdlerAssistedOptimalState}
 The optimal states we have considered thus far have not made use of an idler beam. So  let us now consider states that make use of idler modes.
We consider a state with $4$ populated initial modes labelled as $i,i_I,j,j_I$. The state then takes the form
\begin{align}
    \vert\psi\rangle &=   \underbrace{\hat B \hat U_{ii_I} \hat U_{jj_I}}_{=\hat V}\hat S_{\hat a_i}(r_i) \hat S_{\hat a_{i_I}}(r_{i_I})\hat S_{\hat a_j}(r_j) \hat S_{\hat a_{j_I}}(r_{j_I}) \vert 0\rangle .  
\end{align}
The unitaries $ \hat U_{nn_I} $ with $n=i,j$ are beam splitter unitaries that act as 
\begin{align}
    \hat U_{nn+2}^\dag\begin{pmatrix}
       \hat a_n \\ \hat a_{n_I}
    \end{pmatrix} \hat U_{nn_I} 
    = 
    \underbrace{\begin{pmatrix}
        \cos \theta_{nn_I} & -\sin \theta_{nn_I}\, e^{i\phi_{nn_I}} \\
        \sin \theta_{nn_I} \, e^{-i\phi_{nn_I}} & \cos \theta_{nn_I}
    \end{pmatrix}}_{= U_{nn_I}}
    \begin{pmatrix}
       \hat a_n \\ \hat a_{n_I}
    \end{pmatrix} ,
\end{align}
We recall that the QFI is given by
\begin{align}
     \mathcal{F}(\vert\psi_\lambda\rangle) =4 \text{Tr}[\tilde G C S \tilde G^*  C S] + 4\text{Tr}[\tilde G S ^2 \tilde G C^2]  .
\end{align}
In this case, we have 
\begin{align}
    \tilde G = V^\dag G V = U_{jj_I}^\dag U_{ii_I}^\dag  B^\dag G B U_{ii_I} U_{jj_I} .
\end{align}
With this, we can easily calculate the QFI and the relevant resources. Let us make the parameter choices $r_i=r_{i_I}$, $r_j=r_{j_I}$, $\theta_{ii_I}= \pi/4$, $\theta_{jj_I}= \pi/4$. This yields two two-mode squeezed vacuum states between modes $i$ and $i_I$ and modes $j$ and $J_I$. Let us additionally set $\phi_{ii_I}=0$ and $\phi_{jj_I}=0$, this determines the squeezing angles of our two-mode squeezed vacua. With this, we obtain
\begin{align}
    \mathcal{F}(\vert\psi_\lambda\rangle)= 4 \cdot 2 g_i ^2 s_i^2(s_i^2+1) + 4 \cdot 2 g_j^2 s_j^2(s_j^2+1) 
\end{align}
and for the resources
\begin{align}
    N_S = \text{Tr}[\tilde P_S S^2 \tilde P_S] = s_i^2+s_j^2  
\end{align}
and
\begin{align}
    \bar g N_S = g_i  s_i^2 + g_j  s_j^2   
\end{align}
\begin{align}
    (\bar g^2 +  \Delta g ^2) N_S = \text{Tr}[\tilde G^2 S^2] = g_i  ^2 s_i^2 + g_j  ^2 s_j^2   .
\end{align}
The structure  of the QFI and the resources is exactly the same as for the idler-less optimal state in the previous Subsections~\ref{app:VarianceOptimalState} and \ref{app:CompletlyOptimalState}, compare Eqs.~\eqref{app:eq:VarianceOptimalQFI}, \eqref{app:eq:VarianceOptimalPhoton},\eqref{app:eq:VarianceOptimalMean} and \eqref{app:eq:VarianceOptimalVariance} with the above four equations.

Due to the fact that the structure of QFI and resources is identical for the idler-less and idler-assisted state, the discussion about optimality for our idler-assisted strategy is completely analogous to the idler-less case in Sections~\ref{app:VarianceOptimalState} and \ref{app:CompletlyOptimalState}. Thus, also the idler-assisted state presented here can be optimal (and also variance optimal) in the sense that 
\begin{align}
    \mathcal{F}(\vert\psi_\lambda\rangle) = \left(8 \bar g^2 + 4\Delta g^2 \right) N_S^2 +O(N_S) ,
\end{align}
given that we find appropriate modes and populate them according to our findings in  \ref{app:CompletlyOptimalState}.

\section{Optimal measurements}

\subsection{Optimal measurement for variance-optimal and completely optimal state}\label{app:OptimalMeasurementVarianceCompletly}
Let us now present an optimal measurement for the variance-optimal and the completely optimal state given in Eq.~\eqref{eq:VarianceOptimalState} and Eq.~\eqref{eq:CompletlyOptimalState}, respectively. Recall that the state is of the form $ \vert\psi\rangle  =  \hat S_{e^{-i\phi_i/2}\hat b_{i}}(r_i) \hat S_{e^{-i\phi_j/2}\hat b_{j}}(r_j) \vert 0\rangle$.
Now, as both modes $\hat b_i$ and $\hat b_j$ are orthogonal, we can homodyne both of them simultaneously. We measure
\begin{align}
    \hat x_{e^{-i(\phi_n/2+\varphi_n)}\hat b_n} = \frac{e^{-i(\phi_n/2+\varphi_n)}\hat b_n + e^{+i(\phi_n/2+\varphi_n)}\hat b_n^\dag}{\sqrt{2}} 
\end{align}
for $n=i,j$ and $\varphi_n$ is the adjustable phase that determines the quadrature to be measured. We have $ \langle\psi_\lambda\vert \hat x_{e^{-i(\phi_n/2+\varphi_n)}\hat b_n}    \vert\psi_\lambda\rangle =0 $ and
\begin{align}
      \sigma_{\text{hom},n}^2&:=\langle\psi_\lambda\vert  \hat x^2_{e^{-i(\phi_n/2+\varphi_n)}\hat b_n}  \vert\psi_\lambda\rangle =  \langle\psi_\lambda\vert \frac{ \hat b_n^2 e^{-i2(\phi_n/2+\varphi_n)}+\hat b_n^{\dag 2} e^{+i2(\phi_n/2+\varphi_n)} + \hat b_n \hat b_i^\dag + \hat b_n^\dag\hat b_n  }{2} \vert\psi_\lambda\rangle \\
     & =\langle\psi\vert \hat U_\lambda^\dag\frac{ \hat b_n^2 e^{-i2(\phi_n/2+\varphi_n)}+\hat b_n^{\dag 2} e^{+i2(\phi_n/2+\varphi_n)} + \hat b_n \hat b_i^\dag + \hat b_n^\dag\hat b_n  }{2} \hat U_\lambda\vert\psi \rangle \\
     &=  \langle\psi\vert  \frac{ \hat b_n^2 e^{-i2(\phi_n/2+\varphi_n+\lambda g_n)}+\hat b_n^{\dag 2} e^{+i2(\phi_n/2+\varphi_n+\lambda g_n)}  + \hat b_n^\dag\hat b_n + \hat b_n \hat b_n^\dag }{2}  \vert\psi \rangle \\
     &= \frac{s_n c_n e^{-i2(\varphi_n+\lambda g_n)}+s_n c_n e^{+i2(\varphi_n+\lambda g_n)}+s_n^2+s_n^2+1}{2} \\
      &= s_n c_n \cos(2(\varphi_n+\lambda g_n)) + \frac{2s_n^2+1}{2} \\
      &= \frac{\sinh(2r)\cos(2(\varphi_n+\lambda g_n)) }{2} + \frac{\cosh(2r)}{2} .
\end{align}
In the third line we used $\hat U_\lambda^\dag \hat b_n \hat U_\lambda = \hat b_n  e^{-i\lambda g_n}$.

Note that due to the fact that the state is separable with respect to the measured modes $i$ and $j$, the corresponding measurement outcomes are independent $x_n \sim \mathcal{N}(0, \sigma_{\text{hom},n}^2)$.  The FI of the joint Gaussian probability distributions is the sum of the individual mode FIs. With this, we find 
\begin{align}
    F_{\text{hom}}(\vert\psi_\lambda\rangle) &= \sum_{n= \{i,j \}} F_{\text{hom},n} (\vert\psi_\lambda\rangle) \\
    &=  \sum_{n= \{i,j \}}\frac{1}{2} \frac{(\partial_{ \lambda}  \sigma_{\text{hom},n})^2}{ \sigma_{\text{hom},n}^2}  =  \sum_{n= \{i,j \}}\frac{1}{8} \frac{ \partial_{\lambda} ( \sigma_{\text{hom},n}^2)}{ \sigma_{\text{hom},n}^4} \\
    &= \sum_{n= \{i,j \}} \frac{2 g_n^2 \sinh ^2(2 r_n) \sin ^2(2 (g_n \lambda +\varphi_n ))}{(\cosh (2 r_n)+\sinh (2 r_n) \cos (2
   (g_n \lambda +\varphi_n )))^2} .
\end{align}
Now, the optimal measurement phases $\varphi_n^{\text{opt}}$ can be easily determined \cite{monras2006}, and we find
\begin{align}
    \varphi_n^{\text{opt}} = \frac{1}{2} \arccos(\tanh (2 r_n))-g_n \lambda +\frac{\pi}{2}
\end{align}
with which we find the optimal FI
\begin{align}
   \max_{\varphi_{i},\varphi_{j}} F_{\text{hom}} (\vert\psi_\lambda\rangle) = 8 g_{i}^2 s_{i}^2 (s_{i}^2 +1) + 8 g_{j}^2 s_{j}^2 (s_{j}^2 +1) 
\end{align}
which is identical to the QFIs in Eq.~\eqref{app:OptimalStateQFIeigenvalues} and Eq.~\eqref{app:eq:VarianceOptimalQFI} of the optimal state and the variance-optimal state, from which it follows that our proposed homodyne scheme is an optimal measurement for these states.

\subsection{Effect of loss and noise}\label{app:Loss}
Let us also  consider the effect of photon loss. We model photon loss as a beam splitter of transmissivity $\eta_n$ right before homodyne detection that mixes the signal mode with a zero-mean environmental mode in a thermal (Gaussian) state.
In that case, the mean of the homodyne measurement is still zero and the variance   is given by
\begin{align}
    \sigma_{\text{hom},n}^2 (\eta_n) = \eta_n \sigma_{\text{hom},n}^2(0) + (1-\eta_n) \sigma_{\text{env},n}^2 ,
\end{align}
where $\sigma_{\text{hom},n}^2(0) = \sigma_{\text{hom},n}^2$ is the ideal case.

Now, the calculations for the FI are analogous to above, and we find
\begin{align}
    F_{\text{hom}}(\vert\psi_\lambda\rangle) &=    \sum_{n=\{i, j\}}\frac{1}{8} \frac{ \partial_{ \lambda} (\eta_n \sigma_{\text{hom},n} ^2(0))}{\sigma_{\text{hom},n}^4(\eta)}  \\
    &=  \sum_{n=\{i,j \}} \frac{2 \eta_n^2 g_n^2 \sinh ^2(2 r_n) \sin ^2(2 (g_n  \lambda +\varphi_n ))}{( \eta_n \sigma_{\text{hom},n}^2(0) + (1-\eta_n) \sigma_{\text{env},n}^2 )^2} .
\end{align}
The optimal phase is given by
\begin{align}
     \varphi_n^{\text{opt}} = \frac{1}{2} \arccos\left( \frac{\eta_n\sinh(2r_n)}{\eta_n \cosh(2r_n)+(1-\eta_n)\sigma^2_{\text{env},n}}\right)-g_n \lambda +\frac{\pi}{2}.
\end{align}
 For the optimal FI under loss and thermal noise, we have
\begin{align}
    \max_{\varphi_{i},\varphi_{j}} F_{\text{hom}}( \vert\psi_\lambda \rangle) = \sum_{n=\{i,j\}} \frac{2 \eta_n^2 g_n^2 \sinh^2 (2r_n)}{( \eta_n \cosh(2r_n)+(1-\eta_n) \sigma^2_{\text{env},n})^2+\eta_n^2\sinh(2r_n))^4} .
\end{align}
We thus have obtained a general expression for the FI of an optimal homodyne measurement when loss and noise is present. To gain further insights, let us assume $\eta_{i}=\eta_{j} =: \eta$, which is, for example, approximately the case in lidar and radar scenarios.

For the case $\sigma_{\text{env},n} = 1$, which corresponds to an environment in the vacuum state, and the case $\sinh^2(r_n)\gg 1$ for $n=\{i,j\}$, we have
\begin{align}
    \max_{\varphi_{i},\varphi_{j}} F_{\text{hom}}(\vert\psi_\lambda\rangle ) &\approx \frac{2\eta}{1-\eta} \left[ g_{i}^2 \sinh^2(r_{i}) + g_{j}^2 \sinh^2(r_{j}) \right] \\
    &= \frac{2\eta}{1-\eta} \left[ \bar g^2 +\Delta g^2 \right] N_S ,
\end{align}
where we used $\sinh^2(2r)\approx 4\sinh^2(r)$ and $\cosh(2r)\approx 2 \sinh^2(r)$.
This should be compared to the performance of a classical coherent state protocol, which achieves $ \mathcal{F}(\vert\psi_\lambda^{\text{coh}}\rangle) = 4\eta  \left[ \bar g^2 +\Delta g^2 \right] N_S$. We have $\max_{\varphi_{i},\varphi_{j}} F_{\text{hom}}(\vert\psi_\lambda\rangle) \geq \mathcal{F}(\vert\psi_\lambda^{\text{coh}}\rangle)$, and thus a quantum advantage for transmissivities $\eta \geq 1/2$. Thus, our state that is optimal for $\eta=1$ and shows good resilience even when photon loss is present.

Let us now also consider the effect of noise. Assume that the environmental modes are in thermal states with average photon number $N_B$. We then have $\sigma_{\text{env},n} = 2 N_B/(1-\eta) +1 $ (the factor $1/(1-\eta)$ ensures that we have the same thermal noise photon number for all $\eta$). Assume $N_B\gg 1$ to be such that $\sigma_{\text{env},n} \approx 2 N_B/(1-\eta)$. We then find
\begin{align}
     &\max_{\varphi_{i},\varphi_{j}} F_{\text{hom}}(\Delta\lambda =0) = 8\eta^2 \sum_{n=n_0,n_1}  \frac{g_n^2 \sinh^2(r_n)}{8\eta \sinh^2 (r_n) N_B+4N_B^2}\\
     &= \begin{cases}
         \frac{\eta}{N_B} \left[ g_{i}^2 \sinh^2(r_{i}) + g_{n_1}^2 \sinh^2(r_{j}) \right] = \frac{\eta N_S }{N_B}  \left[ \bar g^2 +\Delta g^2 \right] , \quad \text{for } \eta \sinh^2(r_{n}) \gg N_B \\
         2\eta^2  \left[ g_{i}^2 \frac{\sinh^4(r_{i})}{N_B^2} + g_{j}^2 \frac{\sinh^4(r_{j})}{N_B^2} \right] , \quad\quad\quad\quad\quad\quad\quad\quad\quad\quad \text{for } \eta \sinh^2(r_{n}) \ll N_B
     \end{cases}
\end{align}
with $n=\{ i,j\}$.

The upper case is a factor of $2$ lower than  the classical coherent state performance $ \mathcal{F}(\vert\psi_\lambda\rangle) \approx 2 \frac{\eta N_S }{N_B}  \left[ \bar g^2 +\Delta g^2 \right] $. The performance in the second case is much worse as the FI scales as $\sim (\frac{\eta N_S }{N_B})^2  \left[ \bar g^2 +\Delta g^2 \right] $ with $\eta N_S/N_B \ll 1$. Thus, our state that is otherwise optimal for $\eta=1$ and $N_B=0$, is inferior in the case of $\eta <1$ and $N_B\gg 1$ to the classical coherent state.

\subsection{Optimal homodyne measurement for mean-optimal state}\label{app:OptimalMeasurementMean}
Let us now propose an optimal measurement for the mean-optimal single-mode squeezed state $    \vert \psi\rangle =  \hat V \hat S_{\hat a_n}(r_n)\vert 0\rangle =   \hat S_{\hat V\hat a_n \hat V^\dag}(r_n)\vert 0\rangle$. Let us assume we have a sufficiently good prior guess $\lambda_{\text{prior}}$ of the parameter $\lambda$. With out loss of generality, less us assume this prior guess is $\lambda_{\text{prior}}=0$ and $\lambda \approx 0$.
We then have $\hat U_\lambda \simeq \hat{\mathds{1}}-i\lambda \hat G$.

Now, we propose to homodyne the mode $\hat c_n= \hat V\hat a_n\hat V^\dag$, which corresponds to measuring the operator
\begin{align}
    \hat x_{\hat c_ne^{-i\varphi_n}} = \frac{\hat c_i e^{-i\varphi_n}+\hat c_i^\dag e^{+i\varphi_n}}{\sqrt{2}} .
\end{align}
We  find $\langle\psi_\lambda\vert \hat x_{e^{-i(\phi_n/2+\varphi_n)}\hat c_n}    \vert\psi_\lambda\rangle =0$ and 
\begin{align}
      \sigma_{\text{hom},n}^2&:=\langle\psi_\lambda\vert  \hat x^2_{e^{-i(\phi_n/2+\varphi_n)}\hat b_n}  \vert\psi_\lambda\rangle 
      \approx \frac{\sinh(2r)\cos(2(\varphi_n+\tilde G_{nn})) }{2} + \frac{\cosh(2r)}{2}   .
\end{align}
Here, the approximation sign means that we neglect terms of order $\lambda^2$ and higher. Analogously to the previous Sec.~\ref{app:OptimalMeasurementVarianceCompletly}, we find
\begin{align}
    \max_{\varphi_{n} }  F_{\text{hom}} (\vert\psi_\lambda\rangle) = 8 \tilde G_{nn}^2 s_{n}^2 (s_{n}^2 +1) = 8\bar g^2 N_S^2+ 8\bar g^2 N_S ,
\end{align}
where the expression in terms of resources after the last equals sign follows from our results in Sec.~\ref{app:MeanOptimalState}. Compare this to the QFI of this state $\mathcal{F}(\vert\psi_\lambda) = 8\bar g^2 N_S^2+ 8(\bar g^2+\Delta g^2) N_S $. 
We see that the single homodyne measurement we proposed is only sensitive to the mean resource $\bar g$, and not the the variance resource $\Delta g$. In the limit $N_S\gg 1$, the single-mode homodyne measurement is nearly optimal.

\section{Comparison to previous approaches in the literature}\label{app:PreviousLiterature}
In Sec.~\ref{sec:PreviousLiterature} we discussed the approach of Ref.~\cite{giovannetti2006quantum} for (optimal) states  that maximize the QFI for a parameter imprinting unitary $\hat U_\lambda = e^{-i\lambda \hat G}$ in a finite dimensional Hilbert space, where $\hat G$ is a finite dimensional Hermitian operator. Their results do not apply to our case of mode parameter estimation as in that case we are dealing with an infinite dimensional Hilbert space.

To compare their approach with the one presented here, let us introduce a photon number truncation to turn the infinitely dimensional Hilbert space  encountered in mode parameter estimation into a finite one:
\begin{align}
    \hat G = \sum_{i=1}^M g_i \hat b_i^\dag \hat b_i = \sum_{i=1}^M g_i \sum_{n=0}^\infty n \vert n\rangle_i\langle n\vert \rightarrow \sum_{i=1}^M g_i \sum_{n=0}^{N_{\text{cut}}} n \vert n\rangle_i\langle n\vert 
\end{align}
where $\vert n\rangle_i = \hat{b}_i^{\dag n}\vert 0\rangle/\sqrt{n!}$ are Fock states of mode $i$. Without loss of generality, assume that $g_1 = g_{\text{min}}$ and $g_M = g_{\text{max}}$.
Thus, the initial infinite Hilbert space has been truncated to a finite dimensional one, which lets us apply the results of Ref.~\cite{giovannetti2006quantum}. According to  Ref.~\cite{giovannetti2006quantum}, the (optimal) state that maximizes the QFI takes the form
\begin{align}
    \vert\psi\rangle = \frac{\vert N_{\text{cut}}\rangle_1+\vert N_{\text{cut}}\rangle_M}{\sqrt{2}}
\end{align}
which is a superposition of the generator's eigenstates corresponding to the minimal and maximal eigenvalues.

Ref.~\cite{giovannetti2006quantum} did not introduce the notion of resources, but let us calculate the resources of this state according to our formalism and the express the QFI in terms of these. 
We find $N_S = \sum_i \langle\psi\vert \hat b_i^\dag \hat b_i \vert \psi\rangle = N_{\text{cut}}$ for the signal photon number and 
\begin{align}
      \frac{\langle\psi\vert \hat b_i^\dag \hat b_i \vert \psi\rangle}{N_S} =  \frac{\delta_{i,1}+\delta_{i,M}}{2}  .
\end{align}
for the generator intensity. Note that this state has the same generator intensity distribution $\frac{\langle\psi\vert \hat b_i^\dag \hat b_i \vert \psi\rangle}{N_S}$ as our Gaussian variance-optimal state.
We further find
\begin{align}
\bar g&=   \sum_i  g_i^2 \frac{\langle\psi\vert \hat b_i^\dag \hat b_i \vert \psi\rangle}{N_S}  = \frac{g_{\text{min}}+g_{\text{max}}}{2}  \\
    \Delta g^2 + \bar g^2 &= \sum_i  g_i^2 \frac{\langle\psi\vert \hat b_i^\dag \hat b_i \vert \psi\rangle}{N_S} = \frac{g_{\text{min}}^2+g_{\text{max}}^2}{2} .
\end{align}
From this follows $\Delta g^2 = \frac{1}{4}(g_M-g_1)^2$. Now, from Ref.~\cite{giovannetti2006quantum} we know that $\mathcal{F}(\hat U_\lambda\vert \psi\rangle) =(g_{\text{max}}-g_{\text{min}})^2  N_{\text{cut}}^2 $, from which follows 
\begin{align}
    \mathcal{F}(\hat U_\lambda\vert \psi\rangle) =4 \Delta g^2 N_{S}^2   .
\end{align}
Thus, the optimal state proposed in Ref.~\cite{giovannetti2006quantum} when applied to truncated mode parameter estimation is variance optimal. There appear no mean contribution $\sim \bar g^2 N_S^2$ in the . Thus, their state is suboptimal according to our framework.

\section{Continuous mode bases}\label{app:sec:ContinuousModes}
In the main text, we encountered continuous mode bases $\hat D(z)$ satisfying the commutation relation $[\hat D(z),\hat D^\dag(z')] = \delta(z-z')$. We keep the discussion general and do not assign a specific physical meaning to $\hat D(z)$: the variable $z \in \mathbb{R}$ may represent time, frequency, a spatial coordinate, momentum $k$, or any other continuous degree of freedom.

These continuous mode operators can be expanded in terms of discrete ones as
\begin{align}
    \hat D(z) = \sum_{n=0}^\infty \hat d_n\, \Psi_n(z),
\end{align}
where $\Psi_n(z)$ are the mode functions associated with the discrete basis $\{\hat d_n\}_n$ satisfying $[\hat d_n, \hat d_m^\dag] = \delta_{nm}$ \cite{blow_continuum_1990}. The mode functions form a complete orthonormal set, $\int \mathrm{d}z\, \psi_n^*(z)\,\psi_m(z) = \delta_{nm}$, allowing one to pass freely between a continuous and a discrete mode description. Note, however, that the discrete description still involves countably infinitely many modes.

A second important continuous mode basis is $\hat F(p)$, with $p \in \mathbb{R}$, related to $\hat D(z)$ by a Fourier transform:
\begin{align}
    \hat D(z) &= \int \mathrm{d}p\, \hat F(p)\, \frac{e^{-ipz}}{\sqrt{2\pi}}, \label{app:eq:ContinuousD}\\
    \hat F(p) &= \int \mathrm{d}z\, \hat D(z)\, \frac{e^{+ipz}}{\sqrt{2\pi}}. \label{app:eq:ContinuousF}
\end{align}
A multimode squeezed state may be written equivalently in any of these three bases as
\begin{align}
    \vert\psi\rangle &= \exp\!\left(\iint\mathrm{d}z\,\mathrm{d}z'\, f(z,z')\,\hat D^\dag(z)\hat D^\dag(z') - \mathrm{h.c.}\right)\vert 0\rangle, \label{app:eq:MultiModeSqueezedStateZ}\\
    &= \exp\!\left(\iint\mathrm{d}p\,\mathrm{d}p'\, \tilde{f}(p,p')\,\hat F^\dag(p)\hat F^\dag(p') - \mathrm{h.c.}\right)\vert 0\rangle,\\
    &= \exp\!\left(\sum_{n,m} f^{(\psi)}_{nm}\,\hat d_n^\dag \hat d_m^\dag - \mathrm{h.c.}\right)\vert 0\rangle,
\end{align}
with $\tilde{f}(p,p') = \iint\mathrm{d}z\,\mathrm{d}z'\, f(z,z')\,e^{ipz}e^{ip'z'}/(2\pi)$ and $f^{(\psi)}_{nm} = \iint\mathrm{d}z\,\mathrm{d}z'\, f(z,z')\,\psi_n^*(z)\,\psi_m^*(z')$.

The shift transformations are of particular interest and are studied extensively here. The corresponding unitaries are $\hat U_{a_z} = e^{-ia_z \hat G_{a_z}}$ and $\hat U_{a_p} = e^{-ia_p \hat G_{a_p}}$, with generators
\begin{align}
    \hat G_{a_z} &= \int_{-\infty}^\infty \mathrm{d}p\, p\, \hat F^\dag(p)\hat F(p), \label{app:eq:ShiftTransformGeneratorZ}\\
    \hat G_{a_p} &= \int_{-\infty}^\infty \mathrm{d}z\, z\, \hat D^\dag(z)\hat D(z). \label{app:eq:ShiftTransformGeneratorP}
\end{align}
Let us now examine how states transform under these unitaries. One verifies directly that $\hat U_{a_z}^\dag \hat F(p)\hat U_{a_z} = \hat F(p)\,e^{-ipa_z}$ and $\hat U_{a_p}^\dag \hat D(z)\hat U_{a_p} = \hat D(z)\,e^{-iza_p}$. Combining these with Eqs.~\eqref{app:eq:ContinuousD}--\eqref{app:eq:ContinuousF} yields $\hat U_{a_z}^\dag \hat D(z)\hat U_{a_z} = \hat D(z+a_z)$ and $\hat U_{a_p}^\dag \hat F(p)\hat U_{a_p} = \hat F(p-a_p)$. Using $\hat U_{-a_z} = \hat U_{a_z}^\dag$ and $\hat U_{-a_p} = \hat U_{a_p}^\dag$, it further follows that $\hat U_{a_z}\hat D(z)\hat U_{a_z}^\dag = \hat D(z-a_z)$ and $\hat U_{a_p}\hat F(p)\hat U_{a_p}^\dag = \hat F(p+a_p)$. Finally, applying these results and performing a change of integration variables, one finds that $\hat U_{a_z}$ acting on $\vert\psi^{\text{sq}}\rangle$ induces $f(z,z') \to f(z+a_z, z'+a_z)$, while $\hat U_{a_p}$ induces $\tilde{f}(p,p') \to \tilde{f}(p-a_p, p'-a_p)$.

\subsection{Procedure to reduce infinitely many modes to finitely many modes}\label{app:subsec:TruncationProcedure}
As discussed in the main text, we frequently encounter situations involving infinitely many modes, whereas our formalism was developed for a finite number of modes. We now demonstrate how the formalism can nevertheless be applied in this setting.

Since experimental constraints confine the relevant range of the variable $z$ to a finite interval $[z_{\text{min}}, z_{\text{max}}]$, we partition this interval into $M$ equally sized bins of width $\delta z := (z_{\text{max}} - z_{\text{min}})/M$ and define the discretized annihilation operators
\begin{align}
    \hat d_n := \int_{z_{\text{min}}+\delta z\,(n-1)}^{z_{\text{min}}+\delta z\, n} \mathrm{d}z\,\frac{\hat D(z)}{\sqrt{\delta z}}, \qquad n \in \{1,\ldots,M\},
\end{align}
which satisfy $[\hat d_n, \hat d_m^\dag] = \delta_{nm}$, yielding a finite set of $M$ modes with the desired commutation relations.

Discretized versions of $\hat F(p)$ are obtained as follows. The $p$-interval $[p_{\text{min}}, p_{\text{max}}]$ has length $p_{\text{max}} - p_{\text{min}} = 2\pi/\delta z = 2\pi M/(z_{\text{max}} - z_{\text{min}})$, with bin size $\delta p = 2\pi/(z_{\text{max}} - z_{\text{min}})$. Setting $p_k = p_{\text{min}} + (k-1)\,\delta p$ for $k \in \{1,\ldots,M\}$, we define the discrete modes $\{\hat f_k\}_k$ via the discrete Fourier relations
\begin{align}
    \hat d_n &= \sum_{k=1}^M \hat f_k\, \frac{e^{-iz_n p_k}}{\sqrt{M}}, \\
    \hat f_k &= \sum_{n=1}^M \hat d_n\, \frac{e^{+iz_n p_k}}{\sqrt{M}},
\end{align}
which are the discrete analogues of Eqs.~\eqref{app:eq:ContinuousD}--\eqref{app:eq:ContinuousF}, with $z_n = z_{\text{min}} + n\,\delta z$. The transformation matrix $e^{-iz_n p_k}/\sqrt{M}$ is unitary, as can be seen by writing
\begin{align}
\frac{e^{-iz_n p_k}}{\sqrt{M}} = \frac{e^{-i\frac{2\pi}{M}(n-1)(k-1)}}{\sqrt{M}}\, e^{-i\left(z_{\min}p_{\min} + z_{\min}\delta p\,(k-1) + p_{\min}\delta z\,(n-1)\right)}.
\end{align}
The first factor is the standard $M$-point discrete Fourier transform matrix, which is unitary. The remaining exponential factors depend on at most one index, or are global phases, and therefore correspond to multiplication by diagonal unitary matrices on the rows and columns together with an overall phase — operations that preserve unitarity. It follows that $[\hat f_k, \hat f_q^\dag] = \delta_{kq}$.

For the shift transformations studied here, the continuous generators of Eqs.~\eqref{app:eq:ShiftTransformGeneratorZ}--\eqref{app:eq:ShiftTransformGeneratorP} are approximated by
\begin{align}
    \hat G_{a_z} &\approx \sum_{n=1}^M p_n\, \hat f_n^\dag \hat f_n, \label{app:eq:ShiftTransformGeneratorZdisc}\\
    \hat G_{a_p} &\approx \sum_{n=1}^M z_n\, \hat d_n^\dag \hat d_n. \label{app:eq:ShiftTransformGeneratorPdisc}
\end{align}
The quality of this approximation improves with increasing $M$ and decreasing bin width, both of which can be adjusted freely. Correspondingly, the multimode squeezed state of Eq.~\eqref{app:eq:MultiModeSqueezedStateZ} is well approximated in the discretized basis as
\begin{align}
    \vert\psi^{\text{sq}}\rangle &\approx \exp\!\left(\sum_{n,m=1}^M f_{nm}\,\hat d_n^\dag \hat d_m^\dag - \mathrm{h.c.}\right)\vert 0\rangle, \\
    &\approx \exp\!\left(\sum_{n,m=1}^M \tilde{f}_{nm}\,\hat f_n^\dag \hat f_m^\dag - \mathrm{h.c.}\right)\vert 0\rangle,
\end{align}
with $f_{nm} = f(z_n, z_m)$ and $\tilde{f}_{nm} = \tilde{f}(p_n, p_m)$. This procedure reduces an infinite continuous model to a finite discrete one, justifying the application of our formalism to continuous mode settings.

\subsection{Determining the generator matrix}\label{app:HGmodeGenerator}
Given only the mode functions $\psi_n(z;\lambda)$, it is not immediately clear how to recover the corresponding unitary mode transformation $\hat U_\lambda$. However, the generator $\hat G = \sum_{n,m=1}^M G_{nm}\,\hat c_n^\dag \hat c_m$ can be computed directly from the mode functions, and this is sufficient to evaluate the QFI.

To this end, we define
\begin{align}
    \hat d_n(\lambda) := \hat U_\lambda\,\hat d_n\,\hat U_\lambda^\dag
\end{align}
as the parameter-dependent mode operators in the Schrödinger picture, in contrast to the Heisenberg picture transformation rules of Eq.~\eqref{eq:HeisenbergTranformModeOperators}. Following Ref.~\cite{Gessner:23}, one has
\begin{align}
    \partial_\lambda\vert\psi_\lambda\rangle &= \sum_{m=1}^M \partial_\lambda\hat d_m^\dag(\lambda)\,\hat d_m(\lambda)\,\vert\psi_\lambda\rangle \\
    &= -i\sum_{n,m=1}^M G_{nm}\,\hat d_n^\dag(\lambda)\,\hat d_m(\lambda)\,\vert\psi_\lambda\rangle,
\end{align}
from which it follows that
\begin{align}
    \partial_\lambda\hat d_m^\dag(\lambda) = -i\sum_{n=1}^M G_{nm}\,\hat d_n^\dag(\lambda).
\end{align}
Using $\partial_\lambda\hat d_n^\dag(\lambda) = \int\mathrm{d}z\,[\partial_\lambda\Psi_n(z;\lambda)]\,\hat D^\dag(z)$, this yields the key relation
\begin{align}\label{app:eq:GeneratorModeRelation}
    \partial_\lambda\Psi_m(z;\lambda) = -i\sum_{n=1}^M G_{nm}\,\Psi_n(z;\lambda).
\end{align}
In practice, $\partial_\lambda\Psi_n(z;\lambda)$ often satisfies convenient recurrence relations that allow $G_{nm}$ to be computed straightforwardly. The generator matrix elements can also be obtained directly via
\begin{align}
    -iG_{nm} = \int\mathrm{d}z\,\Psi_n^*(z;\lambda)\,\partial_\lambda\Psi_m(z;\lambda).
\end{align}

We now apply this framework to the Hermite–Gauss (HG) modes, which form a complete orthonormal basis on $L^2(\mathbb{R})$. They are defined as
\begin{align}
    \Phi_n(z;\,z_0,p_0,\sigma_z,\theta) = \left(\frac{1}{2\sigma_z^2\pi}\right)^{1/4} \sqrt{\frac{2^{-n}}{n!}}\,H_n\!\left(\frac{z-z_0}{\sqrt{2}\,\sigma_z}\right) e^{-\frac{1}{2}\left(\frac{z-z_0}{\sqrt{2}\,\sigma_z}\right)^2} e^{-ip_0(z-z_0)}\,e^{-i\theta}.
\end{align}
Consider the mode transformation
\begin{align}
    \hat U_{a_z} = \exp\!\left(-ia_z\int\mathrm{d}p\, p\,\hat D^\dag(p)\,\hat D(p)\right),
\end{align}
which generates a shift in the $z$ domain and acts on the mode operators as
\begin{align}
    \hat U_{a_z}^\dag\,\hat D(p)\,\hat U_{a_z} = \hat D(p)\,e^{-ipa_z},
\end{align}
from which it follows that $\hat U_{a_z}^\dag\,\hat D(z)\,\hat U_{a_z} = \hat D(z+a_z)$. The transformed HG mode operators are then
\begin{align}
    \hat d_n(a_z) = \hat U_{a_z}\,\hat d_n\,\hat U_{a_z}^\dag &= \int\mathrm{d}z\,\Phi_n^*(z)\,\hat U_{a_z}\,\hat D(z)\,\hat U_{a_z}^\dag \\
    &= \int\mathrm{d}z\,\Phi_n^*(z)\,\hat D(z-a_z) \\
    &= \int\mathrm{d}z\,\Phi_n^*(z+a_z)\,\hat D(z),
\end{align}
and correspondingly $\hat d_n^\dag(a_z) = \int\mathrm{d}z\,\Phi_n(z+a_z)\,\hat D^\dag(z)$, so that the mode functions transform as $\Phi_n(z) \to \Phi_n(z+a_z)$, i.e., by a simple shift. Differentiating with respect to $a_z$ and using the recurrence relation for HG modes gives
\begin{align}
    \partial_{a_z}\Phi_m(z+a_z) = \Phi_m'(z+a_z) = \frac{1}{\sqrt{2}\,\sigma_z}\left(\sqrt{\frac{m}{2}}\,\Phi_{m-1}(z+a_z) - \sqrt{\frac{m+1}{2}}\,\Phi_{m+1}(z+a_z)\right) - ip_0\,\Phi_m(z+a_z).
\end{align}
Comparing with Eq.~\eqref{app:eq:GeneratorModeRelation}, the generator matrix elements are
\begin{align}
    -iG_{nm} &= \frac{1}{\sqrt{2}\,\sigma_z}\left(\sqrt{\frac{m}{2}}\,\delta_{n,m-1} - \sqrt{\frac{m+1}{2}}\,\delta_{n,m+1}\right) - ip_0\,\delta_{n,m}, \\
    \Leftrightarrow\quad G_{nm} &= \frac{i}{\sqrt{2}\,\sigma_z}\left(\sqrt{\frac{m}{2}}\,\delta_{n,m-1} - \sqrt{\frac{m+1}{2}}\,\delta_{n,m+1}\right) + p_0\,\delta_{n,m}. \label{eq:GeneratorMatrixHG}
\end{align}

\subsection{QFI for a state that appears in product form in the HG mode basis}\label{app:subsec:QFIofHGstate}
Let us consider a state that takes on a simple form in the HG mode basis
\begin{align}
    \vert\psi\rangle = [\otimes_{n=1}^M e^{-i \frac{-\phi_n}{2} \hat c_n^\dag \hat c_n} \hat S_{\hat c_n} (r_n)] \vert 0\rangle 
\end{align}
with $\hat c_n^\dag = \int \mathrm{d}z \,\Phi_n (z) \hat D^\dag (z)$.
The states squeezing matrix is diagonal in the HG mode basis.
Note that $\hat V  =\otimes_n e^{-i \frac{-\phi_n}{2} \hat c_n^\dag \hat c_n} $, where $\phi_n$ is related to the squeezing angle  of mode $n$.   We   can thus rewrite the state in various ways
\begin{align}
    \vert\psi\rangle &=  \hat V [\otimes_n\hat S_{\hat c_n} (r_n) ]\vert 0\rangle = [\otimes_n\hat S_{\hat V\hat c_n\hat V^\dag} (r_n) ]\vert 0\rangle\\
    &= [\otimes_n\hat S_{e^{-i\phi_n/2}\hat c_n} (r_n) ]\vert 0\rangle ,
\end{align}
which follows from $\hat V \hat a_n \hat V^\dag = \sum_i V^{*}_{ni}\hat a_i$ with $V_{ni}=e^{-i(-\phi_n)/2} \delta_{ni}$.

Let us now calculate the QFI for $\hat U_{a_z}\vert\psi\rangle = \vert\psi_{a_z}\rangle$. We have
\begin{equation} 
     \mathcal{F}(\vert \psi_{a_z}\rangle) = 4\sum_{ij} s_i \tilde G_{ij} c_j s_j \tilde G_{ji}^* c_i +4\sum_{ij} c_i \tilde G_{ij} s_j s_j \tilde G_{ji} c_i 
\end{equation}
For simplicity, we assume that only the first two modes are populated, that is, $s_n = 0$ and $c_n=1$ for all $n\neq 0,1$. Using Eq.~\eqref{eq:GeneratorMatrixHG} and $\tilde G = V^\dag G V$, we find
\begin{align}
  \mathcal{F}(\vert \psi_{a_z}\rangle)  =   \frac{8 p_0^2 \sigma_z ^2 \left(c_0^2 s_0^2+c_1^2
   s_1^2\right)+\left(c_0^2+2\right) s_1^2-2 c_0 c_1
   s_0 s_1 \cos (\phi_0  -\phi_1)+c_1^2 s_0^2}{\sigma_z ^2} .
\end{align}
Let us also calculate the resource parameters. We have 
\begin{align}
    N_S = s_0^2 + s_1^2 .
\end{align}
\begin{align}
    \bar g  = \frac{1}{N_S}\sum_{n m} \tilde G_{nm}   \left( s_n^2 \delta_{nm} +\alpha_n^* \alpha_m \right) = +p_0\frac{ s_0^2 + s_1^2}{s_0^2+s_1^2} = +p_0
\end{align}
  and
\begin{align}
   \left( \Delta g^2 +\bar g^2 \right)  &=  \frac{1}{N_S}\sum_{n m} (\tilde G^2)_{nm}   \left( s_n^2 \delta_{nm} +\alpha_n^* \alpha_m \right) \\
   &=   \frac{1}{N_S}\sum_{n m} \sum_i \tilde G_{ni}\tilde G_{im}   \left( s_n^2 \delta_{nm} +\alpha_n^* \alpha_m \right)
\end{align}
so that 
\begin{align}
    \Delta g^2   &=   \frac{1}{N_S}\sum_{n m} \sum_i \tilde G_{ni}\tilde G_{im}   \left( s_n^2 \delta_{nm} +\alpha_n^* \alpha_m \right) - \bar g^2 \\
    &= \frac{s_0^2+3 s_1^2}{4 \sigma_z ^2 \left(s_0^2+s_1^2\right)} .
\end{align}

For simplicity, let us assume that $s_0 = s_1$. We then have
\begin{align}
    \bar g^2 = p_0^2
\end{align}
and 
\begin{align}
    \Delta g^2 = \frac{1}{2\sigma_z^2}.
\end{align}
For the QFI, we find with $\phi_0-\phi_1 = \pi$ (these phases clearly maximize the QFI):
\begin{align}
 \mathcal{F}(\vert \psi_\lambda\rangle) &=   2 N_S \left(2 \bar g^2 (N_S+2)+\Delta g^2
   (N_S+3)\right) \\
   &= \left(4 \bar g^2 + 2\Delta g^2 \right) N_S^2 + \underbrace{\left(8 \bar g^2 +6 \Delta g^2 \right)N_S}_{=O(N_S)} .
\end{align}
We see that this state is suboptimal. Both prefactors of the mean and variance terms are a factor of $1/2$ worse than the maximum possible value.

This example demonstrates how to calculate the QFI and the resource parameters for a generic state. This state is suboptimal as seen by the prefactors $4$ and $2$ instead of the maximum values $8$ and $4$.

\section{Analytical calculation of QFI for a regularized optimal state}\label{app:QFIregularized}
Let us calculate the QFI of a regularized state. In the main text in Sec.~\ref{sec:OptimalStatesTimeFrequency} we introduced such a regularized state for time estimation with squeezing matrix given in Eq.~\eqref{eq:OptimalRegularizedState}.
In this Appendix, however, we are going to be more general and instead consider the estimation of shifts $a_z$ in the $z$ domain induced by the unitary mode transform $\hat U_{a_z} = e^{-a_z i \hat G_{a_z}}$. Here, $z$ is a general variable that could stand for time, frequency, space, $k$ space, etc.

Let us consider a squeezing matrix as in Eq.~\eqref{eq:OptimalRegularizedState} in the $z$ and $p$ domains:
\begin{align}
    f(z, z') &= r_{-} \Phi_{-,0}(z) \Phi_{-,0} (  z') + r_{+} \Phi_{+,0}(z) \Phi_{+,0} (  z') \label{app:eq:SqueezingFunctionRegularizedSpace}\\
    \tilde f(p,p') &= r_{-} \tilde \Phi_{-,0}(p) \tilde \Phi_{-,0} (  p') + r_{+} \tilde \Phi_{+,0}(p) \tilde \Phi_{+,0} (  p') 
\end{align}
where we denote
\begin{align}
    \Phi_{\pm,n} (z) := \Phi_{n} (z;  z_0 \pm \delta_{\pm,z},  p_0\pm \delta_{\pm,p},\sigma_z,\theta) \label{app:eq:PhiPmDef}
\end{align}
with $\delta_{\pm,k} \geq 0$ and $\delta_{\pm,p} \geq 0$. Thus, the HG mode functions $\Phi_{\pm,n}$ (defined in Eq.~\eqref{eq:DefinitionHGmode}) are centered around $ z_0 \pm \delta_{\pm,z}$ in $z$ space and centered around $ p_0\pm \delta_{\pm,p}$ in $p$ space. Also, from Eq.~\eqref{eq:OptSqueezing1} and \eqref{eq:OptSqueezing2} we know that $r_-\leq r_+$ for optimal states. Let us here additionally assume the stricter $r_- < r_+$ so that the case $r_-\neq r_+$ is excluded.

Now, to compute the QFI given in Eq.~\eqref{eq:QFIcomponents} of the above state for the mode transform with generator $\hat G_{a_z} = \int\mathrm{d}p \, p \hat F^\dag (p)\hat F(p)$, we have to determine the components of the generator matrix $\tilde G = V^\dag G V$ in the Schmidt mode basis, that is, the basis in which the state appears disentangled.

So let us first determine the Schmidt decomposition of $f(z,z')$. First we denote
\begin{align}
    \mathcal{S} = \int \mathrm{d}z \, \Phi_{+,0}^*(z)\Phi_{-,0}(z) .
\end{align}
Now, let us introduce the orthonormal vectors
\begin{align}
    u_1(z) &= \Phi_{+,0}(z) \\
    u_2(z) &= \frac{\Phi_{-,0}(z)- \vert\mathcal{S}\vert\Phi_{+,0}(z)}{\sqrt{1-\vert\mathcal{S}\vert^2}} .
\end{align}
We can invert this and obtain
\begin{align}
    \Phi_{+,0}(z) &= u_1(z) \\
    \Phi_{-,0}(z) &= u_2 (z) \sqrt{1-\vert\mathcal{S}\vert^2} + \vert\mathcal{S}\vert u_1(z).
\end{align}
We can now rewrite $f(z,z')$ in the orthonormal basis
\begin{align}
    f(z,z') &= r_- \left[ u_2 (z) \sqrt{1-\vert\mathcal{S}\vert^2} + \vert\mathcal{S}\vert u_1(z)\right] \left[ u_2 (z') \sqrt{1-\vert\mathcal{S}\vert^2} + \vert\mathcal{S}\vert u_1(z')\right] + r_+ u_1(z) u_1(z') \\
    &=  \left( r_+ + r_{-} \vert \mathcal{S}\vert^2\right) u_1(z) u_1(z') + r_-\left( 1-\vert\mathcal{S}\vert^2\right) u_2(z)u_2(z') \\
    &\quad + r_-\left( \vert \mathcal{S}\vert \sqrt{1-\vert\mathcal{S}\vert^2}\right) u_1(z)u_2(z') +r_- \left( \vert \mathcal{S}\vert \sqrt{1-\vert\mathcal{S}\vert^2}\right) u_1(z')u_2(z) \\
    &= \begin{pmatrix}
        u_1(z')\\ u_2(z')
    \end{pmatrix}^T \underbrace{\begin{pmatrix}
        r_+ + r_{-} \vert \mathcal{S}\vert^2  & r_-\vert \mathcal{S}\vert \sqrt{1-\vert\mathcal{S}\vert^2}\\
       r_- \vert \mathcal{S}\vert \sqrt{1-\vert\mathcal{S}\vert^2} & r_- \left( 1-\vert\mathcal{S}\vert^2\right)
    \end{pmatrix}}_{=A}\begin{pmatrix}
        u_1(z)\\ u_2(z)
    \end{pmatrix} .
\end{align}
It is straightforward to show that the matrix' $A$ eigenvalues are
\begin{align}
   r_1 &= \frac{1}{2} \left(r_-+r_+ +\sqrt{4 r_- r_+
   S^2+(r_+-r_-)^2}\right) \\
    r_2 &= \frac{1}{2} \left(r_-+r_+ -\sqrt{4 r_- r_+
   S^2+(r_+-r_-)^2}\right) .
\end{align}
Note that in the limit $\mathcal{S}\rightarrow 0$, we have $r_{1} \rightarrow r_+$ and $r_{2} \rightarrow r_-$ as expected, which follows from our initial assumption $r_+> r_-$.

Note that $A$ is real symmetric, and thus it can be diagonalized by the orthogonal matrix
\begin{align}
    Q = \begin{pmatrix}
        \cos\chi & -\sin\chi\\
        \sin\chi & \cos\chi
    \end{pmatrix} ,
\end{align} 
that is $Q^T A Q= \text{diag}(r_1, r_2)$. For $Q^T A Q$ to be diagonal, we get the following condition:
\begin{align}
   0  \overset{!}{=}  (Q^T A Q) _{12} &= \cos(\chi) \sin(\chi) (A_{22}- A_{11})+A_{12} (\cos(\chi)-\sin(\chi)) (\cos(\chi)+\sin(\chi)) \\
    &= \cos(\chi) \sin(\chi) (A_{22}- A_{11})+A_{12}  (\cos^2(\chi)-\sin^2(\chi))  \\
    \quad\Leftrightarrow \tan (2\chi) &= -2 \frac{A_{12}}{A_{22}-A_{11}} \\
    &= -2 \frac{r_-\vert \mathcal{S}\vert \sqrt{1-\vert\mathcal{S}\vert^2}}{r_- \left( 1-\vert\mathcal{S}\vert^2\right)-\left[  r_+ + r_{-} \vert \mathcal{S}\vert^2 \right]} \\
    &=   2 \frac{r_-\vert \mathcal{S}\vert \sqrt{1-\vert\mathcal{S}\vert^2}}{r_+-r_- + 2r_- \vert\mathcal{S}\vert^2}
\end{align}
where we used $\cos^2(\chi)-\sin^2(\chi)= \cos(2\chi)$ and $\cos(\chi) \sin(\chi) = \frac{1}{2} \sin(2\chi)$ and assumed $A_{22}-A_{11}\neq 0$.

Now, there are two regimes: 1) $\vert r_+-r_-\vert \ll 2r_- \vert\mathcal{S}\vert^2$. If we additionally assume $\vert S\vert\ll 1$, we find $\chi =\pi/2$. 2) $\vert r_+-r_-\vert \gg 2r_- \vert\mathcal{S}\vert^2$. If we additionally  assume $\vert S\vert\ll 1$, we find $\tan (2\chi) \approx 2 \frac{r_- \vert S\vert }{r_+-r_-} \approx 0$ from which follows $\chi \approx \frac{r_- \vert S\vert }{r_+-r_-}$ up to first order in $\vert\mathcal{S}\vert$.

It is now evident that the first two  Schmidt modes  $\Psi_1(z)$ and  $\Psi_2(z)$ with non-zero Schmidt coefficients defined as 
\begin{align}
    f(z,z') = \sum_{n=1}^2 r_n \Psi_n (z)\Psi_n (z')
\end{align} 
are given by
\begin{align}
   \begin{pmatrix}
        \Psi_1(z) \\
        \Psi_2 (z)
    \end{pmatrix} = \begin{pmatrix}
        \cos\chi & -\sin\chi\\
        \sin\chi & \cos\chi
    \end{pmatrix} \begin{pmatrix}
        u_1(z) \\
        u_2 (z)
    \end{pmatrix} .
\end{align}
We will exclusively consider the regime in which $\vert\mathcal{S}\vert$ is sufficiently small so that 
\begin{align} 
    \Psi_1 (z) & = \cos(\chi) u_1(z) -\sin (\chi) u_2(z) \approx \cos(\chi)\Phi_{+,0}(z)  -\sin (\chi) \Phi_{-,0}(z) \label{app:eq:ApproxRelationSchmidtToIntialModes1} \\
     \Psi_2 (z) & = \cos (\chi) u_2(z) + \sin(\chi) u_1(z) \approx \cos (\chi) \Phi_{-,0}(z) + \sin(\chi) \Phi_{+,0}(z)  \label{app:eq:ApproxRelationSchmidtToIntialModes2}
\end{align}
is a  good approximation.

 These are just the first two elements of the Schmidt mode basis. The rest of the elements could be constructed via a Gram-Schmidt orthogonalization procedure of the set $\{ \Psi_1,\Psi_2,\Phi_{+,1},\Phi_{-,1},\Phi_{+,2},\Phi_{-,2},\ldots \}$. With similar arguments as before, we find for the third and fourth element 
\begin{align}
      \Psi_3 (z) &\approx \Phi_{+,1}(z)  \\
     \Psi_4 (z) &\approx \Phi_{-,1}(z) 
\end{align}
with negligible corrections of order $\vert S\vert$.

The next step is to determine the components of the generator matrix $\tilde G$ in the Schmidt mode basis $\{\Psi_n\}_n$. Let us first denote 
\begin{align}
    G_{nm}^{\Phi_{\pm}} &=  +\frac{i}{\sqrt{2}\sigma_z} \left( \sqrt{\frac{m}{2}}  \delta_{n,m-1} - \sqrt{\frac{m+1}{2}}  \delta_{n,m+1} \right)   + (p_0 \pm\delta_{\pm,p}) \delta_{n,m} \label{app:eq:Gphi}
\end{align}
as the generator matrix in the mode bases $\{\Phi_{\pm,n}(z)\}_n$.

We then find
\begin{align}
    &\partial_{a_z} \Psi_{1}(z+a_z)  \approx   \cos(\chi)\partial_{a_z}\Phi_{+,0}(z+a_z)  -\sin (\chi)\partial_{a_z} \Phi_{-,0}(z+a_z) \\
    &=   \cos(\chi)\left[ -i G_{00}^{\phi_+} \Phi_{+,0} -i G_{10}^{\Phi_+} \Phi_{+,1}\right] -\sin(\chi)\left[ -i G_{00}^{\Phi_-} \Phi_{-,0} -i G_{10}^{\Phi_-} \Phi_{-,1} \right]   \\
    &\approx   \cos(\chi)\left[ -i G_{00}^{\Phi_+} (\cos(\chi) \Psi_1 + \sin (\chi)\Psi_2) -i G_{10}^{\Phi_+} \Psi_3\right] -\sin(\chi)\left[ -i G_{00}^{\Phi_-} (\cos(\chi) \Psi_2 - \sin (\chi)\Psi_1) -i G_{10}^{\Phi_-} \Psi_4 \right]   \\
    &= -i \underbrace{\left[ G_{00}^{\Phi_+} \cos^2 (\chi) + G_{00}^{\Phi_-}\sin^2 (\chi)\right]}_{\approx \tilde G_{11}} \Psi_1 - i \underbrace{\left[ G_{00}^{\Phi_+}- G_{00}^{\Phi_-}\right]\cos(\chi)\sin(\chi)}_{ \approx \tilde G_{21}} \Psi_2 - i \underbrace{G_{10}^{\Phi_+} \cos(\chi)}_{\approx \tilde G_{31}} \Psi_3  - i \underbrace{G_{10}^{\Phi_-} (-\sin(\chi))}_{\approx\tilde G_{41}} \Psi_4 ,
\end{align}
where we used $\Phi_{+,0} \approx \cos(\chi) \Psi_1 + \sin (\chi)\Psi_2$ and $\Phi_{-,0} \approx \cos(\chi) \Psi_2 - \sin (\chi)\Psi_1$, which follows from Eqs.~\eqref{app:eq:ApproxRelationSchmidtToIntialModes1} and \eqref{app:eq:ApproxRelationSchmidtToIntialModes2} by inversion, and Eq.~\eqref{app:eq:GeneratorModeRelation} to determine elements of $\tilde G$.  
Completely analogously, we find
\begin{align}
    &\partial_{a_z} \Psi_{2}(z+a_z) \approx \cos (\chi) \partial_{a_z} \Phi_{-,0}(z+a_z) + \sin(\chi) \partial_{a_z}\Phi_{+,0}(z+a_z)  \\
    &= \cos (\chi)  \left[ -i G_{00}^{\Phi_-} \Phi_{-,0} - i G_{10}^{\Phi_-} \Phi_{-,1} \right] + \sin(\chi) \left[ -i G_{00}^{\phi_+} \Phi_{+,0} -i G_{10}^{\Phi_+} \Phi_{+,1}\right] \\
    &\approx \cos (\chi)  \left[ -i G_{00}^{\Phi_-} \left[  \cos(\chi) \Psi_2 - \sin (\chi)\Psi_1\right] - i G_{10}^{\Phi_-} \Psi_4 \right] + \sin(\chi) \left[ -i G_{00}^{\phi_+}  \left[   \cos(\chi) \Psi_1 + \sin (\chi)\Psi_2 \right] -i G_{10}^{\Phi_+} \Psi_3\right] \\
    & = -i \underbrace{\left[  G_{00}^{\Phi_-} \cos^2 (\chi) +  G_{00}^{\Phi_+} \sin^2 (\chi)\right]}_{\approx \tilde G_{22}} \Psi_2 - i \underbrace{\left[ G_{00}^{\Phi_+}-G_{00}^{\Phi_-}\right]\cos(\chi)\sin(\chi)}_{\approx \tilde G_{12}} \Psi_1 -i \underbrace{G_{10}^{\Phi_-} \cos(\chi)}_{\approx \tilde G_{42}}\Psi_4 -i \underbrace{G_{10}^{\Phi_+} \sin(\chi)}_{\approx \tilde G_{32}}\Psi_3
\end{align}
From above, we can read off the matrix elements  of $\tilde G$.

Now, let us calculate the QFI. We use the fact that only the two modes $\Psi_1$ and $\Psi_2$ are populated, and thus $s_n=0$ and $c_n=1$ for $n\neq 1,2$.
\begin{align} 
     \mathcal{F}(\vert \psi_\lambda\rangle) &= 4\sum_{i,j=1}^\infty s_i \tilde G_{ij} c_j s_j (\tilde G_{ji})^* c_i +4\sum_{i,j=1}^\infty c_i \tilde G_{ij} s_j s_j \tilde G_{ji} c_i \\
     &= 4 \left[ 2 \vert G_{11}\vert^2 s_1^4 + 2 \vert G_{22}\vert^2 s_2^4 + 2 s_1^2 s_2^2 ( \vert G_{12}\vert^2 + \text{Re}[G_{12}^2]) \right] +O(N_S) \\
     &= 8 \left[  \vert G_{11}\vert^2 s_1^4 +  \vert G_{22}\vert^2 s_2^4 +  s_1^2 s_2^2 ( \vert G_{12}\vert^2 + \text{Re}[G_{12}^2]) \right] +O(N_S) .
\end{align}
Let us now also calculate the resource parameters and then express the QFI in terms of these. We have
\begin{align}
    N_S = s_1^2+ s_2^2
\end{align}
and
\begin{align}
    \bar g  &= \frac{1}{N_S}\sum_{n m}  \tilde G_{nm}   \left( s_n^2 \delta_{nm} +\alpha_n^* \alpha_m \right) = \frac{\tilde G_{11}  s_1^2 + \tilde G_{22} s_2^2}{s_1^2+s_2^2} .
\end{align}
For the variance term, we have
\begin{align}
   \left( \Delta g^2 +\bar g^2 \right)  &=  \frac{1}{N_S}\sum_{n m} (\tilde G^2)_{nm}   \left( s_n^2 \delta_{nm} +\alpha_n^* \alpha_m \right) \\
   &=   \frac{1}{N_S}\sum_{n m} \sum_i \tilde G_{ni}\tilde G_{im}   \left( s_n^2 \delta_{nm} +\alpha_n^* \alpha_m \right) \\
   &= \frac{1}{N_S}  \sum_{n=1}^2\sum_{i=1}^\infty \tilde G_{ni} \tilde G_{in}  s_n^2   = \frac{1}{N_S}  \sum_{n=1}^2\sum_{i=1}^\infty \vert \tilde G_{ni} \vert^2  s_n^2 \\
   &=  \frac{1}{N_S}   \sum_{i=1}^\infty \left[ \vert \tilde G_{1i} \vert^2  s_1^2  +  \vert \tilde G_{2i} \vert^2  s_2^2 \right]  \\
   &= \frac{1}{N_S} \left[\left( \vert \tilde G_{11} \vert^2 + \vert\tilde G_{12}\vert^2 + \vert\tilde G_{13}\vert^2 + \vert\tilde G_{14}\vert^2  \right)  s_1^2 + \left( \vert \tilde G_{21}\vert^2 + \vert \tilde G_{22}\vert^2  + \vert \tilde G_{23}\vert^2 + \vert \tilde G_{24}\vert^2 \right) s_2^2     \right] .
\end{align}

\subsubsection{QFI of fully optimal regularized state}\label{app:FullyOptimalRegularizedQFI}
Above, we have given the calculation of the QFI of  state with squeezing  matrix as given in Eq.~\eqref{app:eq:SqueezingFunctionRegularizedSpace}. 

Let us now evaluate the QFI of a fully optimal state for which we generally have $r_1\neq r_2$. Let us assume we are in regime 2)  with $\vert r_+-r_-\vert \gg 2r_- \vert\mathcal{S}\vert^2$. Assuming $\vert\mathcal{S}\vert$ is small enough, we may approximate $\sin(\chi)\approx 0$ and $\cos (\chi)\approx1$.

 Inspired by the choice of modes in Eqs.~\eqref{eq:opteigenval1} and \eqref{eq:opteigenval2} for the optimal state in with a finite  number of modes presented in Sec.~\ref{app:CompletlyOptimalState}, we set
\begin{align}
    \delta_{+,p} &=  \Delta \cdot \sqrt{\frac{s_2^2}{s_1^2}}   \\
     \delta_{-,p} &=  \Delta \cdot \sqrt{\frac{s_1^2}{s_2^2}} ,
\end{align}
where $\Delta >0$.  

Now, we find
\begin{align}
    \tilde G_{11} &\approx G_{00}^{\Phi_+} = p_0+\delta_{+,p}  \\
    \tilde G_{22} &\approx G_{00}^{\Phi_-} = p_0-\delta_{-,p} \\
     \tilde G_{12}   &\approx 0 + O(\vert\mathcal{S}\vert) \\
      \tilde G_{23}   &\approx 0 + O(\vert\mathcal{S}\vert) \\
      \tilde G_{14}   &\approx 0 + O(\vert\mathcal{S}\vert) \\
      \tilde G_{31}   &\approx  G_{10}^{\Phi_+} = +  \frac{i}{\sqrt{2}\sigma}  \frac{1}{\sqrt{2}} = +  \frac{i}{2\sigma_z}   \\
       \tilde  G_{42}   &\approx G_{10}^{\Phi_-}= + \frac{i}{\sqrt{2}\sigma}  \frac{1}{\sqrt{2}} = + \frac{i}{2\sigma_z}   .
\end{align}
Completely analogous to the finite case in Sec.~\ref{app:CompletlyOptimalState}, we find
\begin{align}
    \bar g = + p_0
\end{align}
and 
\begin{align}
     \Delta g^2 + \bar g^2 &= \frac{1}{N_S} \left[\left( (p_0+\delta_{+,p})^2 +  0 + \frac{1}{4\sigma_z^2} +0\right)  + \left( 0 + (p_0-\delta_{-,p})^2 +  0+\frac{1}{4\sigma_z^2}  \right)      \right] \frac{N_S}{2} \\
     & = p_0^2 + \Delta^2 + \frac{1}{4\sigma_z^2}
\end{align}
from which follows
\begin{align}
    \Delta g^2 = \Delta ^2 +\frac{1}{4\sigma_z^2} .
\end{align}
In contrast to the finite dimensional case in Sec.~\ref{app:CompletlyOptimalState}, here the extra term $\frac{1}{4\sigma_z^2}$ appears due to regularization, which we assume to be negligible.

Until now, we have not specified the squeezing parameter. Inspired by Eqs.~\eqref{eq:OptSqueezing1} and \eqref{eq:OptSqueezing2}, we choose $r_1,r_2$ such that
\begin{align}
    s_i^2  &=  \frac{N_S}{2} \left( 1 - \frac{\bar g  }{\sqrt{\bar g^2+\Delta  ^2}} \right)  \\
    s_j^2  &=  \frac{N_S}{2} \left( 1 + \frac{\bar g  }{\sqrt{\bar g^2+\Delta  ^2}} \right) . 
\end{align}
With this, the calculation for the QFI is equivalent to the one in Sec.~\ref{app:CompletlyOptimalState}
and we find
\begin{align}
    \mathcal{F} &= \left( 8 p_0^2  + 4 \Delta^2 \right) N_S^2 + O(N_S) \\
    &=  \left( 8 p_0^2  + 4 \left[\Delta g^2 - \frac{1}{4\sigma_z^2}\right] \right) N_S^2 + O(N_S) ,
\end{align}
where again compared to the finite dimensional case in Sec.~\ref{app:CompletlyOptimalState} the extra term $\frac{1}{4\sigma_z^2}$ appears due to regularization.

\subsubsection{QFI of variance optimal regularized state}\label{app:VarianceOptimalRegularizedQFI}
Let us consider the regularized variance-optimal state. A variance-optimal state employs equal squeezing in both modes, i.e., $r_+=r_-$ from which follows $r_1=r_2$. We also have $\delta_{+,p}=\delta_{-,p}=:\delta$. In this regime, we have $\chi \approx \pi/4$ (assuming $\vert\mathcal{S}\vert$ sufficiently small) and thus $\cos(\chi)\approx\sin(\chi) \approx 1/\sqrt{2}$. The photon number in both populated Schmidt modes is then
\begin{align}
    s_1^2  =  s_2^2 = N_S/2 .
\end{align}
Now, we find that $\tilde G_{11} = \tilde G_{22} = (G_{00}^{\Phi_+}+ G_{00}^{\Phi_-}) (\frac{1}{\sqrt{2}})^2= (p_0 + \delta + p_0 -\delta)\frac{1}{2} = p_0$. From this follows
\begin{align}
    \bar g \approx +p_0 .
\end{align}
Note that we have $\tilde G_{12}=\tilde G_{21} \approx  (G_{00}^{\Phi_+} - G_{00}^{\Phi_+}) (\frac{1}{\sqrt{2}})^2= (p_0+\delta -(p_0-\delta)) \frac{1}{2} = \delta$. We also have $\vert \tilde G_{13}\vert^2 \approx \vert \tilde G_{14}\vert^2 \approx \vert \tilde G_{24}\vert^2\approx \vert \tilde G_{23}\vert^2 \approx \vert G_{10}^{\Phi_{\pm}}\vert^2 (\frac{1}{\sqrt{2}})^2 = \frac{1}{2\sigma_z^2} (\sqrt{\frac{1}{2}})^2(\sqrt{\frac{1}{2}})^2 = \frac{1}{2^3\sigma^2} $

We then find
\begin{align}
    \Delta g^2 + \bar g^2 &= \frac{1}{N_S} \left[\left( \vert \tilde G_{11} \vert^2 + \vert\tilde G_{12}\vert^2 + \vert\tilde G_{13}\vert^2 + \vert\tilde G_{14}\vert^2  \right)  s_1^2 + \left( \vert \tilde G_{21}\vert^2 + \vert \tilde G_{22}\vert^2  + \vert \tilde G_{23}\vert^2 + \vert \tilde G_{24}\vert^2 \right) s_2^2     \right] \\
    &=\frac{1}{N_S} \left[ \left( p_0^2 +\delta^2+ \frac{1}{2^3\sigma^2} + \frac{1}{2^3\sigma^2}  \right) s_1^2 + \left( p_0^2 +\delta^2+ \frac{1}{2^3\sigma^2} + \frac{1}{2^3\sigma^2}  \right) s_2^2\right] \\
    &=\frac{1}{N_S} \frac{N_S}{2} \left[ \left( p_0^2 +\delta^2  + \frac{1}{2^2\sigma^2}  \right)   + \left( p_0^2 +\delta^2+\frac{1}{2^2\sigma^2}  \right)  \right] \\
    &= p_0^2 + \delta^2 +\frac{1}{2^2\sigma^2} 
\end{align}
so that
\begin{align}
   & \Delta g^2 =  \delta^2 + \frac{1}{4\sigma_z^2}    \\
    &\Leftrightarrow \delta^2 = \Delta g^2 -\frac{1}{4\sigma_z^2} .
\end{align}
Now, the QFI is 
\begin{align}
    \mathcal{F}(\vert\psi_\lambda\rangle) &\approx 8 \left[  \vert G_{11}\vert^2 s_1^4 +  \vert G_{22}\vert^2 s_2^4 +  s_1^2 s_2^2 ( \vert G_{12}\vert^2 + \text{Re}[G_{12}^2]) \right] +O(N_S) \\
    &= 8 \frac{N_S^2}{4} \left[ p_0^2 +p_0^2+ \delta^2+\delta^2\right] + O(N_S) = 4  N_S^2\left[ p_0^2+\delta^2\right] + O(N_S) \\
    &= 4 N_S^2 \left[ \bar g^2 + \Delta g^2 - \frac{1}{4\sigma_z^2} \right]  + O(N_S)
\end{align}
In the limit of $\Delta g^2 \gg \frac{1}{4\sigma_z^2}$, the state is variance optimal.

Above we have assumed that $r_1=r_2$ from which follows $\chi\approx\pi/4$. Due to experimental imperfections, the equality $r_1-r_2 $ may be slightly perturbed so that $\vert r_1-r_2 \vert =\epsilon$. If $\vert \mathcal{S}\vert$ is small enough so that $\epsilon \gg 2r_- \vert\mathcal{S}\vert^2$ we have $\chi \ll 1$ so that $\sin(\chi)\approx 0 $ and $\cos (\chi)\approx 1$. In that case, the state is still variance optimal, as can be shown with a similar calculation to the one presented in the previous Section \ref{app:FullyOptimalRegularizedQFI}.

\section{Optimal homodyne measurement for regularized state}\label{app:OptimalMeasurementRegularizedState}

Let us now propose an optimal measurement for the regularized optimal two-mode squeezed state discussed in Appendix~\ref{app:QFIregularized} for the mode transform $\hat U_{a_z} = e^{-a_z i \hat G_{a_z}}$. The measurement scheme is analogous to the optimal homodyne scheme for the discrete case presented in \ref{app:OptimalMeasurementVarianceCompletly}.

The squeezing matrix  is $ f(z, z') = r_{-} \Phi_{-,0}(z) \Phi_{-,0} (  z') + r_{+} \Phi_{+,0}(z) \Phi_{+,0} (  z')$  with $\Phi_{\pm,n}(z)$ as defined in Eq.~\eqref{app:eq:PhiPmDef}.
Recall that the two modes $\Phi_{+,0}(z)$ and $\Phi_{-,0}(z)$ are orthogonal to a very good approximation under   our assumptions given in \ref{app:QFIregularized}.

 We denote $\hat c_{\pm,n}^\dag ( a_z) = \int \mathrm{d}z \, \Phi_{\pm,n}(z+a_z) \hat D^\dag(z)$ and note that the transformed state is given by $\vert \psi_{a_z}\rangle \approx \hat S_{\hat c_{+,0}(a_z)}(r_+) \hat S_{\hat c_{-,0}(a_z)}(r_-)\vert 0\rangle$, where the approximate sign is due to the fact that the modes are only approximately orthogonal.

 Now, let us assume we have a prior guess $a_{z,\text{prior}}$ of $a_z$. 
 We propose to homodyne the modes $\hat c_{+,0}(a_{z,\text{prior}})$ and $\hat c_{+,0}(a_{z,\text{prior}})$. We assume that our prior guess $a_{z,\text{prior}} = a_z +\epsilon$ is good enough so that 
 \begin{align}
     \hat c_{\pm,0}(a_{z,\text{prior}}) &\approx \hat c_{\pm,0}(a_{z}) + \epsilon \partial_{y}  \hat c_{\pm,0}(y)\Big\vert_{y=a_{z}} \\
     &=\hat c_{\pm,0}(a_{z}) + \epsilon  \, (-i)^* \left[ G^{\Phi_{\pm}*}_{00}   \hat c_{\pm,0}(a_{z}) + G^{\Phi_{\pm}*}_{10}   \hat c_{\pm,1}(a_{z})   \right] \\
     &=\hat c_{\pm,0}(a_{z}) + \epsilon  \, (-i)^* \left[ (p_0\pm \delta_{\pm,p})  \hat c_{\pm,0}(a_{z}) +  \frac{-i}{\sqrt{2}\sigma_z} \sqrt{\frac{1}{2}} \hat c_{\pm,1}(a_{z})   \right]  \\
     &\approx \hat c_{\pm,0}(a_{z}) e^{+i\epsilon (p_0\pm \delta_{\pm,p})}  + \epsilon     \frac{1}{2\sigma_z}  \hat c_{\pm,1}(a_{z}) , 
 \end{align}
 where we used Eq.~\eqref{app:eq:Gphi} and Eq.~\eqref{app:eq:GeneratorModeRelation}. 

 Now, we propose to measure the two quadratures
 \begin{align}
     \hat x_{ \hat c_{\pm,0}(a_{z,\text{prior}})e^{-i\varphi_\pm}} := \frac{\hat c_{\pm,0}(a_{z,\text{prior}})e^{-i\varphi_\pm}+ \hat c_{\pm,0}^\dag(a_{z,\text{prior}})e^{+i\varphi_\pm}}{\sqrt{2}}
 \end{align}

Now, from this point on,  we neglect all terms of order $\epsilon^2$, which is legitimate assuming we have a good enough prior. Then,  the analysis is completely analogous to the analysis of the finite mode setting discussed in the Appendix~\ref{app:OptimalMeasurementVarianceCompletly}. This is due to the fact that the expectation of terms like $\hat c_{\pm,1}$, $\hat c_{\pm,1}\hat c_{\pm,0}$, $\hat c_{\pm,1}\hat c_{\mp,0}$ and $\hat c_{\pm,1}\hat c_{\mp,1}$  is zero, as our squeezed state has no population in $\hat c_{\pm,1}$. The expectation of terms quadratic in $\hat c_{\pm,1}$ is non zero, however, these terms are always accompanied by a prefactor $\epsilon^2$, which is why we can neglect them.
Thus, the only relevant operators in the analysis become $\hat c_{\pm,0}(a_{z}) e^{+i\epsilon (p_0\pm \delta_{\pm,p})} $ like in Appendix~\ref{app:OptimalMeasurementVarianceCompletly}.

Repeating all the steps from Appendix~\ref{app:OptimalMeasurementVarianceCompletly} then leads us to 
\begin{align}
   \max_{\varphi_{+},\varphi_{-}} F_{\text{hom}} (\vert\psi_{a_z}\rangle) = 8 (p_0 + \delta_{+,p})^2 s_{+}^2 (s_{+}^2 +1) + 8 (p_0- \delta_{-,p})^2 s_{-}^2 (s_{-}^2 +1) 
\end{align}
with $s_\pm =\sinh(r_\pm)$. 

Now by using the relations between the parameters $p_0$, $\delta_{\pm,p}$ and the resource parameters $\bar p$, $\Delta p$ given in \ref{app:FullyOptimalRegularizedQFI} for the optimal state and given in \ref{app:VarianceOptimalRegularizedQFI}, we can show that the above FI equals its respective QFI. Thus, the homodyne measurement scheme presented here is optimal for the regularized optimal state from \ref{app:FullyOptimalRegularizedQFI} and the variance-optimal state from \ref{app:VarianceOptimalRegularizedQFI}.

We can also introduce loss by modeling it as a beam splitter of transmissivity $\eta_{\pm}$, and the analysis is completely analogous to the one presented in Appendix~\ref{app:Loss}.

\section{Variance-optimal measurement for shift mode transforms}
Recall from Definition~\ref{def:OptimalMeasurements} that
variance-optimal measurements are measurements whose FI fully attains the variance contribution of the QFI. So, say we have a state $\vert\psi_\lambda\rangle$ whose QFI is $\mathcal{F}(\vert\psi_\lambda\rangle) = c_{\bar g} \bar g^2 N_S^2 + c_{\Delta g} \Delta g^2 N_S^2 + O(N_S)$. A variance-optimal measurement then attains an FI of $c_{\bar g}' \bar g^2 N_S^2 + c_{\Delta g}' \Delta g^2 N_S^2 + O(N_S)$ with $c_{\Delta g}' =c_{\Delta g}$ and $c_{\bar g}' \leq c_{\bar g}$.

Let us consider the mode parameter shift transformation in the $z$ domain. The analysis of a shift in the $p$ domain is analogous due to Fourier duality.
So we consider
\begin{align}
    \hat U_{a_z} = e^{-i a_z \int\mathrm{d}p \, p \hat F^\dag (p)\hat F(p)}
\end{align}
which has the effect
\begin{align}
    \hat U_{a_z}^\dag \hat F (p) \hat U_{a_z} &= \hat F (p) e^{-ia_z p}\\
    \hat U_{a_z}^\dag \hat D(z) \hat U_{a_z} &= \hat D (z+a_z) .
\end{align}
Also recall the unitary of  a shift in the $p$ domain with unitary $\hat U_{a_p} = e^{-i a_p \int\mathrm{d}z \, z \hat D^\dag (z)\hat D(z)}$ that acts as $    \hat U_{a_p}^\dag \hat D(z)\hat U_{a_p} = \hat D (z) e^{-ia_p z}$ and $\hat U_{a_p}^\dag \hat F (p)\hat U_{a_p} = \hat F(p-a_p)$.

Now, any state $\vert \psi\rangle$ is characterized by mean parameters $\bar z,\bar p$, one in $z$ and the other in its Fourier dual $p$ space.  We have
\begin{align}
    \bar z &= \int \mathrm{d}z \, z \frac{\langle \psi  \vert \hat D^\dag (z)\hat D(z)\vert \psi \rangle}{N_S} \\
    \bar p &= \int \mathrm{d}p \, p \frac{\langle \psi  \vert \hat F^\dag (p)\hat F(p)\vert \psi \rangle}{N_S}
\end{align}
So, let us denote 
\begin{align}
    \vert \psi\rangle =: \vert\psi (\bar z,\bar p)\rangle = \hat U_{-\bar z} \hat U_{+\bar p} \vert \psi (0,0)\rangle ,
\end{align}
explicitly indicating the mean resource parameters of the respective state.
Here, $\vert\psi (0,0)\rangle$ is a version of the state  $\vert\psi (\bar z,\bar p)\rangle$ whose mean resource parameters are both zero. Note that
\begin{align}
    \vert \psi (0,0)\rangle = \hat U_{-\bar p} \hat U_{+\bar z} \vert \psi (\bar z,\bar p)\rangle
\end{align}

Let us show this:
\begin{align}
    &\int \mathrm{d}z \, z \langle \psi (0,0)  \vert \hat D^\dag (z)\hat D(z)\vert \psi(0,0) \rangle \\
    &= \int \mathrm{d}z \, z \langle \psi (\bar  z,\bar p)  \vert \hat U_{+\bar z}^\dag \hat U_{-\bar p}^\dag \hat D^\dag (z)\hat D(z)\hat U_{-\bar p} \hat U_{+\bar z} \vert \psi (\bar z,\bar p)\rangle \\
    &=   \int \mathrm{d}z \, z \langle \psi (\bar  z,\bar p)  \vert \hat U_{+\bar z}^\dag  \hat D^\dag (z)\hat D(z)  \hat U_{+\bar z} \vert \psi (\bar z,\bar p)\rangle e^{-i\bar p z} e^{+i\bar p z} \\
    &= \int \mathrm{d}z \, z \langle \psi (\bar  z,\bar p)  \vert \hat U_{+\bar z}^\dag  \hat D^\dag (z)\hat U_{+\bar z}  \hat U_{+\bar z}^\dag \hat D(z)  \hat U_{+\bar z} \vert \psi (\bar z,\bar p)\rangle  \\
    &= \int \mathrm{d}z \, z \langle \psi (\bar  z,\bar p)  \vert   \hat D^\dag (z+\bar z)    \hat D(z+\bar z)   \vert \psi (\bar z,\bar p)\rangle  \\
    &=  \int \mathrm{d}z \, (z-\bar z) \langle \psi (\bar  z,\bar p)  \vert   \hat D^\dag (z )    \hat D(z )   \vert \psi (\bar z,\bar p)\rangle  \\
    &= (\bar z  - \bar z) N_S \\
    &=0
\end{align}

Similarly, we have
\begin{align}
    &\int \mathrm{d}p \, p  \langle \psi (0,0)  \vert \hat F^\dag (p)\hat F(p)\vert \psi (0,0) \rangle  \\
    &= \int \mathrm{d}p \, p  \langle \psi  (\bar z,\bar p)   \vert  \hat U_{+\bar z}^\dag \hat U_{-\bar p}^\dag \hat F^\dag (p)\hat F(p)\hat U_{-\bar p} \hat U_{+\bar z} \vert \psi (\bar z,\bar p)\rangle \\
    &= \int \mathrm{d}p \, p  \langle \psi  (\bar z,\bar p)   \vert  \hat U_{+\bar z}^\dag   \hat F^\dag (p- (-\bar p))\hat F(p- (-\bar p))  \hat U_{+\bar z} \vert \psi (\bar z,\bar p)\rangle \\
      &= \int \mathrm{d}p \, p  \langle \psi  (\bar z,\bar p)   \vert    \hat F^\dag (p + \bar p)\hat F(p + \bar p)    \vert \psi (\bar z,\bar p)\rangle e^{-i \bar z (p+\bar p)} e^{+i \bar z (p+\bar p)} \\
      &= \int \mathrm{d}p \, (p- \bar p)  \langle \psi  (\bar z,\bar p)   \vert    \hat F^\dag (p )\hat F(p )    \vert \psi (\bar z,\bar p)\rangle  \\
      &= (\bar p - \bar p) N_S\\
      &=0 .
\end{align}

\subsection{The variance-optimal measurement condition}\label{app:VarianceOptimalMeasurementCondition}
In the following we want to show that resolved photon counting in the $z$ domain can be optimal if the state satisfies the variance-optimal measurement condition we are about to introduce. Resolved photon counting in the $z$ domain corresponds to projecting the quantum state onto 
 \begin{align}
     \vert \boldsymbol z\rangle =  \frac{1}{\sqrt{K!}}\otimes_{i=1}^K \hat D^\dag (z_i) \vert 0\rangle
 \end{align}
 where $\boldsymbol z= (z_1,\ldots, z_K)^T \in \mathds{R}^K$ and $K\in\mathds{N}_0$.

The variance-optimal measurement condition is as follows:
\begin{lemma} \label{lemma:VarianceOptimal}
 If a quantum state of interest $\vert \psi(\bar z,\bar p )\rangle$ satisfies 
      \begin{align}
         \partial_{a_z}\arg [ \langle \boldsymbol{z} \vert \hat U_{a_z}\vert \psi(\bar z,0)\rangle ] =0
      \end{align}
     for all $\boldsymbol z \in \mathds{R}^K$ and all $K\in \mathds{N}$, then
     \begin{align}
         F(\hat U_{a_z}\vert \psi(\bar z,\bar p)\rangle,\{ \vert \boldsymbol z\rangle \langle \boldsymbol z\vert \}) = \mathcal{F}(\hat U_{a_z} \vert\psi(\bar z, 0)).
     \end{align}
\end{lemma}
Let us explain the above lemma. The condition is essentially that the measurement amplitudes' phase is independent of the parameter $a_z$ to be estimated. If a state satisfies this condition, then the corresponding FI of this state $\hat U_{a_z}\vert \psi(\bar z,\bar p)\rangle $ (centered at $\bar p$ in the $p$ domain) for the measurement $\{ \vert \boldsymbol z\rangle \langle \boldsymbol z\vert \}$ equals the QFI of the state $\hat U_{a_z}\vert \psi(\bar z,0)\rangle $ centered at $0$ in the $p$ domain. To clarify what this means, suppose that the state under consideration has a QFI of $\mathcal{F}(\hat U_{a_z} \vert\psi(\bar z, 0)) = c_{\bar p} \bar p^2 N_S^2 + c_{\Delta p}\Delta p^2 N_S^2+ O(N_S)$ (see Eq.~\eqref{eq:QFIoptimalityDef}). If this state satisfies our condition, then the FI of this state for the  measurement $\{ \vert \boldsymbol z\rangle \langle \boldsymbol z\vert \}$ is $F(\hat U_{a_z}\vert \psi(\bar z,\bar p)\rangle,\{ \vert \boldsymbol z\rangle \langle \boldsymbol z\vert \}) =  c_{\Delta p}\Delta p^2 N_S^2+ O(N_S)$. That is, the FI attains the variance term of the QFI, but does not attain the mean term. Such a measurement is called variance optimal according to our Definition~\ref{def:OptimalMeasurements}.\\

Let us now proof lemma~\ref{lemma:VarianceOptimal}.
\begin{proof}
Let us consider
\begin{align}
    &\langle \boldsymbol z \vert  \hat U_{a_z} \vert \psi(\bar z,\bar p)\rangle =  \langle \boldsymbol z \vert \hat U_{a_z} \hat U_{-\bar z}\hat U_{\bar p}   \vert  \psi(0, 0)\rangle  = \langle \boldsymbol z \vert \hat U_{a_z-\bar z}  \hat U_{\bar p}   \vert  \psi(0, 0)\rangle\\
    &= \frac{1}{\sqrt{K!}} \langle 0\vert [\otimes_{i=1}^K \hat D(z_i)]\hat U_{a_z-\bar z}  \hat U_{\bar p}    \vert  \psi(0, 0)\rangle= \frac{1}{\sqrt{K!}} \underbrace{\langle 0\vert \hat U_{a_z-\bar z}}_{=\langle 0\vert} [\otimes_{i=1}^K\hat U_{a_z-\bar z}^\dag \hat D(z_i)\hat U_{a_z-\bar z}]  \hat U_{\bar p}   \vert  \psi(0, 0)\rangle \\
    = &\frac{1}{\sqrt{K!}}  \langle 0\vert   [\otimes_{i=1}^K    \hat D(z_i+a_z -\bar z)  ]  \hat U_{\bar p}   \vert  \psi(0, 0)\rangle = \frac{1}{\sqrt{K!}}  \underbrace{\langle 0\vert\hat U_{\bar p}}_{=\langle 0\vert}   [\otimes_{i=1}^K  \hat U_{\bar p}^\dag  \hat D(z_i+a_z -\bar z)  \hat U_{\bar p} ]     \vert  \psi(0, 0)\rangle \\
    &= \frac{1}{\sqrt{K!}}   \langle 0\vert\hat      [\otimes_{i=1}^K   e^{-i\bar p (z_i+a_z-\bar z)}  \hat D(z_i+a_z -\bar z)  ]     \vert  \psi(0, 0)\rangle \\
    &= \left(\prod_{i=1}^K  e^{-i\bar p (z_i+a_z-\bar z)} \right) \frac{1}{\sqrt{K!}}   \langle 0\vert      [\otimes_{i=1}^K   \hat D(z_i)  ]   \hat U_{a_z -\bar z}   \vert  \psi(0, 0)\rangle \\
    &=   \left(\prod_{i=1}^K  e^{-i\bar p (z_i+a_z-\bar z)} \right)   \langle \boldsymbol z \vert  \hat U_{a_z}   \vert  \psi(\bar z, 0)\rangle
\end{align}
From the above, we see that $\bar p$ only enter as a phase into the amplitudes, and thus the measurement probabilities will be independent of $\bar p$, i.e., 
\begin{align}
    \vert \langle \boldsymbol z \vert  \hat U_{a_z} \vert \psi(\bar z,\bar p)\rangle\vert^2 = \vert  \langle \boldsymbol z \vert  \hat U_{a_z}   \vert  \psi(\bar z, 0)\rangle\vert^2 =: p(\boldsymbol{z}\vert a_z) .
\end{align}
Let us denote
\begin{align}
  \langle \boldsymbol z \vert  \hat U_{a_z}   \vert  \psi(\bar z, 0)\rangle =: q e^{i\xi}  ,
\end{align}
where we defined $q = \vert\langle \boldsymbol z \vert  \hat U_{a_z}   \vert  \psi(\bar z, 0)\rangle \vert$ as the magnitude and $\xi =\arg [ \langle \boldsymbol z \vert  \hat U_{a_z}   \vert  \psi(\bar z, 0)\rangle]$ as the phase. The argument function is defined as$\arg [z] = \arg [\vert z\vert e^{i\phi}] = \phi$.

Next, we want to calculate the FI of the state $    \hat U_{a_z} \vert \psi(\bar z,\bar p)\rangle$ using resolved photon counting in the $z$ domain. First note that we have
\begin{align}
     F(\hat U_{a_z}\vert \psi(\bar z,\bar p)\rangle,\{ \vert \boldsymbol z\rangle \langle \boldsymbol z\vert \}) &=  F(\hat U_{a_z}\vert \psi(\bar z,0)\rangle,\{ \vert \boldsymbol z\rangle \langle \boldsymbol z\vert \}) \\
     &= \int\kern-1.35em\sum \mathrm{d}^k z \frac{(\partial_{a_z}p(\boldsymbol{z}\vert a_z))^2}{p(\boldsymbol{z}\vert a_z)} = \int\kern-1.35em\sum \mathrm{d}^k z \frac{(\partial_{a_z} q^2)^2}{q^2}=   \int\kern-1.35em\sum \mathrm{d}^k z \frac{( 2 q \partial_{a_z} q)^2}{q^2}  =   4 \int\kern-1.35em\sum \mathrm{d}^k z  (\partial_{a_z}q)^2 \\ 
     &=   4 \int\kern-1.35em\sum \mathrm{d}^k z  \left(\partial_{a_z} \big\vert  \langle \boldsymbol z \vert  \hat U_{a_z}   \vert  \psi(\bar z, 0)\rangle \big\vert \right)^2 \\
     &=     4 \int\kern-1.35em\sum \mathrm{d}^k z    \big\vert  \langle \boldsymbol z \vert  \partial_{a_z}\hat U_{a_z}   \vert  \psi(\bar z, 0)\rangle \big\vert^2 =   4      \langle\psi (\bar z,0) \vert(\partial_{a_z}\hat U_{a_z})^\dag \left[ \int\kern-1.35em\sum \mathrm{d}^k z \vert \boldsymbol z\rangle \langle \boldsymbol z \vert \right] \partial_{a_z}\hat U_{a_z}   \vert  \psi(\bar z, 0)\rangle   \\
     &=4      \langle\psi (\bar z,0) \vert(\partial_{a_z}\hat U_{a_z})^\dag \partial_{a_z}\hat U_{a_z}   \vert  \psi(\bar z, 0)\rangle \\
     &= 4      \langle\psi (\bar z,0) \vert(\partial_{a_z}\hat U_{a_z})^\dag \partial_{a_z}\hat U_{a_z}   \vert  \psi(\bar z, 0)\rangle  - 4       \underbrace{ \vert \langle\psi (\bar z,0) \vert \hat U_{a_z}^\dag \partial_{a_z}\hat U_{a_z}   \vert  \psi(\bar z, 0)\rangle\vert^2}_{=0} \\
     &= \mathcal{F}(\hat U_{a_z} \vert\psi(\bar z, 0))
\end{align}
where the symbol $\int\kern-1.0em\sum \mathrm{d}^k z$ means that we sum the integrals over all the photon number subspaces. In the fourth line we used $\left(\partial_{a_z} \big\vert  \langle \boldsymbol z \vert  \hat U_{a_z}   \vert  \psi(\bar z, 0)\rangle \big\vert \right)^2 = \big\vert  \langle \boldsymbol z \vert  \partial_{a_z}\hat U_{a_z}   \vert  \psi(\bar z, 0)\rangle \big\vert^2$. This follows from $\langle \boldsymbol z \vert \partial_{a_z} \hat U_{a_z} \vert \psi(\bar z,0 )\rangle = \partial_{a_z} \langle \boldsymbol z \vert \hat U_{a_z} \vert \psi(\bar z,0 )\rangle = \partial_{a_z} (q e^{i\xi}) = e^{i\xi} \partial_{a_z} q = e^{i\xi} \partial_{a_z}\vert \langle \boldsymbol z \vert  \hat U_{a_z} \vert \psi(\bar z,0 )\rangle\vert$. Taking the squared modulus $\vert \cdot \vert^2$ of this chain of equalities and noting that $\vert e^{i\xi} \partial_{a_z}\vert \langle \boldsymbol z \vert  \hat U_{a_z} \vert \psi(\bar z,0 )\rangle\vert \vert^2 = \vert   \partial_{a_z}\vert \langle \boldsymbol z \vert  \hat U_{a_z} \vert \psi(\bar z,0 )\rangle\vert \vert^2 = (  \partial_{a_z}\vert \langle \boldsymbol z \vert  \hat U_{a_z} \vert \psi(\bar z,0 )\rangle\vert )^2 $ proves the used relation.
We also used in the fourth line we used $\int\kern-1.0em\sum \mathrm{d}^k z \vert \boldsymbol z\rangle \langle \boldsymbol z\vert = \hat{\mathds{1}}$. In the sixth line we used $ \vert \langle\psi (\bar z,0) \vert \hat U_{a_z}^\dag \partial_{a_z}\hat U_{a_z}   \vert  \psi(\bar z, 0)\rangle\vert^2 = 0$ which follows from $ 0 = \partial_{a_z} \int\kern-1.0em\sum \mathrm{d}^k z\,  p(\boldsymbol z\vert a_z) =   \int\kern-1.0em\sum \mathrm{d}^k z\, \partial_{a_z} q^2 =   \int\kern-1.0em\sum \mathrm{d}^k z\, 2 q \partial_{a_z} q  =   \int\kern-1.0em\sum \mathrm{d}^k z\, 2 qe^{-i\xi} \partial_{a_z} (q e^{i\xi}) = 2    \int\kern-1.0em\sum \mathrm{d}^k z\,   \langle \psi(\bar z,0) \vert \hat U_{a_z}^\dag\vert \boldsymbol z \rangle \langle\boldsymbol z \vert \partial_{a_z}\hat U_{a_z}   \vert  \psi(\bar z, 0)\rangle = \langle \psi(\bar z,0) \vert \hat U_{a_z}^\dag \partial_{a_z}\hat U_{a_z}   \vert  \psi(\bar z, 0)\rangle$, where we used the compleness relation and $\partial_{a_z}\xi =0$. With this, our proof is complete.    
\end{proof}

\subsection{Determining if a state satisfies the variance-optimal measurement condition}\label{app:VarianceOptimalMeasurementForHGState}
So, let us now check when the condition $ \partial_{a_z}\arg [ \langle \boldsymbol z \vert \hat U_{a_z}\vert \psi(\bar z,0)\rangle ] =0$ is satisfied for some arbitrary multi-mode squeezed vacuum state. Let us write the state in the Schmidt mode basis $\hat c_n^\dag = \int\mathrm{d}z\, \Psi_n (z) \hat D^\dag (z)$ in which the state appears in product form: 
\begin{align}
  \vert \psi(\bar z,\bar p)\rangle &=  \bigotimes_{n} \exp \left( +\frac{r_n e^{i\phi_n}}{2} \hat c_n^{\dag2} -\frac{r_n e^{-i\phi_n}}{2} \hat c_n^{2}\right)\vert 0\rangle \\
  &=\bigotimes_{n} \exp \left( +\frac{\tanh (r_n)e^{i\phi_n}}{2} \hat c_n^{\dag2} \right)\vert 0\rangle  \\
   &=\bigotimes_{n} \exp \left( +\frac{\tanh (r_n)e^{i\phi_n}}{2}  \int \mathrm{d}z \, \Psi_n (z)\hat D^\dag (z) \int \mathrm{d}\tilde z \, \Psi_n (\tilde z)\hat D^\dag (\tilde  z) \right)\vert 0\rangle  \\
   &= \exp \left( \int \mathrm{d}z\int \mathrm{d}\tilde z \left[\sum_{n}\frac{\tanh (r_n)e^{i\phi_n}}{2} \Psi_n (z)\Psi_n (\tilde z)\right]   \hat D^\dag (z)   \hat D^\dag (\tilde z) \right)\vert 0\rangle \\
   &= \exp \left( \int \mathrm{d}z\int \mathrm{d}\tilde z  \, g(z,\tilde z)   \hat D^\dag (z)   \hat D^\dag (\tilde z) \right)\vert 0\rangle
\end{align}
Now, let us gain some insight and  first look at a simple two-photon example:
\begin{align}
    \langle (z,\tilde z)^T\vert \psi(\bar z,\bar p)\rangle &= \frac{1}{\sqrt{2!}}\langle 0 \vert \hat D(z) \hat D(\tilde z) \vert \psi(\bar z,\bar p)\rangle \\
    &= \frac{1}{\sqrt{2!}}\langle 0 \vert \hat D(z) \hat D(\tilde z) \int \mathrm{d}z'\int \mathrm{d}\tilde z'  \, g(z',\tilde z')   \hat D^\dag (z')   \hat D^\dag (z') \vert 0\rangle \\
    &= \frac{1}{\sqrt{2!}} \int \mathrm{d}z'\int \mathrm{d}\tilde z'  \, g(z',\tilde z')   \langle 0 \vert \hat D(z) \hat D(\tilde z)\hat D^\dag (z')   \hat D^\dag (z') \vert 0\rangle \\
      &= \frac{1}{\sqrt{2!}} \int \mathrm{d}z'\int \mathrm{d}\tilde z'  \, g(z',\tilde z')    \left[ \delta (z-z')\delta (\tilde z-\tilde z')+  \delta (z-\tilde z')\delta (\tilde z-  z')\right] \\
      &= \frac{1}{\sqrt{2!}} \left[g(z,\tilde z)+g(\tilde z,  z)\right] .
\end{align}
For the parameter imprinted state we have
\begin{align}
    \langle (z,\tilde z)^T  \vert\hat U_{a_z}\vert \psi(\bar z,\bar p)\rangle &= \frac{1}{\sqrt{2!}} \left[g(z+a_z,\tilde z+a_z)+g(\tilde z+a_z,  z+a_z)\right] .
\end{align}

Now, for an arbitrary $z$-resolving photon counting measurement, we find
\begin{align}
    \langle \boldsymbol z\vert \psi(\bar z,\bar p)\rangle &=\frac{1}{\sqrt{(2K)!}}\langle 0 \vert [\otimes_{i=1}^{2K}\hat D(z_i)]    \vert \psi(\bar z,\bar p)\rangle \\
    &= \frac{1}{\sqrt{(2K)!}} \frac{1}{K!} \int \mathrm{d}z_1'\ldots \int \mathrm{d}z_{2K}' \,\,   g(z_1',z_2') \cdots g(z_{2K-1}',z_{2K}') \langle 0 \vert [\otimes_{i=1}^{2K}\hat D(z_i) ]  [\otimes_{i=1}^{2K}\hat D^\dag(z_i') ] \vert 0\rangle \\
    &= \frac{1}{\sqrt{(2K)!}} \frac{1}{K!}  \sum_{\pi\in S_{2K}} \prod_{b=1}^{K}  g(z_{\pi(2b-1)}, z_{\pi(2b)})  .
\end{align}
Here, $\sum_{\pi\in S_{2K}}  $ means to sum over all permutations $\pi$ of the indices $\{1,\ldots,2K \}$.

Let us now consider various states and see if they satisfy the condition for a variance-optimal measurement.

\subsubsection{Phase-insensitive resolved photon counting measurement for product of single-mode squeezed vacuua in the HG mode basis}
Let us consider 
\begin{align}
    \vert \psi (\bar z,\bar p)\rangle &=  \bigotimes_{n} \exp \left( +\frac{r_n e^{i\phi_n}}{2} \hat c_n^{\dag2} -\frac{r_n e^{-i\phi_n}}{2} \hat c_n^{2}\right)\vert 0\rangle
\end{align}
where 
\begin{align}
    \hat c_n^\dag = \int\mathrm{d}z\, \Phi_n (z)  \hat D^\dag (z)
\end{align}
is the HG mode basis with
\begin{align}
    \Phi_n (z) = \Phi (z; z_0, p_0,\sigma,\theta)
\end{align}
It is straightforward to show that
\begin{align}
    \bar z &= z_0 \\
    \bar p &= p_0
\end{align}
regardless of how the squeezing $r_n$ is distributed among the modes. 

Now, it follows that the state $\hat U_{a_z}\vert \psi (\bar z, 0) \rangle =  \bigotimes_{n} \exp \left( +\frac{r_n e^{i\phi_n}}{2} \hat c_n'^{\dag2} -\frac{r_n e^{-i\phi_n}}{2} \hat c_n'^{2}\right)\vert 0\rangle$ has the same form, but with the mode operators
\begin{align}
    \hat c_n'^\dag = \int \mathrm{d}z \,  \Phi_n (z+a_z; z_0, 0,\sigma,\theta_n) \hat D^\dag(z)
\end{align}
Now, let us consider:
\begin{align}
    \langle \boldsymbol z\vert \hat U_{a_z} \vert \psi (\bar z, 0)\rangle  = \frac{1}{\sqrt{(2K)!}} \frac{1}{K!}  \sum_{\pi\in S_{2K}} \prod_{b=1}^{K}  g(z_{\pi(2b-1)}+a_z, z_{\pi(2b)}+a_z)  .
\end{align}
We then find
\begin{align}
    g(z +a_z,  \tilde z+a_z) =   \sum_{n}\frac{\tanh (r_n)}{2} \Phi_n (z+a_z; z_0, 0,\sigma,\theta_n) (z) \Phi_n (\bar z+a_z; z_0, 0,\sigma,\theta_n)
\end{align}
Now, for $\theta =0$, $g(z +a_z,  \tilde z+a_z)$ is real, and thus, $\langle \boldsymbol z\vert \hat U_{a_z} \vert \psi (\bar z, 0)\rangle$ is real valued as well. So the condition $ \partial_{a_z}\arg [ \langle \boldsymbol z \vert \hat U_{a_z}\vert \psi(\bar z,0)\rangle ] =0$ is satisfied.

Similarly, $\theta_n =\theta$ for all $n$, we have $\arg g(z+a_z,\tilde z+a_z) = -2\theta$ from which follows $\arg [\langle \boldsymbol z\vert \hat U_{a_z} \vert \psi (\bar z, 0)\rangle] = - 2K\theta$, and the condition is also satisfied. In addition, one can show that the condition can be satisfied if   $\theta_n$ differ from one another by integer multiples of $\pi$.

 For the general case $\theta_n$, the condition is no longer generally satisfied.

\subsubsection{Phase-insensitive measurement for the optimal state}\label{app:VarianceOptimalMeasurementForVarianceOptimalState}

Let us now consider our optimal state and examine under what conditions phase-insensitive measurements can be variance optimal.
Again, the state can be written as
\begin{align}
      \vert \psi (\bar z, \bar p)\rangle &=  \bigotimes_{n=\{i,j\}} \exp \left( +\frac{r_n e^{i\phi_n}}{2} \hat c_n^{\dag2} -\frac{r_n e^{-i\phi_n}}{2} \hat c_n^{2}\right)\vert 0\rangle 
\end{align}
with 
\begin{align}
    \hat c_n^\dag = \int \mathrm{d}z \, \Psi_n (z) \hat D^\dag (z) .
\end{align}
Recall that the optimal states from Sec.~\ref{app:QFIregularized} has  Schmidt modes fo the form
\begin{align}
    \Psi_1 (z) &=    \Phi_{0} (z;  z_1,  p_1,\sigma_z,\theta_1) + O(\vert\mathcal{S}\vert) \\
     \Psi_2 (z) &= \Phi_{0} (z;  z_2,  p_2,\sigma_z,\theta_2) + O(\vert\mathcal{S}\vert) \\
     \Psi_3 (z) &=    \Phi_{1} (z;  z_1,  p_1,\sigma_z,\theta_1) + O(\vert\mathcal{S}\vert) \\
     \Psi_4 (z) &= \Phi_{1} (z;  z_2,  p_2,\sigma_z,\theta_2) + O(\vert\mathcal{S}\vert) .
\end{align}
Now assume that like in Sec.~\ref{app:QFIregularized} that
\begin{align}
    p_1 &= p_0 + \delta_{+,p} \\
    p_2 &= p_0 - \delta_{-,p} ,
\end{align}
and our choice of $r_n$ shall ensure that 
\begin{align}
    \bar p = p_0 .
\end{align}
How to achieve this was discussed in Sections~\ref{app:FullyOptimalRegularizedQFI} and \ref{app:VarianceOptimalRegularizedQFI}.

Now again, the state  $\hat U_{a_z}\vert \psi (\bar z, 0) \rangle =  \bigotimes_{n} \exp \left( +\frac{r_n e^{i\phi_n}}{2} \hat c_n'^{\dag2} -\frac{r_n e^{-i\phi_n}}{2} \hat c_n'^{2}\right)\vert 0\rangle$ has the same form, but with the mode operators
\begin{align}
    \hat c_n'^{\dag} = \int \mathrm{d}z \, \Psi_n (z+a_z) e^{+ i \bar p (z+a_z) } \hat D^\dag (z)  = \int \mathrm{d}z \, \Psi_n (z+a_z) e^{+ i p_0 (z+a_z) } \hat D^\dag (z)
\end{align}
and we have 
\begin{align}
    \Psi_1 (z+a_z) e^{+ i p_0 (z+a_z) } &= \Phi_0 (z;z_1, +\delta_{+,p},\sigma_z,\theta_1) \\
     \Psi_2 (z+a_z)e^{+ i p_0 (z+a_z) } &= \Phi_0 (z;z_2,-\delta_{-,p},\sigma_z,\theta_2) \\
      \Psi_3 (z+a_z) e^{+ i p_0 (z+a_z) } &= \Phi_1 (z;z_1, +\delta_{+,p},\sigma_z,\theta_1) \\
     \Psi_4 (z+a_z)e^{+ i p_0 (z+a_z) } &= \Phi_1 (z;z_2,-\delta_{-,p},\sigma_z,\theta_2) .
\end{align}
Now, let us consider:
\begin{align}
    \langle \boldsymbol z\vert \hat U_{a_z} \vert \psi (\bar z, 0)\rangle  = \frac{1}{\sqrt{(2K)!}} \frac{1}{K!}  \sum_{\pi\in S_{2K}} \prod_{b=1}^{K}  g(z_{\pi(2b-1)}+a_z, z_{\pi(2b)}+a_z) 
\end{align}
Now, let us assume $z_1=z_2=z_0$ (as we will see, this choice ensures that certain terms can be factored out): 
\begin{align}
    g(z +a_z,  \tilde z+a_z) &=   + \frac{\tanh (r_1)}{2} \Psi_1 (z+a_z) e^{+ i p_0 (z+a_z) } \Psi_1 (\tilde z+a_z) e^{+ i p_0 (\tilde z+a_z) } \\
     &\quad \,\,+ \frac{\tanh (r_2)}{2} \Psi_2 (z+a_z) e^{+ i p_0 (z+a_z) } \Psi_2 (\tilde z+a_z) e^{+ i p_0 (\tilde z+a_z) } \\
     &=  \left( \frac{1}{2\sigma_z^2 \pi}\right)^{2/4} \sqrt{\frac{2^{-n}}{n!}}^2   e^{-\frac{1}{2}\left( \frac{z+a_z-z_0}{\sqrt{2}\sigma_z}\right)^2}e^{-\frac{1}{2}\left( \frac{\tilde z +a_z-z_0}{\sqrt{2}\sigma_z}\right)^2}  \\
     & \times \Big[ t_1  e^{-i(p_0+\delta_{+,p})(z+\tilde z+2a_z-2z_0)} e^{-i2\theta_1} \cdot e^{+ i p_0 (z+ \tilde z+ 2a_z)} + t_2  e^{-i(p_0-\delta_{-,p})(z+\tilde z +2a_z-2z_0)} e^{-i2\theta_2} \cdot e^{+ i p_0 (z+ \tilde z+ 2a_z)} \Big] \\
     &=  \left( \frac{1}{2\sigma_z^2 \pi}\right)^{2/4} \sqrt{\frac{2^{-n}}{n!}}^2   e^{-\frac{1}{2}\left( \frac{z+a_z-z_0}{\sqrt{2}\sigma_z}\right)^2}e^{-\frac{1}{2}\left( \frac{\tilde z +a_z-z_0}{\sqrt{2}\sigma_z}\right)^2}  \\
     & \times e^{+2ip_0 z_0}\Big[ t_1  e^{- i \delta_{+,p} (z+\tilde z+2a_z - 2z_0)} e^{-i 2\theta_1} + t_2   e^{+i\delta_{-,p} (z+\tilde z+2a_z-2z_0)} e^{-i 2\theta_2} \Big]
\end{align}
where $t_n = \tanh (r_n)$.

Now, if we pick $\delta_{+,p} = \delta_{-,p} =\delta$, $r_1 =  r_2$, and also $\theta_1= \theta_2 = \theta$ , it follows that
\begin{align}
       g(z +a_z,  \tilde z+a_z) &=   \left( \frac{1}{2\sigma_z^2 \pi}\right)^{2/4} \sqrt{\frac{2^{-n}}{n!}}^2   e^{-\frac{1}{2}\left( \frac{z+a_z-z_0}{\sqrt{2}\sigma_z}\right)^2}e^{-\frac{1}{2}\left( \frac{\tilde z +a_z-z_0}{\sqrt{2}\sigma_z}\right)^2} \tanh (r)  \\
     & \times e^{+i2p_0 z_0} e^{-i2\theta}\Big[    e^{+ i \delta (z+\tilde z+2a_z - 2z_0)}   +     e^{-i\delta  (z+\tilde z+2a_z-2z_0)}   \Big] \\
      &=   \left( \frac{1}{2\sigma_z^2 \pi}\right)^{2/4} \sqrt{\frac{2^{-n}}{n!}}^2   e^{-\frac{1}{2}\left( \frac{z+a_z-z_0}{\sqrt{2}\sigma_z}\right)^2}e^{-\frac{1}{2}\left( \frac{\tilde z +a_z-z_0}{\sqrt{2}\sigma_z}\right)^2} \tanh (r)  \\
     & \times e^{+i2p_0 z_0} e^{-i2\theta}\cdot  2 \cos ( \delta (z+\tilde z+ 2a_z- 2z_0)) .
\end{align}

We thus find that $\arg[g(z+a_z, \tilde{z}+a_z)] = 2(p_0 z_0 - \theta)$, and therefore $\arg[\langle\boldsymbol{z}\vert\hat U_{a_z}\vert\Psi(\bar{z},0)\rangle] = 2K(p_0 z_0 - \theta)$, from which it follows that $\partial_{a_z}\arg[\langle\boldsymbol{z}\vert\hat U_{a_z}\vert\psi(\bar{z},0)\rangle] = 0$. Note that the choice $\theta_1 = \theta$ and $\theta_2 = \theta + \pi$ leads to an analogous result, with a $\sin$ dependence replacing the $\cos$.

The conditions for a variance-optimal measurement thus require equal squeezing $r_1 = r_2$, together with $\delta_- = \delta_+$ and $z_1 = z_2$. For $z_1 \neq z_2$, one generally has $\partial_{a_z}\arg[\langle\boldsymbol{z}\vert\hat U_{a_z}\vert\psi(\bar{z},0)\rangle] \neq 0$, and the variance-optimality condition is no longer satisfied.

\end{document}